%% file: arXiv.tex
\documentclass{article} 
\usepackage{iclr2027_conference,times}
\input{math_commands.tex}

\usepackage{hyperref}
\usepackage{url}
\usepackage{amssymb}
\usepackage{amsmath}
\usepackage{amsthm}
\usepackage{caption}
\usepackage{subcaption} 
\usepackage{booktabs}
\usepackage{adjustbox}
\usepackage{multirow} 
\usepackage{enumitem}
\usepackage{kotex}
\usepackage{braket}
\usepackage{booktabs}
\usepackage{makecell}
\usepackage[section]{placeins}

\usepackage{cuted}
\usepackage{apxproof}
\definecolor{lightgray}{cmyk}{0,0,0,.05}

\newlength\GreyboxOuterVspace
\newlength\GreyboxPadding
\newlength\GreyboxRule
\newcommand{\GreyboxFrameColor}{black}

\newcommand{\greybox}[1]{%
  \par\addvspace{\GreyboxOuterVspace}%
  \noindent\begingroup
  \setlength{\fboxsep}{\GreyboxPadding}%
  \setlength{\fboxrule}{\GreyboxRule}%
  \fcolorbox{\GreyboxFrameColor}{lightgray}{%
    \parbox{\dimexpr\linewidth-2\fboxsep-2\fboxrule\relax}{#1}
    
  }
  \endgroup%
  \par\addvspace{\GreyboxOuterVspace}%
}

\newtheorem{proposition}{Proposition}
\newtheorem{definition}{Definition}

\newtheorem{lemma}{Lemma}

\usepackage{graphicx,wrapfig,lipsum, tikz}
\usepackage{algorithm}
\usepackage{algpseudocode}

\newtheorem{corollary}{Corollary}

\hypersetup{
  colorlinks,
  linkcolor={red},
  citecolor={blue!50!black},
  urlcolor={blue!80!black},
}

\definecolor{ForestGreen}{RGB}{34,200, 1}
\definecolor{salmon}{RGB}{250,128,114}
\definecolor{OliveGreen}{RGB}{0,127,0}
\PassOptionsToPackage{dvipsnames,svgnames,table}{xcolor}

\title{Fourier-Geometric Circuit Design for \\  Gate and Entanglement Placement in \\  Quantum Neural Networks}

\author{
Seungcheol Oh\textsuperscript{1},
Chaemoon Im\textsuperscript{1},
Daeyeun Kim\textsuperscript{1},
Soohyun Park\textsuperscript{2},
Vaneet Aggarwal\textsuperscript{3}, \\
\textbf{
Mohsen Heidari\textsuperscript{4},
Joongheon Kim\textsuperscript{1,5,$\ast$}
}
\\[0.5em]
\normalfont
\textsuperscript{1}Korea University \\
\textsuperscript{2}Sookmyung Women's University \\
\textsuperscript{3}Purdue University \\
\textsuperscript{4}Indiana University\\
\textsuperscript{5}Seoul National University Hospital
}

\iclrfinalcopy 

\begin{document}
\maketitle

\begingroup
\renewcommand{\thefootnote}{*}
\footnotetext{
Corresponding author: Joongheon Kim (\texttt{joongheon@korea.ac.kr})
}
\endgroup

\lhead{}
\renewcommand{\headrulewidth}{0pt}

\begin{abstract}
The output of a parameterized quantum circuit (PQC) can be expressed as a finite Fourier series whose accessible frequencies are fixed by the data-encoding gates. While the encoder determines which frequencies can appear, the corresponding Fourier coefficients depend on how the trainable and entangling gates are arranged. Although existing studies provide metrics for characterizing how gate structure affects Fourier coefficients, they do not translate these analyses into an explicit design criterion specifying how gates should be arranged to make a target coefficient reachable. In this paper, we provide such a criterion. Using the adjoint action of the encoding generator, we decompose operator space into two-dimensional invariant planes indexed by frequency and show that each Fourier coefficient is exactly a sum of bilinear projections of the effective state and observable onto the planes at that frequency. Because, for Pauli encodings, high-frequency planes are spanned by mixed multi-qubit Pauli strings, a target coefficient can contribute to the output only when local rotations and entangling layers are arranged so that both the effective state and observable acquire support on one of its planes. For Pauli readouts and commuting two-qubit entanglers, this yields a circuit design rule that specifies, for a given interaction graph, a placement of local rotations and entangling layers that makes a target Fourier coefficient reachable. Using a single encoding layer, we validate the proposed design through a placement ablation, regression tasks from PDEBench and physics-informed Maxwell field modeling, and further demonstrate its robustness to moderate simulated gate noise and damping.
\end{abstract}

\section{Introduction}
Parameterized quantum circuits (PQCs) have emerged as a standard framework for near-term quantum machine learning (QML) due to their compatibility with shallow, low-qubit implementations \citep{benedetti2019parameterized}. However, their design, particularly the arrangement of encoding, local and entangling gates, remains closer to an art than a mathematically grounded discipline because circuits are assembled from a small catalog of templates, chosen with little reference to the function the PQC is meant to represent. This leaves a basic question unanswered. When a target needs particular frequencies, as the oscillatory fields of scientific machine learning do, where should the local and entangling gates be placed?

The closest answer comes from studies of the representational capacity of PQCs. \citet{Schuld} showed that the output of a PQC is a finite Fourier series whose accessible frequencies are fixed by the data-encoding gates, and repeating the encoding, either across multiple qubits or through data re-uploading \citep{perez2020data}, enlarges that set in a way that follows directly from the encoder. Therefore, the accessible frequency set is fixed by the encoding gates before training. What is not settled is whether the coefficient at an available frequency is actually nonzero. Coefficients are known to depend on the trainable blocks and the observable, and a growing body of work has since computed them, characterized their statistics, and traced their consequences for generalization and trainability \citep{casas2023multidimensional,barthe2024gradients,peters2023generalization,wiedmann2025fourier,campbell2026harmonic,ragone2024lie}. However, \emph{these analyses begin from an already chosen circuit.} They tell us which coefficients a particular arrangement of gates produces, but not how that arrangement should be chosen. No existing result specifies how to place local and entangling gates so that a desired Fourier coefficient becomes reachable.

An explicit placement rule is hard to obtain by following the state through the circuit, where a Fourier coefficient appears only when the state meets the observable at the end, so it belongs to the whole circuit rather than to any single gate. Seen through operators, it can be traced. We show that under the adjoint action of the encoding generator, operator space separates into two-dimensional planes, each carrying exactly one frequency, with the input rotating each plane at its own rate. This lets us write the coefficient at a frequency exactly, as the inner product of the effective state and observable with the corresponding planes. We find that high-frequency planes involve multi-qubit operators, so a local measurement cannot reach them unless entangling gates first spread it across qubits. Whether that happens depends on where those gates sit, and the same gates in a different order can leave a Fourier component present in or absent from the output. We exploit this to derive a placement rule for Pauli encodings and readouts with commuting two-qubit entanglers. The rule makes the frequencies a target needs reachable, which is necessary for learning them but does not guarantee it, so we test that link empirically.

We highlight the contributions as follows:
\\
\textbf{(i) An exact expression for each Fourier coefficient.} Using the adjoint action of the encoding generator, we express each coefficient as a sum of bilinear projections of the effective state and observable onto the two-dimensional invariant planes associated with that frequency.
\\
\textbf{(ii) A placement rule for trainable gates and entanglement.} For Pauli encodings and readouts, commuting two-qubit entanglers and a given interaction graph, we give a placement of local rotations and entangling gates that suffices for a target Fourier coefficient to be nonzero at almost every parameter setting, identify which of its conditions are also necessary, and show that each such frequency can reach unit amplitude on its own and that reordering the same gates can suppress it. The rule builds the circuit in linear time.
\\
\textbf{(iii) Controlled validation.} We evaluate the resulting design through parameter-matched comparisons with circuit templates over multiple seeds, a placement ablation that moves the same gates, supervised PDE regression and physics-informed Maxwell modeling, together with damping, gate-noise and finite-shot noise studies.

\section{Prior Work}
\noindent \textbf{Analyses of a Given Circuit.}
\citet{Schuld} established the finite Fourier representation and its encoder-determined frequency set, \citet{perez2020data} showed that data re-uploading enlarges this set. \citet{casas2023multidimensional} move closer to circuit design by showing that parallel layouts can represent multidimensional series of higher degree than line layouts. Their trainable blocks, however, are treated as generic unitaries. Later works similarly analyze predefined circuits by computing exact spectra and structurally vanishing coefficients \citep{wiedmann2025fourier,nemkov2023fourier}, comparing coefficient statistics across circuit families \citep{campbell2026harmonic,mhiri2025constrained}, and studying the frequency profiles and generalization of re-uploading models \citep{barthe2024gradients,peters2023generalization}. These works characterize what a given circuit can represent, but do not specify which gate, acting on which qubit and at which position, determines a particular Fourier coefficient. We derive this relationship and use it to construct trainable blocks whose reachable coefficients are determined by design.
\\
\noindent \textbf{Lie-Algebraic and Symmetry-Based Design.}
The dynamical Lie algebra (DLA) generated by a circuit characterizes its expressivity and trainability, but the circuit is represented through its generating algebra rather than the placement of individual gates. \citet{ragone2024lie} explicitly leave the arrangement of gates within a layer to future work \citep{larocca2022diagnosing,fontana2024characterizing}. Symmetry-aware constructions prescribe circuit structure from algebraic constraints \citep{larocca2022group,meyer2023exploiting,nguyen2024theory}, but likewise determine suitable gate sets without specifying their ordering. Consequently, circuits with the same encoding, interaction graph, and gate family may have the same accessible frequencies and generating algebra while differing in whether a spectral component reaches the measured observable. Our work addresses this finer dependence with a placement of local and entangling gates that makes a target Fourier coefficient reachable.
\\
\noindent \textbf{Circuit Templates and Applications.}
PQCs are commonly built from heuristic circuit templates \citep{sim2019expressibility,kandala2017hardware}, including those used for differential equations \citep{kyriienko2021solving,trahan2024quantum,sedykh2024hybrid,berger2025trainable,panichi2026quantum,farea2025qcpinn,xiao2024physics,chen2025maxwell}, where oscillatory solutions make spectral representation particularly relevant. Increasing entanglement alone does not necessarily improve performance, and \citet{trahan2024quantum} found strongly entangled layers suboptimal for their model.
This motivates studying not only how much entanglement is used, but where the gates are placed so that the required spectral components can contribute to the output.

\section{Fourier-Geometric Circuit Design}
\noindent\textbf{Background.} 
A PQC implements a unitary transformation $U(\mathbf{x},\boldsymbol{\Theta})$ acting on an initial state $\ket{\psi}^{\otimes N}$, where $\mathbf{x}\in\mathbb{R}^{d}$, $\boldsymbol{\Theta}=(\boldsymbol{\theta}^{(L)},\boldsymbol{\theta}^{(R)})\in\mathbb{R}^{p}$, and $N$ are the input data, trainable parameters, and number of qubits, respectively. We consider an architecture consisting of a data-encoding gate $S(x)$ and trainable blocks $W_{L},W_{R}$, each composed of parameterized and entangling gates, such that $U(\mathbf{x},\boldsymbol{\Theta})=W_{L}(\boldsymbol{\theta}^{(L)})\,S(x)\, W_{R}(\boldsymbol{\theta}^{(R)})$. With the density matrix $\rho_0=(|\psi\rangle\langle\psi|)^{\otimes N}$ and an observable $M$, the model output is $f(x,\Theta)=\mathrm{Tr}(\rho_0U^{\dagger}MU)$.
Both $\rho_{0}$ and $M$ belong to the Hilbert--Schmidt operator space $\mathcal{B}(\mathcal{H})$ with inner product $\langle A,B\rangle=2^{-N}\operatorname{Tr}(A^{\dagger}B)$ \citep{mielnik1968geometry}, which is spanned by the Pauli strings $P_k\in\{I,X,Y,Z\}^{\otimes N}$. The parameterized gates $R_{P_k}(\theta)=e^{-i\theta P_k/2}$ act on operators through the adjoint map $\mathrm{ad}_{P_k}(A)=[P_k,A]$, since $U^{\dagger}MU=e^{\operatorname{ad}_{iA}}(M)$ for $U=e^{-iA}$, so gates rotate the Pauli-string components of the state and observable within $\mathcal{B}(\mathcal{H})$.
For an encoding gate $S(x)=e^{-ixG/2}$, the adjoint action of the encoding generator $G$ decomposes the operator space into invariant subspaces associated with differences between its eigenvalues. Consequently, the PQC output admits a finite Fourier representation whose accessible frequencies are determined by the spectrum of the encoding generator \citep{Schuld}. The trainable blocks $W_L$ and $W_R$, in contrast, determine how the state and observable project onto these frequency subspaces, and hence the Fourier coefficients. However, \emph{an explicit circuit-level relationship between the placement of trainable and entangling gates and the reachability of a target Fourier coefficient remains unclear.} We establish this relationship and use it to derive a circuit design rule for making target coefficients reachable. The decomposition underlying it holds for any encoding generator, whereas the design rule assumes a specific encoder, readout and entangler, which we state when we derive it.

\noindent \textbf{Closed Invariant Subspace.}
For any encoding generator $G$, conjugation by $S(x)=e^{-\frac{ix}{2}G}$ expands into nested commutators with $G$, which produce the trigonometric terms of the Fourier series, as the following lemma makes precise.
\begin{lemma}[\textbf{Hadamard's Lemma}]
\label{lem:Hadamard's Lemma1}
Based on the Baker--Campbell--Hausdorff formula \citep{chevalley2018theory}, for a generator $G$, an operator $\mathcal{O}\in\mathcal{B}(\mathcal{H})$ and a scalar $t$,
$e^{tG} \mathcal{O} e^{-tG} \;=\; \sum_{k=0}^{\infty}\nolimits \frac{t^{k}}{k!}\,\mathrm{ad}_G^{k}(\mathcal{O}).$
\end{lemma}
While Lemma \ref{lem:Hadamard's Lemma1} gives an infinite series, the series sums to trigonometric functions whenever two operators $\mathcal{X}_{\omega},\mathcal{Y}_{\omega}\in\mathcal{B}(\mathcal{H})$ span a subspace closed under the adjoint map of $G$\footnote{Detailed proof of Lemma \ref{lem:Hadamard's Lemma1} is provided in Appendix \ref{appendix:Hadamard_Lemma}.}, as the following proposition states.
\greybox{
\begin{proposition}[\textbf{Closed Invariant Subspace}]
\label{prop:2Dclosure}
Let $G$ be a generator for encoder $S(x)$ and $\mathcal{X_{\omega},Y_{\omega}}\in\mathcal{B}(\mathcal{H})$ be operators such that $\mathrm{ad}_G(\cdot)=[G,\cdot]$ closes on $\mathrm{span}\{\mathcal{X_{\omega},Y_{\omega}}\}$ in the sense that $\exists$  $\omega\in\mathbb{R}$ such that
$[G,\mathcal{X}_{\omega}]=-2i\omega\,\mathcal{Y}_{\omega},$ and $[G,\mathcal{Y}_{\omega}]=2i\omega\,\mathcal{X}_{\omega}.$
Then, $\forall x\in\mathbb{R}$,
\begin{align}
\label{eq:Conjugations}
S(x)^\dagger \mathcal{X}_{\omega} S(x) = \cos(\omega x)\mathcal{X}_{\omega}+\sin(\omega x)\mathcal{Y}_{\omega}, \quad
S(x)^\dagger \mathcal{Y}_{\omega} S(x) = \cos(\omega x)\mathcal{Y}_{\omega}-\sin(\omega x)\mathcal{X}_{\omega}.
\end{align}
\end{proposition}
}
Proposition~\ref{prop:2Dclosure} solves the evolution of a closed pair exactly as a rotation by the angle $\omega x$ within its plane,\footnote{Detailed proof of Proposition \ref{prop:2Dclosure} is provided in the Appendix \ref{appendix:closed_invariant}.} so $\omega$ is a \emph{frequency} of the model output in the strict trigonometric sense, fixed by the algebraic structure of the generator.
We note that Proposition~\ref{prop:2Dclosure} holds for any $G$. We now describe \emph{every} closed pair for single-qubit Pauli encodings, which single-qubit Clifford gates map to $G=\sum_{n=1}^{N}\hat{X}_{n}$, where $\hat{(\cdot)}_n\triangleq I^{\otimes (n-1)}\otimes(\cdot)\otimes I^{\otimes (N-n)}$ acts on qubit $n$.

\greybox{
\begin{proposition}[Complete Harmonic Decomposition]
\label{prop:planes}
Let $c=(S_{+},S_{-},T)$ be pairwise disjoint subsets of $\{1,\dots,N\}$ and set
\begin{equation}
E_{c}=\prod_{j\in S_{+}}E^{+}_{j}\prod_{j\in S_{-}}E^{-}_{j}
\prod_{j\in T}\hat{X}_{j},
\qquad
\omega(c)=|S_{+}|-|S_{-}|,
\qquad
m(c)=|S_{+}|+|S_{-}|.
\end{equation}
Then $\mathrm{ad}_{G}(E_{c})=-2\omega(c)E_{c}$, and the $4^{N}$ operators
$\{E_{c}\}$ form an orthogonal basis of $\mathcal{B}(\mathcal{H})$. For
$\omega(c)>0$ the Hermitian operators
$\mathcal{X}_{c}=E_{c}+E_{c}^{\dagger}$ and
$\mathcal{Y}_{c}=\tfrac{1}{i}(E_{c}-E_{c}^{\dagger})$ satisfy
$[G,\mathcal{X}_{c}]=-2i\omega\,\mathcal{Y}_{c}$ and
$[G,\mathcal{Y}_{c}]=2i\omega\,\mathcal{X}_{c}$, with
$\|\mathcal{X}_{c}\|^{2}=\|\mathcal{Y}_{c}\|^{2}=2^{m(c)+1}$.
\end{proposition}
}

Hence $\mathcal{B}(\mathcal{H})=\bigoplus_{\omega=-N}^{N}\mathcal{B}_{\omega}$, so every observable decomposes over these planes and no frequency component is left unaccounted for. Since $\left(E^{\pm}_{j}\right)^{2}=0$, each qubit contributes at
most one factor, giving $|\omega|\le m(c)\le N$ and recovering the qubit-count bound on accessible frequencies. We write $d_{\omega}$ for the number of planes at frequency $\omega$ and index them by $r$; the count is given in Appendix~\ref{appendix:defineclosedinvariant}.
The planes with $S_{-}=T=\emptyset$ saturate $|\omega|=m(c)$ and are the ones a design must reach to activate the highest frequencies.
Higher frequencies therefore require mixed multi-qubit Pauli strings rather than single Pauli operators. Visual intuition is given in Fig.~\ref{fig:invariantsubspaceprojection}\footnote{Detailed proof of Proposition \ref{prop:planes} is provided in Appendix~\ref{appendix:invariant_planes}.}.

\textbf{Fourier Series Under Closed Invariant Subspaces.} 
\begin{figure}
    \centering
    \includegraphics[width=1.0\linewidth]{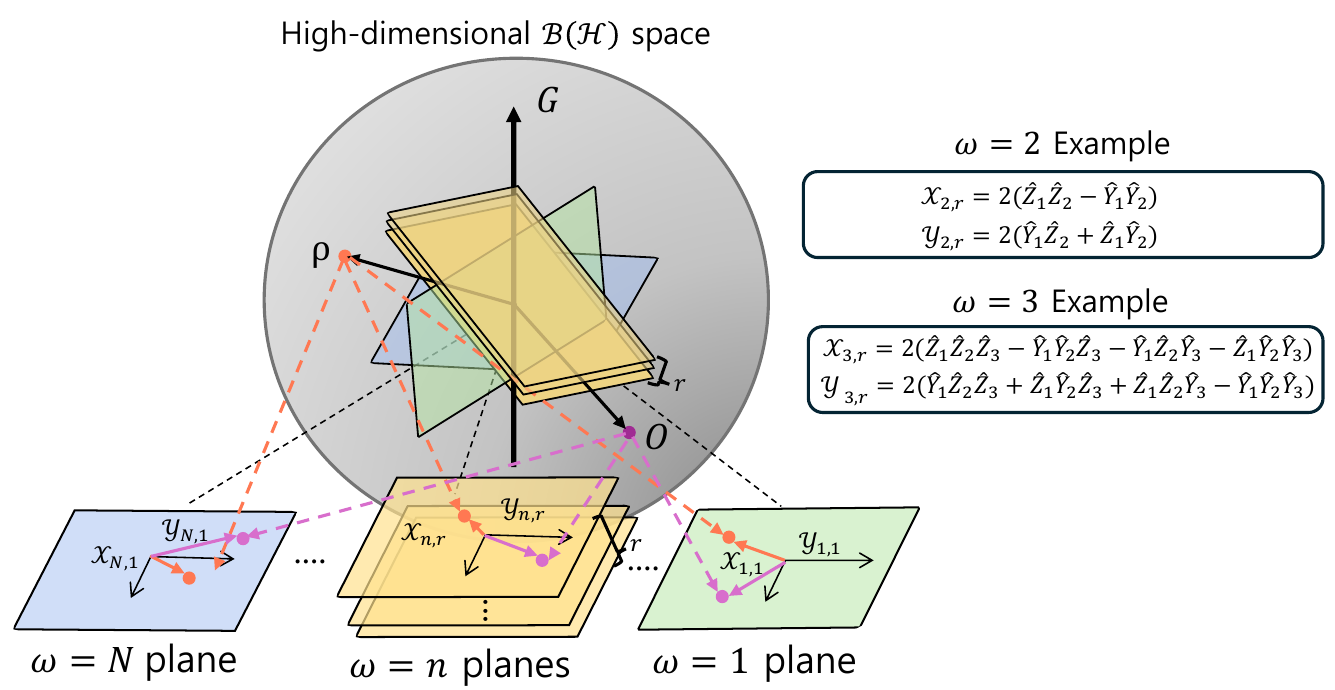}
    \caption{\textbf{Invariant subspace projection.} The operator space $\mathcal{B}(\mathcal{H})$ splits into invariant planes, each carrying one frequency, and the coefficient of that frequency is set by the Hilbert--Schmidt inner products of $\rho$ and $O$ with $\mathcal{X}_{\omega,r}$ and $\mathcal{Y}_{\omega,r}$. }
    \label{fig:invariantsubspaceprojection}
    \vspace{-6mm}
\end{figure}
We now show how the Fourier series arises from the closed pairs $\{\mathcal{X}_{\omega,r},\mathcal{Y}_{\omega,r}\}$.For
    $U(x)=W_LS(x)W_R$,
the identity $\langle\psi|B|\psi\rangle=\mathrm{Tr}(|\psi\rangle\langle\psi|B)$ gives
\begin{align}
    \label{eq:traceQuantumOutput}
    f(x)=\langle\psi|U(x)^\dagger MU(x)|\psi\rangle
    =\mathrm{Tr}(\underbrace{W_R\rho_0W_R^\dagger}_\rho S(x)^\dagger \underbrace{W_L^\dagger M W_L}_OS(x))=\mathrm{Tr}(\rho S(x)^\dagger O S(x)), 
\end{align}
where $\rho_0=|\psi\rangle\langle\psi|$ is the initial state, $M$ the measured observable and $O$ the effective observable. We decompose both $\rho$ and $O$ over all closed pairs $\{(\mathcal{X}_{\omega,r},\mathcal{Y}_{\omega,r})\}$.
By Proposition~\ref{prop:2Dclosure}, each pair $(\mathcal{X}_{\omega,r},\mathcal{Y}_{\omega,r})$ spans a two-dimensional subspace invariant under the adjoint action of $G$. Conjugation by $S(x)$ therefore rotates each plane by the angle $\omega x$, which yields the Fourier series of the following proposition.
\greybox{
\begin{proposition}[\textbf{Projection onto Invariant Subspaces}]
\label{proposition:multi-qubit_Projected}
Let $\{(\mathcal{X}_{\omega,r},\mathcal{Y}_{\omega,r})\}_{r=1}^{d_{\omega}}$ be the
invariant pairs, and let
$m_{\omega,r}=|S_{+}^{(r)}|+|S_{-}^{(r)}|$ denote the number of active qubits of
the $r$-th pair. These operators are mutually orthogonal with
$\langle\mathcal{X}_{\omega,r},\mathcal{X}_{\omega,r}\rangle
=\langle\mathcal{Y}_{\omega,r},\mathcal{Y}_{\omega,r}\rangle
=2^{\,m_{\omega,r}+1}$, and together with the commutant $\mathcal{B}_{0}$ they
span $\mathcal{B}(\mathcal{H})$. Consequently, for \emph{every} state $\rho$ and
\emph{every} observable $O$, the output decomposes into a constant term and contributions at frequencies $\omega\le N$:
\begin{equation}
f(x)=\mathrm{Tr}\bigl(\rho\,\Pi_{0}O\bigr)
+\sum_{\omega=1}^{N}\nolimits\sum_{r=1}^{d_{\omega}}\nolimits f_{\omega,r}(x),
\end{equation}
where $\Pi_{0}$ is the orthogonal projection onto $\mathcal{B}_{0}$ and
\begin{equation}
\label{eq:bilinear_form}
f_{\omega,r}(x)
=
2^{\,N-m_{\omega,r}-1}
\begin{bmatrix}
\langle \rho,\mathcal{X}_{\omega,r}\rangle\\
\langle \rho,\mathcal{Y}_{\omega,r}\rangle
\end{bmatrix}^{\top}
\begin{bmatrix}
\cos(\omega x) & -\sin(\omega x)\\
\sin(\omega x) & \cos(\omega x)
\end{bmatrix}
\begin{bmatrix}
\langle O,\mathcal{X}_{\omega,r}\rangle\\
\langle O,\mathcal{Y}_{\omega,r}\rangle
\end{bmatrix}.
\end{equation}
\end{proposition}
}
The constant term $\mathrm{Tr}(\rho\Pi_{0} O)$ is the $\omega=0$ component, the part of $O$ that commutes with $G$. With $\mathbf{a}_{\omega,r}$ and $\mathbf{b}_{\omega,r}$ the projection vectors of $\rho$ and $O$ in \eqref{eq:bilinear_form}, each component is a single sinusoid
\begin{equation}
\label{eq:amplitude}
f_{\omega,r}(x)=2^{\,N-m_{\omega,r}-1}\,\|\mathbf{a}_{\omega,r}\|\,\|\mathbf{b}_{\omega,r}\|\cos(\omega x+\varphi_{\omega,r}),
\end{equation}
where $\varphi_{\omega,r}$ is the angle between the two projections (Fig.~\ref{fig:invariantsubspaceprojection}). We call $\omega$ \emph{reachable} if its coefficient $c_{\omega}$ in $f(x)=\sum_{|\omega|\le N}c_{\omega}e^{i\omega x}$ is not identically zero over the trainable parameters. It is unreachable if $\mathbf{b}_{\omega,r}=0$ for every $r$ at every parameter setting. Planes at the same $\omega$ can cancel, so the converse needs proof, and Corollary~\ref{cor:activation} below shows that such cancellation is confined to a set of parameters of measure zero.

\noindent\textbf{Multi-Qubit Entanglement Design.} Proposition~\ref{proposition:multi-qubit_Projected} says that a frequency reaches the output only if both $\rho$ and $O$ project onto one of its planes, and Proposition~\ref{prop:planes} says that every plane at frequency $\omega$ involves at least $\omega$ active qubits, each carrying a $\hat{Y}$ or $\hat{Z}$ factor. A local readout $M=\hat{Z}_{n}$ has a single active qubit, and single-qubit rotations keep it on qubit $n$, so an effective observable built from local gates alone has zero projection onto every plane at $\omega>1$, however its parameters are trained. Only entangling gates can add active qubits, so the design question is how an entangler transforms the readout and where it should sit relative to the local rotations. We restrict the entangler to the couplings the hardware allows, described by an interaction graph $\mathcal{G}=(V,E)$ with $V=\{1,\dots,N\}$ and neighbour set $\mathcal{N}(n)=\{k:\{n,k\}\in E\}$, and use the Ising entangler \citep{briegel2001persistent}, which represents the entangler class above up to single-qubit Clifford gates, $CZ$ included. Its generator is diagonal and built from commuting $\hat{Z}_{a}\hat{Z}_{b}$ couplings, so only the couplings at qubit $n$ act on its Pauli operators, and their action has the closed form of Lemma~\ref{lem:ising}.

\begin{lemma}[Ising Entanglement]
\label{lem:ising}
For a qubit $n$, define
\begin{equation}
U_{\mathrm{ent}}=\exp\Bigl(-\tfrac{i}{2}\sum_{k\in\mathcal{N}(n)}\nolimits
\gamma_{nk}\hat{Z}_{n}\hat{Z}_{k}\Bigr)
=\exp\Bigl(-\tfrac{i}{2}\hat{Z}_{n}\Phi_{n}\Bigr),
\qquad
\Phi_{n}=\sum_{k\in\mathcal{N}(n)}\nolimits\gamma_{nk}\hat{Z}_{k}.
\end{equation}
Then $U_{\mathrm{ent}}^{\dagger}\hat{Z}_{n}U_{\mathrm{ent}}=\hat{Z}_{n}$,
$U_{\mathrm{ent}}^{\dagger}\hat{X}_{n}U_{\mathrm{ent}}
=\hat{X}_{n}\cos\Phi_{n}-\hat{Y}_{n}\sin\Phi_{n}$, and
$U_{\mathrm{ent}}^{\dagger}\hat{Y}_{n}U_{\mathrm{ent}}
=\hat{Y}_{n}\cos\Phi_{n}+\hat{X}_{n}\sin\Phi_{n}$.
\end{lemma}
\begin{figure}
    \centering
    \includegraphics[width=.92\linewidth]{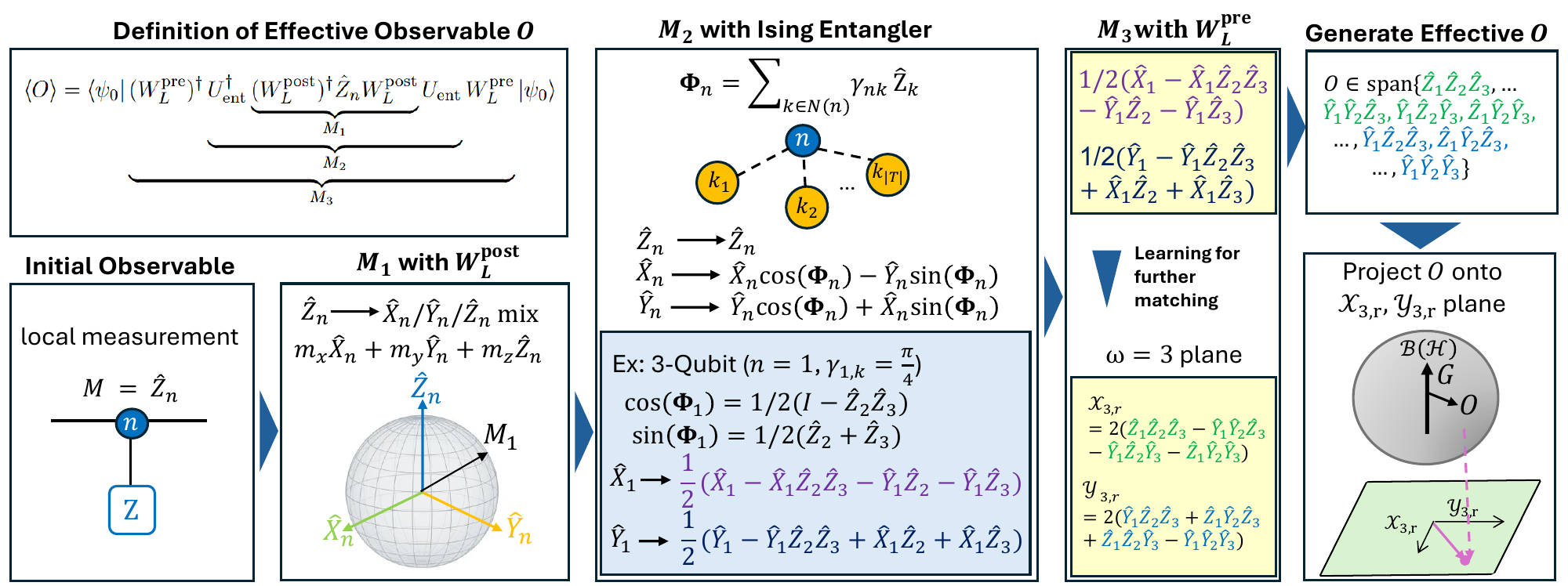}
    \caption{\textbf{Design flow for generating an effective observable $O$.} Local rotations and Ising entanglement transform the initial readout $Z_n$ into mixed Pauli strings, allowing $O$ to project onto the target $\omega=3$ invariant plane.}
    \label{fig:Design_flow}
    \vspace{-4mm}
\end{figure}
The three identities answer the design question. The first shows that the entangler leaves $\hat{Z}_{n}$ unchanged, so an entangler placed directly before the readout has no effect on $O$. The other two show that it rotates $\hat{X}_{n}$ and $\hat{Y}_{n}$ by the operator-valued angle $\Phi_{n}$, and since $\cos\Phi_{n}$ and $\sin\Phi_{n}$ expand into $\hat{Z}$-strings over $\mathcal{N}(n)$, every neighbour in these strings adds an active qubit. The entangler can therefore create support at higher frequencies, but only from the $\hat{X}_{n}$ and $\hat{Y}_{n}$ components, which the bare readout lacks. A local rotation must first create these components and hence sit between the readout and the entangler, while a second local rotation on the other side can reshape the strings the entangler produces. This gives the trainable block $W_{L}=W_{L}^{\mathrm{post}}U_{\mathrm{ent}}W_{L}^{\mathrm{pre}}$, and we use the same form $W_{R}=W_{R}^{\mathrm{post}}U_{\mathrm{ent}}W_{R}^{\mathrm{pre}}$ on the state side. Reading $O=W_{L}^{\dagger}\hat{Z}_{n}W_{L}$ from the inside out gives three stages, illustrated in Fig.~\ref{fig:Design_flow} for $\omega=3$ and derived for any $N$ and graph in Appendix~\ref{app:Nqubit}.
\begin{enumerate}[leftmargin=*,itemsep=2pt,topsep=3pt]
\item \textbf{Expose.} $W_{L}^{\mathrm{post}}$ tilts the readout into $M_{1}=m_{x}\hat{X}_{n}+m_{y}\hat{Y}_{n}+m_{z}\hat{Z}_{n}$, creating the $\hat{X}_{n}$ and $\hat{Y}_{n}$ components that the entangler acts on. Without this stage the next one is the identity.
\item \textbf{Spread.} $U_{\mathrm{ent}}$ multiplies these components by $\cos\Phi_{n}$ and $\sin\Phi_{n}$, attaching $\hat{Z}$-strings over $\mathcal{N}(n)$ and producing strings on $\{n\}\cup\mathcal{N}(n)$ with up to $\deg(n)+1$ active qubits. The graph $\mathcal{G}$ fixes which qubits can appear, so one round caps the reachable frequency at $\deg(n)+1$ for any $N$ and any graph, which gives $\omega=3$ for the star graph of Fig.~\ref{fig:Design_flow}. The string that reaches this ceiling has an overlap with its plane proportional to $\prod_{k}|\sin\gamma_{nk}|$ (Appendix~\ref{app:Nqubit}), so by \eqref{eq:amplitude} the strengths $\gamma_{nk}$ set the amplitude of the top frequency, which is largest for $CZ$ couplings, $|\gamma_{nk}|=\pi/2$.
\item \textbf{Align.} $W_{L}^{\mathrm{pre}}$ rotates the Pauli factors on this support without changing it, reweighting the strings to increase the projection of $O$ onto the target plane, as the matching stage of $W_{R}$ does for $\rho$.
\end{enumerate}
\begin{corollary}[Sufficient Condition for a Nonzero Coefficient]
\label{cor:activation}
Let the local blocks be general single-qubit rotations and $\rho_{0}$ a pure product state. If $\mathbf{b}_{\omega,r}\neq0$ for some plane at frequency $\omega\ge1$ at one parameter setting, then $c_{\omega}\neq0$ at almost every parameter setting. For $W_{L}=W_{L}^{\mathrm{post}}U_{\mathrm{ent}}W_{L}^{\mathrm{pre}}$ with $\sin\gamma_{nk}\neq0$, this holds for every $\omega\le\deg(n)+1$.
\end{corollary}
The three stages are therefore also sufficient, and they set the magnitude as well. With $CZ$ couplings and the same entangler in $W_{R}$, the circuit can output $f(x)=\cos(\omega x)$ for any single $\omega\le\deg(n)+1$, whereas any product state $\rho$ limits the amplitude at $\omega=\deg(n)+1$ to $2^{-\deg(n)}$ (Appendix~\ref{app:activation}).

\textbf{Design procedure.} Given the largest frequency $\omega^{\star}$ a task needs, we measure a qubit with $\deg(n)\ge\omega^{\star}-1$ and place Expose, Spread and Align around its couplings on both sides of the encoder (Algorithm~\ref{alg:design}). The procedure reads only the degrees of $\mathcal{G}$, runs in $O(|E|)$ time and never enumerates the $4^{N}$-dimensional operator space. When the spectrum is unknown, $\omega^{\star}$ can be estimated from a discrete Fourier transform of the training targets \citep{boyd2001chebyshev,xu2024fits}, and our experiments use $\omega^{\star}=N$, for which the rule returns the star centered on the measured qubit.

\section{Experiments}
\textbf{Experiment Setup.} We evaluate the proposed design on sinusoid approximation, supervised regression of the PDEBench advection and Burgers' equations \citep{takamoto2022pdebench}, a placement ablation, a noise study and the physics-informed electromagnetic wave equation. We compare against the Basic Entangler and Strongly Entangling layers of PennyLane \citep{bergholm2018pennylane}, widely used in QNN studies \citep{Strong1,Strong2,Strong3,Strong4}, Circuits 15--19 of \citet{sim2019expressibility}, the circular hardware-efficient ansatz \citep{kandala2017hardware,strobl2025fourier} and the YZY ans\"atze of \citet{strobl2025fourier} (Appendix~\ref{appendix:baseline_circuits}). Unless stated otherwise, every circuit encodes its inputs once, which is the setting of our analysis, and a template reaches the parameter budget of the proposed circuit by stacking processing layers. Appendix~\ref{app:reupload} repeats the advection comparison and the noise study with data re-uploading, which lies beyond our analysis. Variants marked ($\gamma$) also train the entangling strengths $\gamma_{nk}$ of Lemma~\ref{lem:ising}, both proposed variants use the same two-qubit gates, and 2q in the tables counts them. Regression is scored by the MSE divided by the variance of the target, so a constant predictor scores $1$. The physics-informed tasks are scored by the relative $L^2$ error $\|\hat{u}_\theta-u\|_2/\|u\|_2$ between the prediction $\hat{u}_\theta$ and the analytic solution $u$ on an evaluation grid, which equals $1$ for the zero field, and a model solves a task when this error falls below $0.1$. We report mean $\pm$ s.d.\ over seeds and compare models with Welch's $t$-test under Holm correction.
\begin{table*}[t]
\centering
\small
\caption{Effect of $O$ and $\rho$ on PQC. }
\renewcommand{\arraystretch}{1.0}
\resizebox{0.95\linewidth}{!}{
\begin{tabular}{c|c||c|c}
\toprule[1pt]
Circuit (i) (Proposed) & $f_{\rm target,1}(x)$ & Circuit (ii) & $f_{\rm target,1}(x)$
\\ 
\midrule[1pt]
\includegraphics[width=0.26\linewidth]{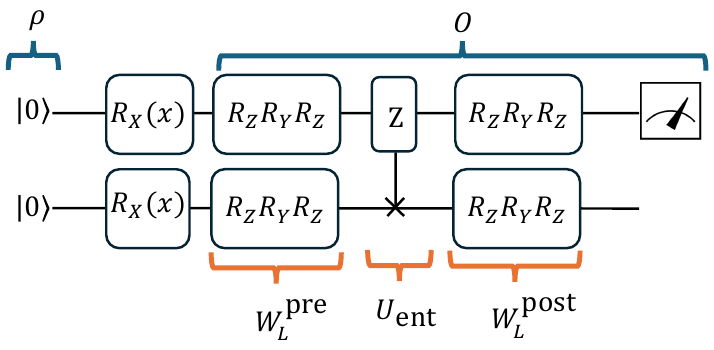}     & \includegraphics[width=0.22\linewidth]{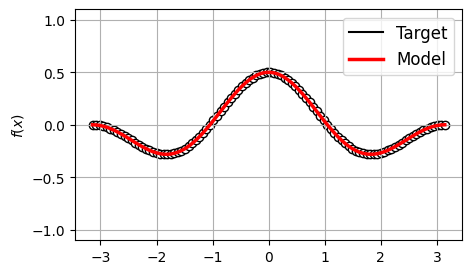}  &   
\includegraphics[width=0.26\linewidth]{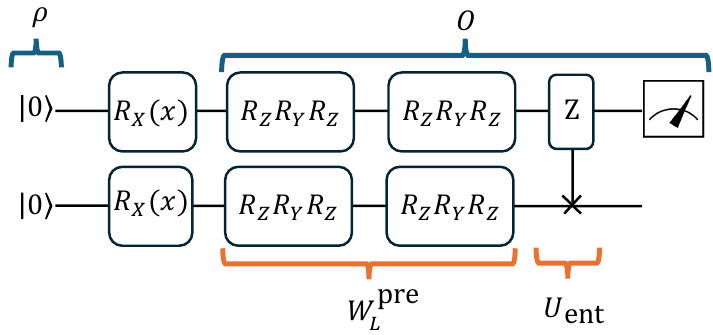}     & \includegraphics[width=0.22\linewidth]{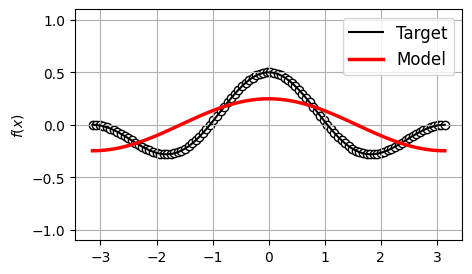}     \\   

\toprule[1pt]
Circuit (iii) (Proposed) & $f_{\rm target,2}(x)$ & Circuit (iv) & $f_{\rm target,2}(x)$
\\ 
\midrule[1pt]
\includegraphics[width=0.27\linewidth]{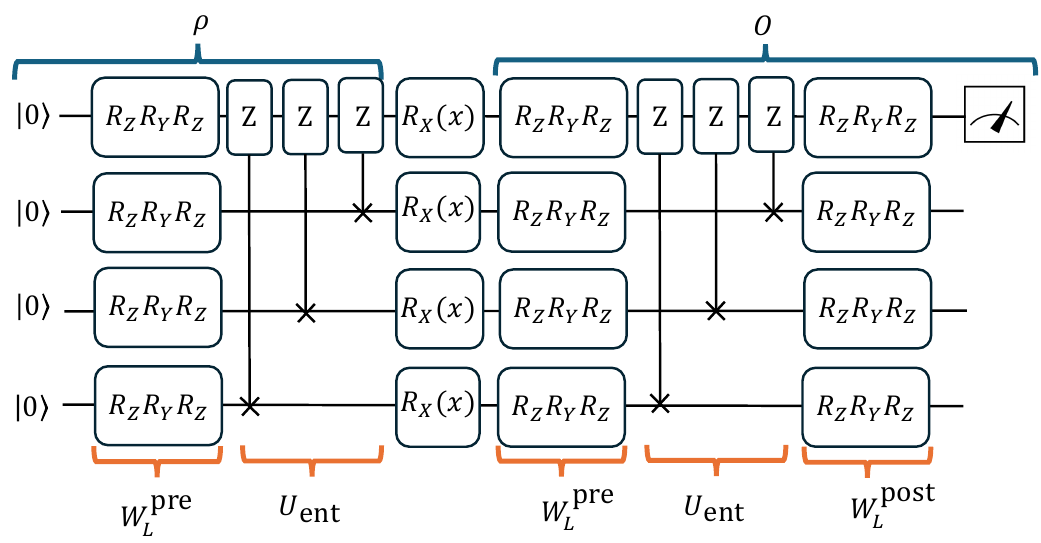}     & \includegraphics[width=0.22\linewidth]{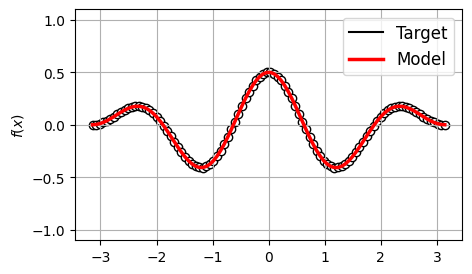} &

\includegraphics[width=0.26\linewidth]{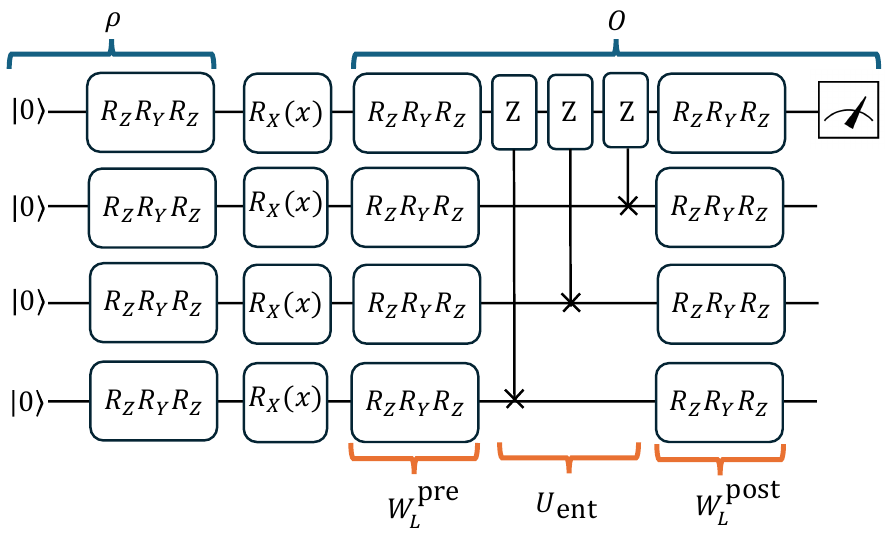}     & \includegraphics[width=0.22\linewidth]{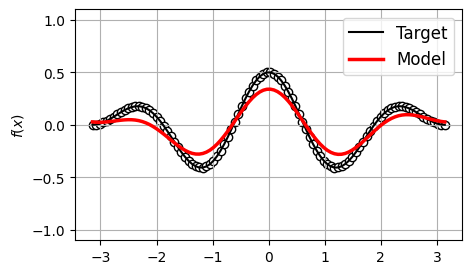} \\

\bottomrule[1pt]
\end{tabular}
}
\label{tab:table2}
\vspace{-3mm}
\end{table*}
\begin{figure}[t]
    \centering

    \begin{subfigure}[t]{0.36\linewidth}
        \centering
        \includegraphics[width=\linewidth]{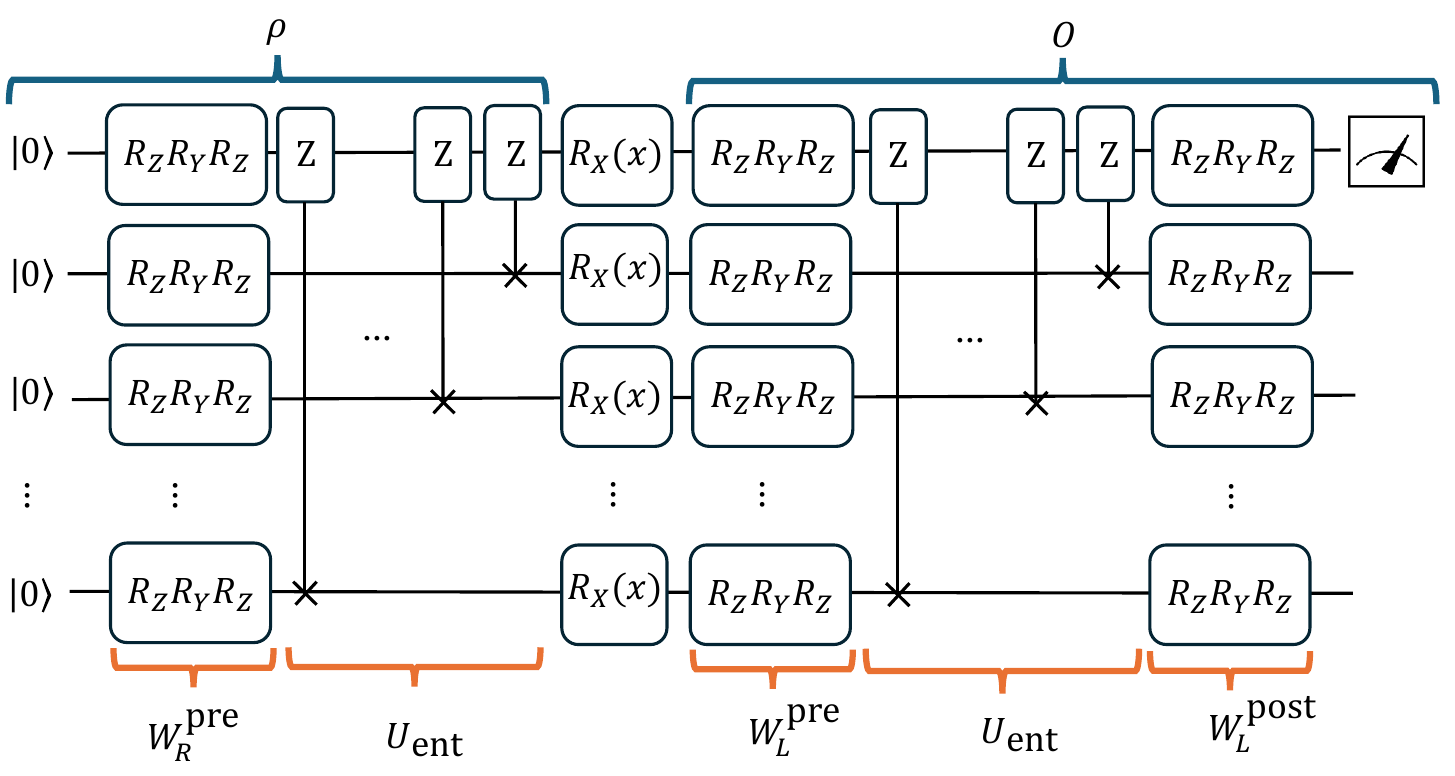}
        \caption{\textbf{Proposed circuit.}}
        \label{fig:proposed_circuit}
    \end{subfigure}
    \hspace{3mm}
    \begin{subfigure}[t]{0.22\linewidth}
        \centering
        \includegraphics[width=\linewidth]{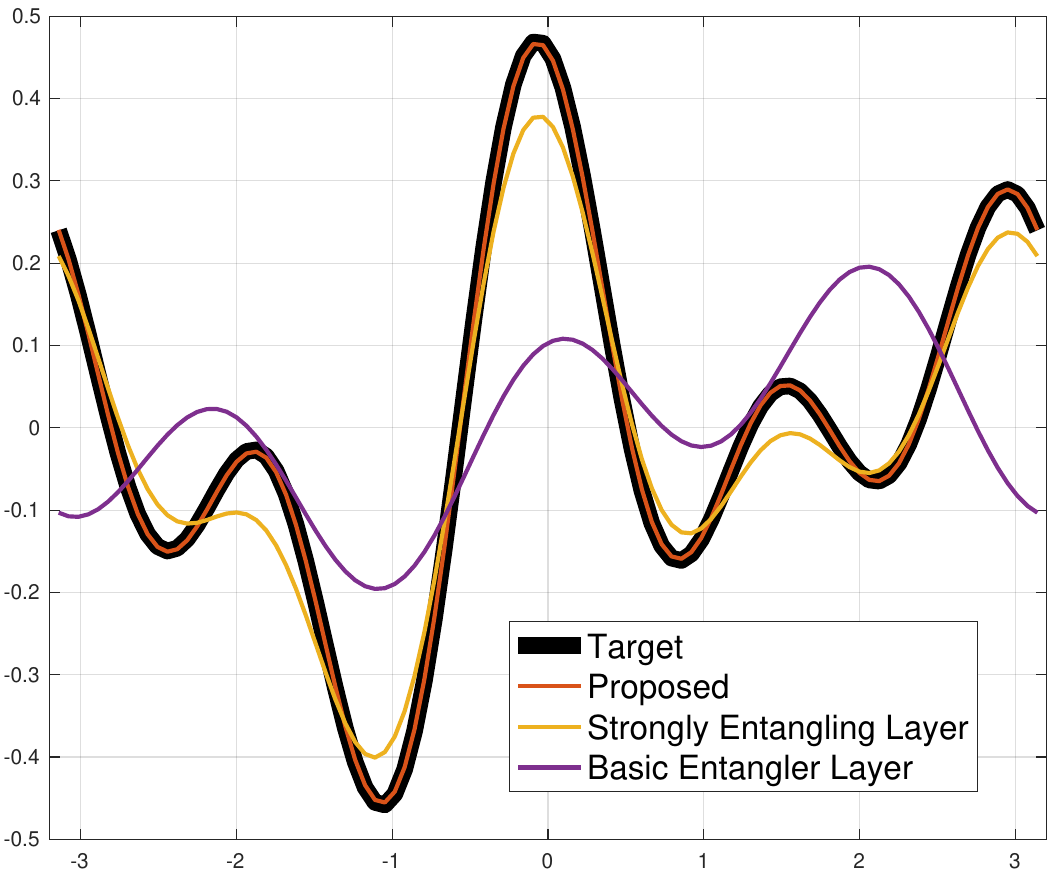}
        \caption{\textbf{Performance}}
        \label{fig:comparison_results}
    \end{subfigure}

    \caption{\textbf{Proposed circuit architecture and performance comparison.}
    (a) Structure of the proposed parameterized quantum circuit.
    (b) Comparison with the considered baseline architectures.}
    \label{fig:circuit_diagrams}
    \vspace{-3mm}
\end{figure}

\textbf{Sinusoid Function Approximation.}
We test the roles of the effective observable $O$ and effective state $\rho=W_R\rho_0W_R^\dagger$ with the circuit ablations in Table~\ref{tab:table2}. Circuits (i) and (iii) follow $W_L=W_L^{\mathrm{post}}U_{\mathrm{ent}}W_L^{\mathrm{pre}}$ with $R_ZR_YR_Z$ local blocks and a $CZ$ entangler, and Circuits (ii) and (iv) remove or misplace these local rotations and entangling layers. For $f_{\rm target,1}=0.25\cos(x)+0.25\cos(2x)$, two encoded qubits admit $\omega\leq2$, and Circuit (i) reconstructs the target while Circuit (ii) underfits the $\omega=2$ component. The encoder thus fixes which frequencies are admissible, while their amplitudes depend on how $O$ projects onto the invariant planes, and without local dressing $O$ lacks support on the mixed Pauli strings that span the higher-frequency planes. For $f_{\rm target,2}=0.25\cos(2x)+0.25\cos(3x)$, Circuit (iii), which dresses both $\rho$ and $O$, outperforms Circuit (iv), consistent with a coefficient that is bilinear in the two projections and is suppressed when either side lacks support. We then compare the Basic Entangler, the Strongly Entangling layer and the proposed circuit (Fig. \ref{fig:proposed_circuit}) on $f_{\rm target}=0.2\cos(2x)+0.1\sin(x)+0.11\cos(3x)+0.15\cos(4x)-0.1\sin(4x)$. The Basic Entangler fails to capture the target, the Strongly Entangling layer fits it with visible deviation, and the proposed circuit matches it across the domain, with a structure that follows from the invariant-plane analysis rather than a heuristic gate arrangement.

\textbf{Supervised PDE Regression.}
\begin{figure}[t]
    \centering
    \small
    \includegraphics[width=0.15\linewidth]{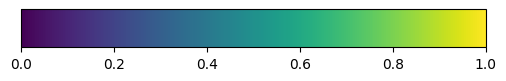}\\
    \begin{tabular}{ccccc}
     \includegraphics[width=0.163\linewidth]{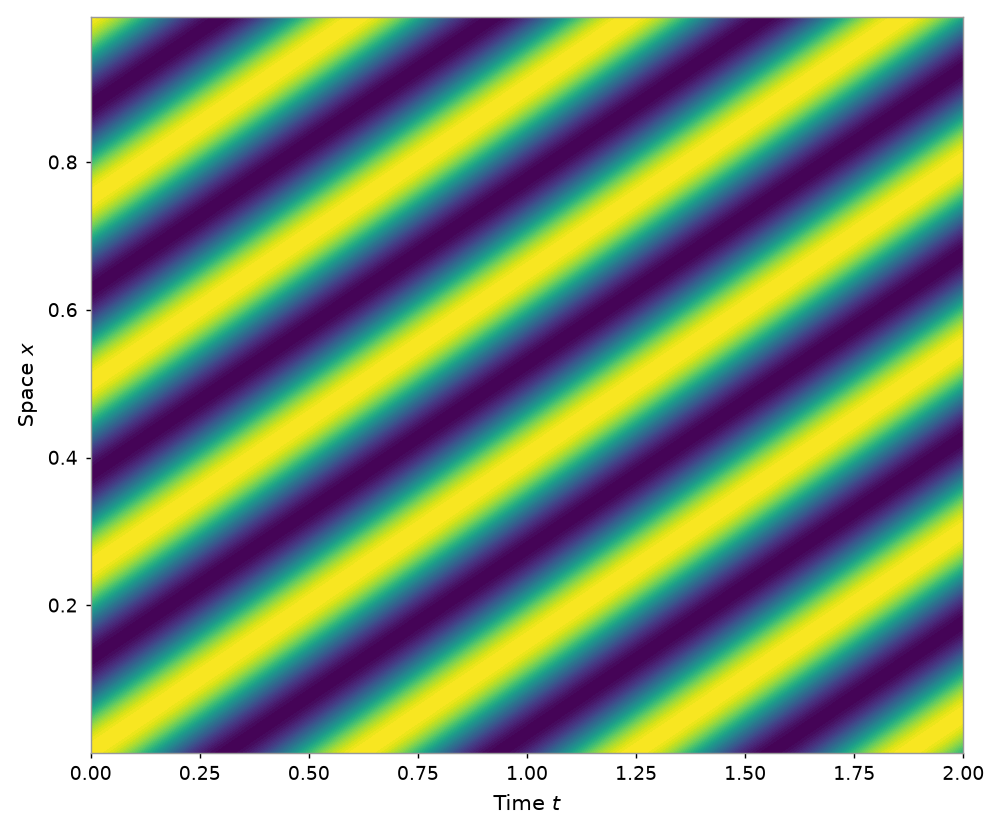} &
     \includegraphics[width=0.163\linewidth]{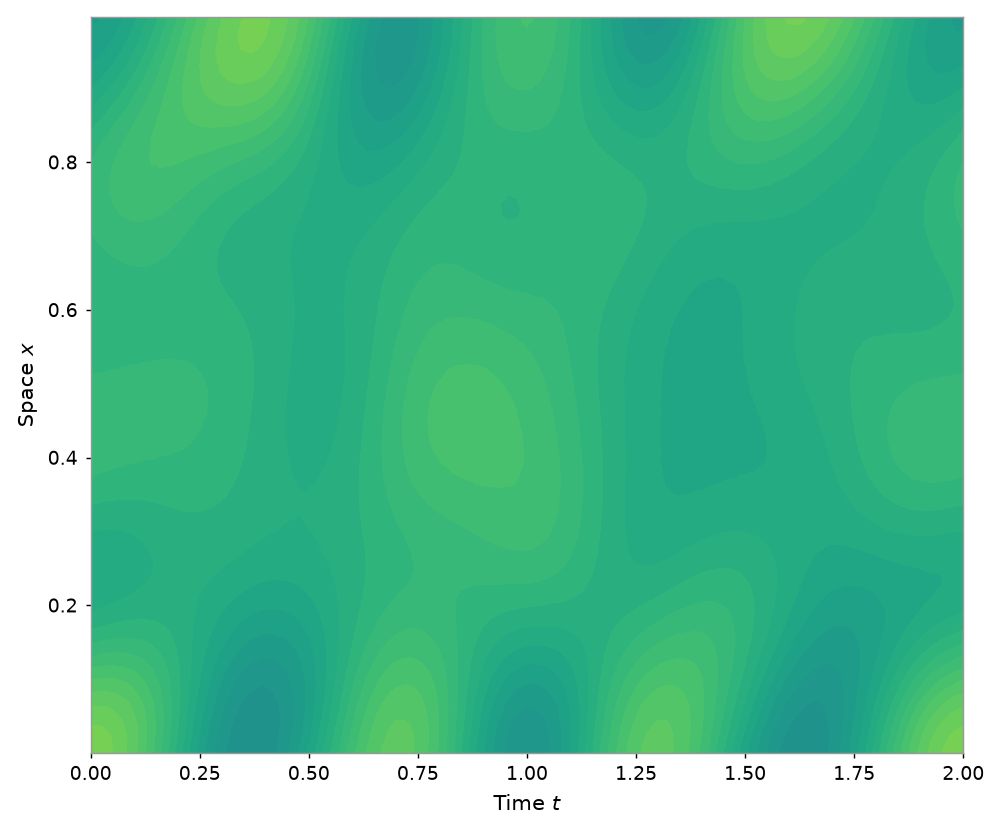} &
     \includegraphics[width=0.163\linewidth]{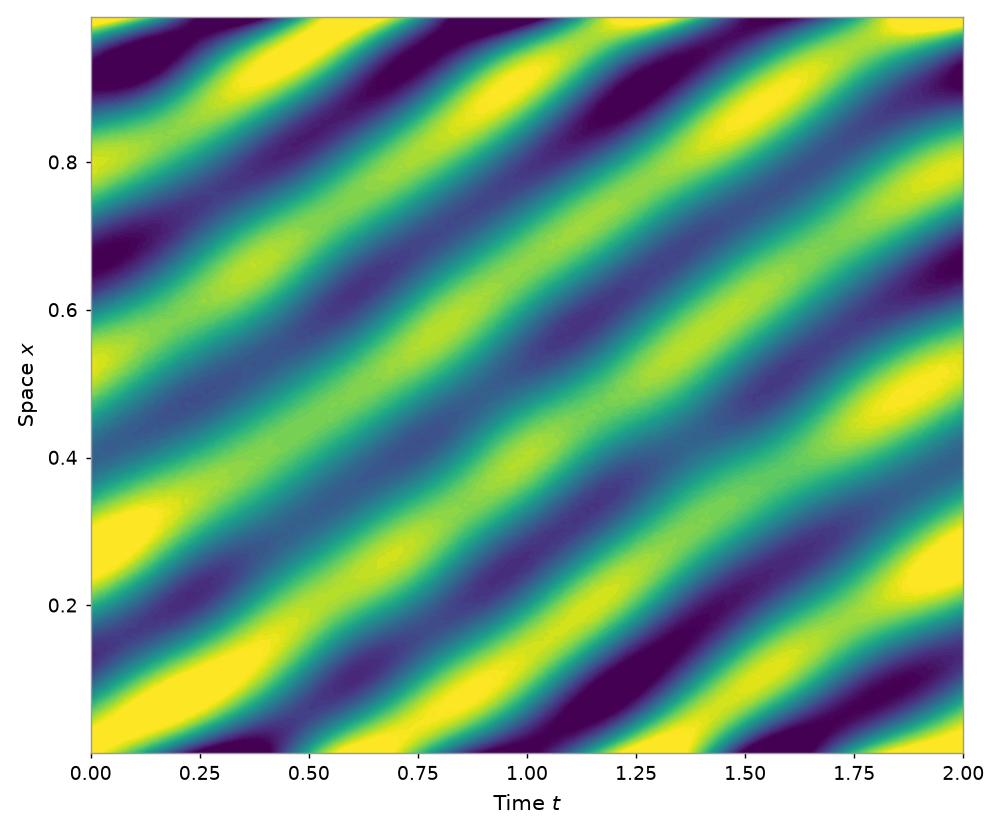} &
     \includegraphics[width=0.163\linewidth]{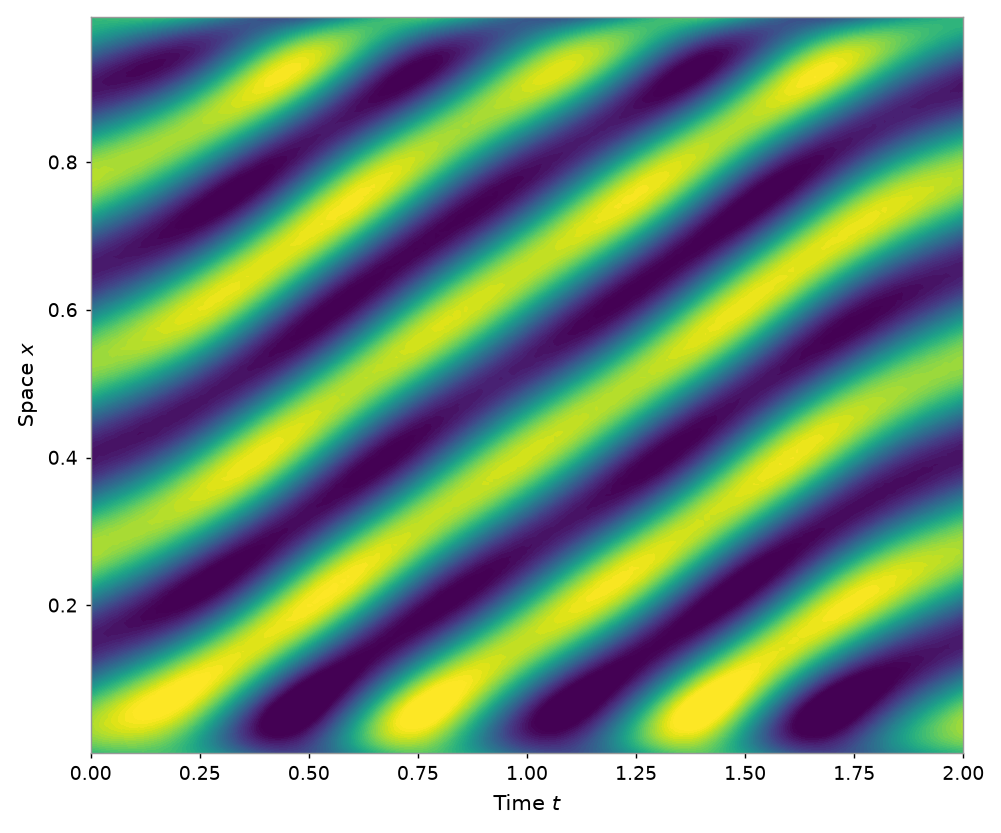} &
     \includegraphics[width=0.163\linewidth]{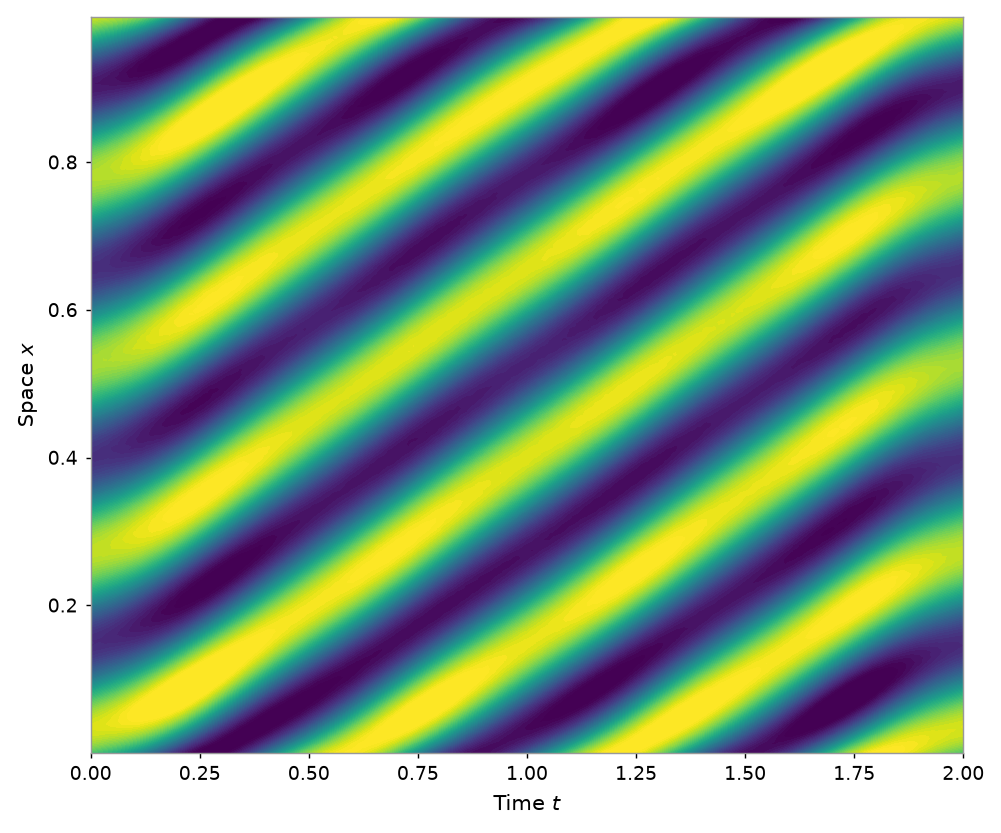}
     \\
     Ground Truth & Basic Ent. & Strongly Ent. ($L=4$) & \textbf{Proposed} & \textbf{Proposed} ($\gamma$)
    \end{tabular}
    \caption{\textbf{Single-upload reconstructions at 14 qubits}. Advection equation ($\beta=0.4$).}
    \label{fig:Advection_equation_04}
    \vspace{-4mm}
\end{figure}
We further evaluate the circuits on two PDEBench tasks \citep{takamoto2022pdebench}, the advection equation with $\beta=0.4$ and Burgers' equation with $\nu=0.001$ (Table~\ref{tab:pdebench_setup}, Appendix~\ref{appendix:ExpDetails}), whose solutions share their initial condition. We regress each field $u(x,t)$ with $x$ encoded on odd-indexed and $t$ on even-indexed qubits. We first use the setting of our analysis, 14 qubits and a single encoding layer, where a template can reach the budget only by stacking processing layers that add parameters but no frequencies (Table~\ref{tab:advection14}). On advection, six of the eight templates, including the Basic Entangler and Circuits 16--19, do not beat the constant predictor, and the best, the circular HEA, reaches $0.155$ with four layers and $168$ parameters, while the proposed circuit reaches $0.107$ with one layer and $126$ parameters, or $0.039$ when it also trains its entangling strengths, and reproduces the stripes most closely (Fig.~\ref{fig:Advection_equation_04}). Burgers, whose steep fronts make it harder for every circuit, keeps this order, with seven templates near the constant predictor, the circular HEA successful in one of three runs and both proposed variants below $0.21$ on every seed. With data re-uploading, a setting beyond our analysis, the proposed circuit also leads every matched template at 7 qubits while using $48$ two-qubit gates against $96$ to $329$ (Appendix~\ref{app:reupload}).

\begin{table}[htb]
\caption{Single-upload regression at 14 qubits over 3 seeds, MSE normalized by the variance of $u$
so that a constant predictor scores $1$. $L$ is the number of layers.
\textbf{Bold} and \textcolor{blue}{\textbf{blue bold}} mark the best and second-best result per task.}
\label{tab:advection14}
\centering

\resizebox{\textwidth}{!}{%
\begin{tabular}{l|ccc|cc||l|ccc|cc}
\Xhline{3\arrayrulewidth}
\textbf{Method} & $|\theta|$ & $L$ & 2q & Advection & Burgers &
\textbf{Method} & $|\theta|$ & $L$ & 2q & Advection & Burgers \\
\Xhline{1\arrayrulewidth}

Circuit 18 & 164 & 4 & 65 & $1.0604\pm0.0725$ & $1.0222\pm0.0297$ &
Circuit 19 & 168 & 4 & 65 & $1.0533\pm0.0660$ & $1.0220\pm0.0295$ \\

Circuit 17 & 164 & 4 & 28 & $1.0503\pm0.0620$ & $1.0222\pm0.0297$ &
Circuit 16 & 164 & 4 & 28 & $1.0552\pm0.0645$ & $1.0223\pm0.0299$ \\

HEA circular & 168 & 4 & 69 & $0.1549\pm0.0469$ & $0.7744\pm0.2781$ &
YZY (ent.) & 168 & 4 & 32 & $1.0989\pm0.0611$ & $1.0380\pm0.0501$ \\

StronglyEnt. & 168 & 4 & 56 & $0.2898\pm0.0761$ & $0.9267\pm0.0259$ &
BasicEnt. & 154 & 11 & 154 & $1.0335\pm0.0757$ & $1.0224\pm0.0302$ \\

\Xhline{1\arrayrulewidth}

\textbf{Proposed} & 126 & 1 & 26 &
\textcolor{blue}{$\mathbf{0.1065\pm0.0284}$} &
\textcolor{blue}{$\mathbf{0.1951\pm0.0021}$} &
\textbf{Proposed ($\gamma$)} & 152 & 1 & 26 &
$\mathbf{0.0393\pm0.0037}$ &
$\mathbf{0.1890\pm0.0143}$ \\

\Xhline{3\arrayrulewidth}
\end{tabular}%
}

\end{table}

\textbf{Placement Ablation.} We keep the $126$ angles and $26$ $CZ$ gates of the proposed 14-qubit circuit and change only their order, the graph or the measured qubit (Table~\ref{tab:ablation}). Since each Fourier coefficient $c(\omega_{x},\omega_{t})$ is identically zero or nonzero almost everywhere (Corollary~\ref{cor:activation}), we count the nonzero ones among the $225$ exactly at random parameters, without any training (Appendix~\ref{app:ablation}). Every change removes exactly the pairs our analysis predicts, and No Align shows that reaching the top frequency is not enough, since it keeps $(7,7)$ but loses $103$ of the $225$ pairs. To connect these counts to the PDE error, we fit the advection field of Table~\ref{tab:advection14} by least squares on the reachable terms of each variant, which gives a lower bound on the error of every parameter setting. Without Expose, with a leaf readout or with a path, this bound is at least $0.999$, near the constant predictor, and without Align it is $0.102$, against $0.013$ for the proposed circuit. Training every variant the same way, on one eighth of the grid with about one sixth of the optimizer steps of Table~\ref{tab:advection14}, gives the proposed circuit $0.63$ and every other variant between $0.96$ and $1.02$.

\begin{table}[htb]
\caption{Placement ablation, with gates in time order, $R$ a rotation layer, $S$ the encoder and $C$ or $P$ a star or path of $13$ $CZ$ gates. Pairs counts the nonzero $c(\omega_{x},\omega_{t})$ of $225$ and $(k_{x},k_{t})$ is their largest $(|\omega_{x}|,|\omega_{t}|)$, both as predicted. Best possible MSE is the least-squares error on the reachable terms, a lower bound for any parameter setting, and Trained MSE uses a reduced budget over 5 seeds, both normalized by the variance of $u$.}
\label{tab:ablation}
\centering
\scriptsize
\setlength{\tabcolsep}{4pt}
\begin{tabular}{l|c|cccc}
\Xhline{3\arrayrulewidth}
 & \textbf{Proposed} & No Expose & No Align & Leaf readout & Path \\
\Xhline{1\arrayrulewidth}
Gates & $R\,C\,S\,R\,C\,R$ & $R\,C\,S\,R\,R\,C$ & $R\,C\,S\,C\,R\,R$ & $R\,C\,S\,R\,C\,R$ & $R\,P\,S\,R\,P\,R$ \\
$(k_{x},k_{t})$ & $(7,7)$ & $(1,0)$ & $(7,7)$ & $(1,1)$ & $(2,1)$ \\
Pairs & $225$ & $3$ & $122$ & $9$ & $15$ \\
Best possible MSE $\downarrow$ & $\mathbf{0.013}$ & $1.000$ & $0.102$ & $1.000$ & $0.999$ \\
Trained MSE $\downarrow$ & $\mathbf{0.628\pm0.212}$ & $1.020\pm0.024$ & $0.957\pm0.016$ & $1.015\pm0.012$ & $1.022\pm0.012$ \\
\Xhline{3\arrayrulewidth}
\end{tabular}
\end{table}

\textbf{Noise Simulation.} We evaluate the trained 14-qubit advection models under depolarizing noise of strength $\epsilon$ after every gate, amplitude and phase damping of strength $\lambda$ after every circuit moment, and finite-shot readout with $S$ shots (Table~\ref{tab:noise14}, Appendix~\ref{app:noise14}). The proposed circuit stays accurate under moderate noise. With trained entangling strengths, its error moves from $0.026$ to $0.036$ and $0.105$ at $\epsilon=3\times10^{-3}$ and $10^{-2}$, against $0.250$ and $0.473$ for HEA, the strongest template, and to $0.067$ at $\lambda=10^{-3}$, while stronger damping at $\lambda=3\times10^{-3}$ raises it to about $0.5$. Finite-shot readout is the limiting factor at 14 qubits. Its error grows with the square of the trained readout scale, which is four times smaller for our circuit than for HEA, so our shot error is $16\times$ smaller, but our circuit still needs more than $10^{4}$ shots per point to beat the constant predictor. An exact density-matrix study at 7 qubits with every template, under data re-uploading, is in Appendix~\ref{app:reupload}.

\begin{table}[htb]
\caption{Noise at 14 qubits with a single encoding layer, MSE normalized by the variance of $u$. Gate noise and damping use seed 0 on 128 evaluation points with 512 debiased trajectories, and shots use 3 seeds on 2048 points with 10 draws each. $|w|$ is the trained readout scale.}
\label{tab:noise14}
\centering
\scriptsize
\setlength{\tabcolsep}{3pt}
\begin{tabular}{l|cc|c|cc|cc|cc}
\Xhline{3\arrayrulewidth}
 & & & & \multicolumn{2}{c|}{Gate noise $\epsilon$} & \multicolumn{2}{c|}{Damping $\lambda$} & \multicolumn{2}{c}{Shots $S$} \\
\textbf{Method} & 2q & $|w|$ & Clean & $3\times10^{-3}$ & $10^{-2}$ & $10^{-3}$ & $3\times10^{-3}$ & $128$ & $512$ \\
\Xhline{1\arrayrulewidth}
HEA circular & 69 & 18.48 & 0.193 & 0.250 & 0.473 & -- & -- & 1585 & 397 \\
\textbf{Proposed} & 26 & 5.79 & 0.064 & 0.079 & 0.154 & 0.082 & 0.494 & 156.1 & 39.3 \\
\textbf{Proposed ($\gamma$)} & 26 & 4.64 & 0.026 & 0.036 & 0.105 & 0.067 & 0.557 & 99.3 & 25.0 \\
\Xhline{3\arrayrulewidth}
\end{tabular}
\end{table}

\begin{figure}[t]
    \centering
    \small
    \begin{tabular}{cccc}
     \includegraphics[width=0.20\linewidth]{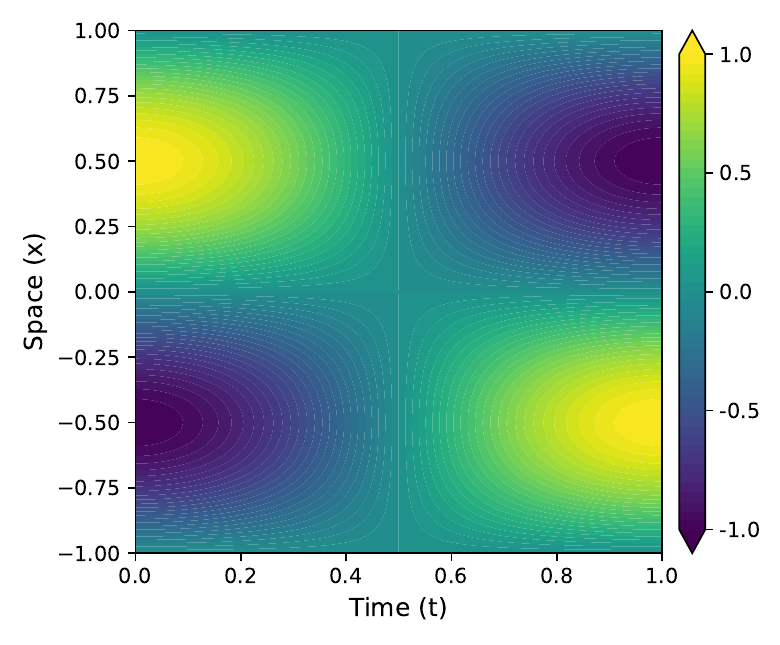} &
     \includegraphics[width=0.20\linewidth]{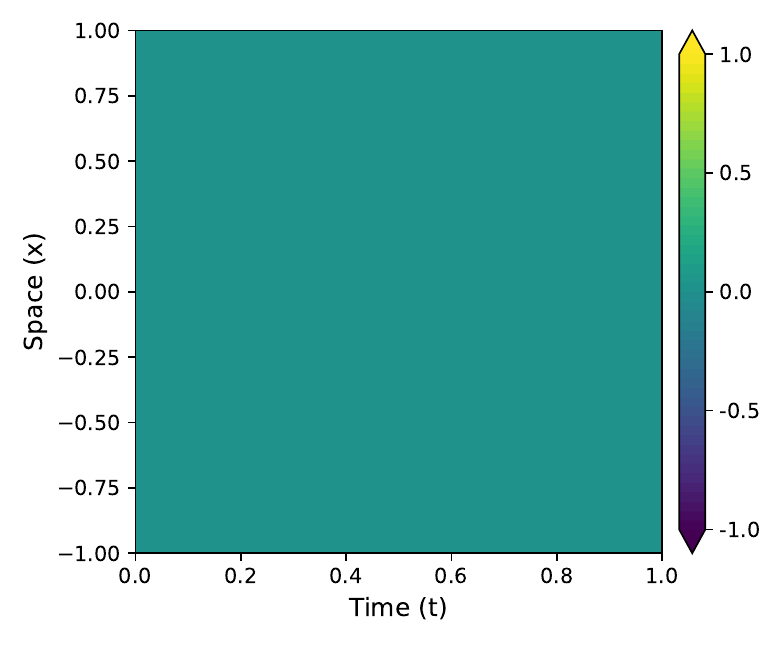} &
     \includegraphics[width=0.20\linewidth]{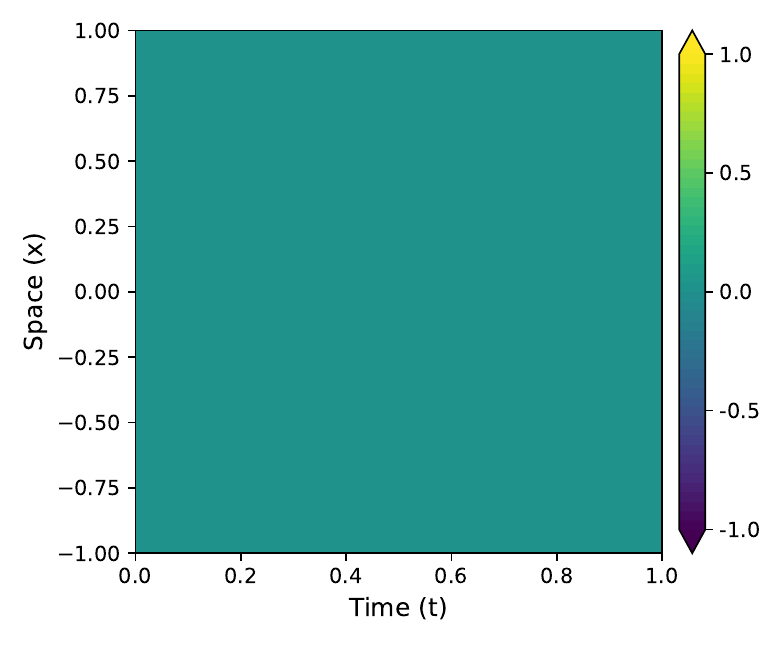} &
     \includegraphics[width=0.20\linewidth]{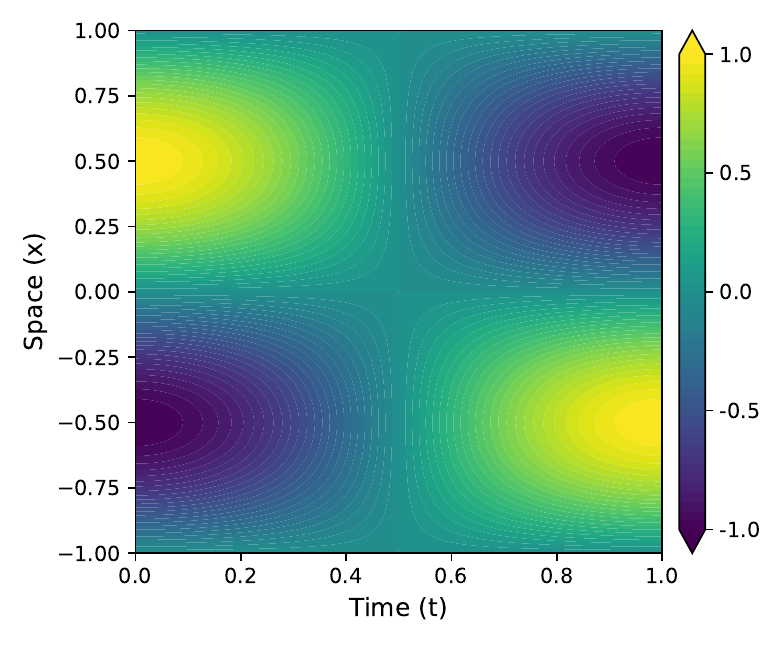}
     \\
     Ground Truth & Basic Ent. & Strongly Ent. & \textbf{Proposed}
    \end{tabular}
    \caption{\textbf{Physics-informed wave equation} with one layer per circuit. The Strongly Entangling layer on 6 qubits and the proposed circuit on 2 qubits both use 20 parameters.}
    \label{fig:wave}
\end{figure}

\begin{table}[t]
\caption{Physics-informed wave equation, mean relative $L^2$ error. Templates need more layers, since with one layer each outputs the zero field (Table~\ref{tab:wave_app}).}
\label{tab:wave}
\centering
\scriptsize
\setlength{\tabcolsep}{4pt}
\begin{tabular}{lcccc|lcccc}
\Xhline{3\arrayrulewidth}
Circuit & Qubits & $L$ & $|\theta|$ & Rel.\ $L^2$ & Circuit & Qubits & $L$ & $|\theta|$ & Rel.\ $L^2$ \\
\Xhline{1\arrayrulewidth}
Circuit 16 & 6 & 2 & 36 & $8.1\times10^{-5}$ & YZY & 6 & 2 & 38 & $0.336$ \\
Circuit 17 & 6 & 2 & 36 & $3.1\times10^{-5}$ & YZY (ent.) & 6 & 2 & 38 & $0.336$ \\
Circuit 18 & 6 & 2 & 38 & $3.6\times10^{-4}$ & HEA circular & 6 & 2 & 38 & $0.455$ \\
Circuit 19 & 6 & 2 & 38 & $3.0\times10^{-5}$ & Circuit 15 & 6 & 2 & 26 & $1.000$ \\
Strongly Ent. & 6 & 3 & 56 & $9.9\times10^{-3}$ & Basic Ent. & 6 & 2 & 14 & $1.000$ \\
\Xhline{1\arrayrulewidth}
\textbf{Proposed} & 2 & 1 & 20 & $2.7\times10^{-5}$ & & & & & \\
\Xhline{3\arrayrulewidth}
\end{tabular}
\end{table}

\textbf{Physics-Informed Wave Equation.} We also solve the electromagnetic wave equation from Maxwell's equations with the physics-informed objective of Table~\ref{tab:wave_setup} and no solution data. Its solution $\sin(\pi x)\cos(\pi t)$ has one frequency per input, so our rule needs one qubit per input, and we compare circuits at equal depth. With only 2 qubits and one layer of 20 parameters, the proposed circuit solves the task with a relative $L^2$ error of $2.7\times10^{-5}$, while no template learns anything with one layer even on three times as many qubits, each outputting the zero field (Fig.~\ref{fig:wave}, Table~\ref{tab:wave_app}). Templates need more layers, two and 36 to 38 parameters for Circuits 16--19 and three and 56 for the Strongly Entangling layer (Table~\ref{tab:wave}), so ours uses a third of their qubits and about half their parameters, although larger templates can exceed its accuracy.

\section{Conclusion}
In this paper, we presented a geometric framework for analyzing the Fourier expressivity of multi-qubit PQCs. Using the adjoint map, we showed that the data-encoding generator decomposes the operator space into invariant frequency subspaces, where Fourier coefficients are determined by geometric projections between the effective state and observable. This perspective reveals how entanglement activates higher-frequency components by transforming local operators into mixed multi-qubit Pauli strings. For the encoder, readout and entangler we analyze, it yields a placement rule, built from the interaction graph in linear time, that makes every frequency up to the degree of the measured qubit plus one nonzero at almost every parameter setting, and each of these frequencies can reach unit amplitude on its own when the state side is entangled as well. Moving the same gates removes exactly the frequencies that the rule predicts. With trained entangling strengths the proposed circuit leads every parameter-matched template on advection and Burgers while using fewer two-qubit gates, it stays accurate under moderate gate noise and damping, and a single encoding layer suffices to solve the wave equation.

\newpage 
\section*{GenAI Usage Disclosure}
We occasionally used ChatGPT to refine our wording and grammar. All manuscript and contents of this paper were checked and reviewed by authors. 

\section*{Ethics statement}
All authors have read and adhere to the ICLR Code of Ethics. This work develops theory and a design rule for parameterized quantum circuits and evaluates them in classical simulation. It involves no human subjects, no personal or sensitive data and no deployed system. The supervised experiments use the public PDEBench datasets \citep{takamoto2022pdebench} under their license, and every other target is synthetic and given in closed form. We see no direct path to harmful use, and the computational cost of the work is limited to classical simulation on a single workstation (Table~\ref{tab:computational_resources}).

\section*{Reproducibility statement}
The proofs of all lemmas and propositions are in Appendix~\ref{sec:appendix_proofsofLemmasandprop}, the proof of Corollary~\ref{cor:activation} is in Appendix~\ref{app:activation}, the scope of the results is discussed in Appendix~\ref{app:general_ent}, and the design procedure is stated as Algorithm~\ref{alg:design}. Appendix~\ref{appendix:baseline_circuits} defines every baseline circuit, and Appendix~\ref{sec:appendix_experiment_setup} gives the data, encoding, objective and training settings of each experiment, with the number of seeds and the hardware and software we used (Table~\ref{tab:computational_resources}). The comparison tables list the parameter and two-qubit gate counts of every circuit, and we report the mean and standard deviation over seeds with Holm-corrected Welch tests. The supervised data are public PDEBench samples, the other targets are given in closed form, and the reachable-pair counts of Table~\ref{tab:ablation} require no training and follow from the procedure of Appendix~\ref{app:ablation}. Our code, built on PennyLane \citep{bergholm2018pennylane}, JAX and Optax, is included in the supplementary material and will be made public upon publication.

\bibliography{iclr2027_conference}
\bibliographystyle{iclr2027_conference}

\newpage
\appendix

\section{Baseline Circuits}
\label{appendix:baseline_circuits}
\begin{figure*}[h]
    \centering
    \begin{subfigure}[t]{0.48\textwidth}
        \centering
        \includegraphics[width=\linewidth]{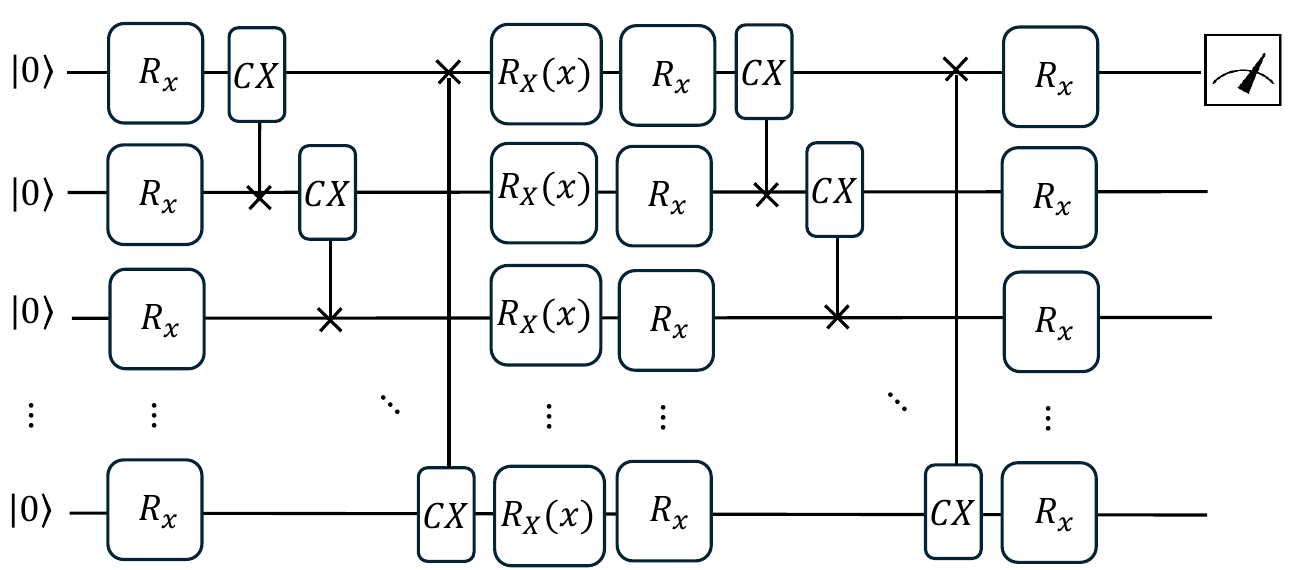}
        \caption{Basic Entangler}
        \label{fig:basic}
    \end{subfigure}
    \hfill
    \begin{subfigure}[t]{0.45\textwidth}
        \centering
        \includegraphics[width=\linewidth]{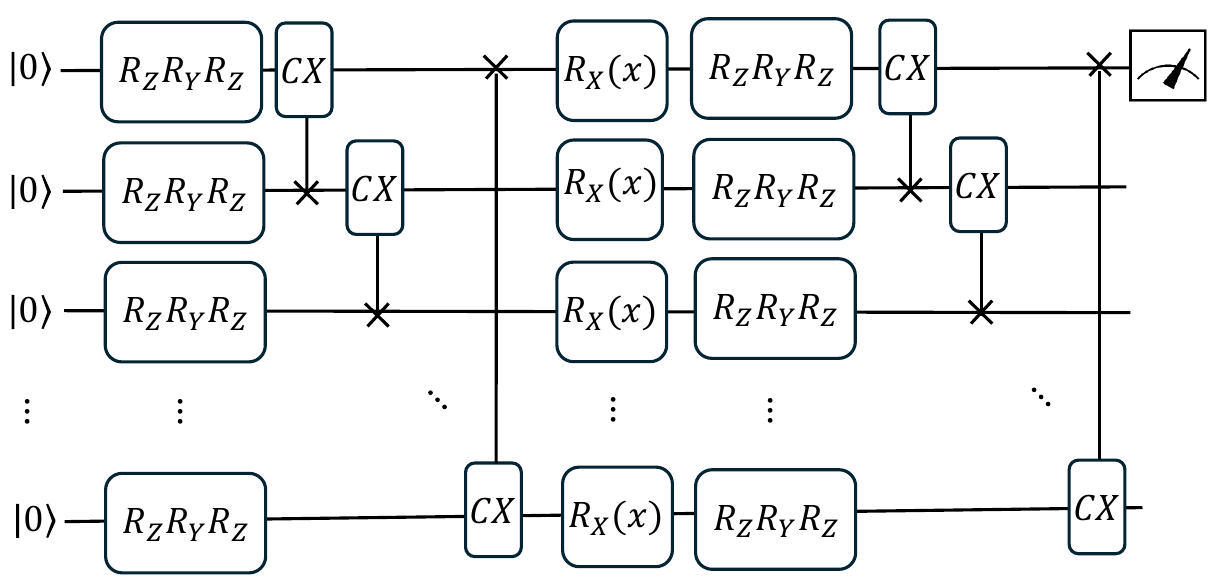}
        \caption{Strongly Entangling}
        \label{fig:strongly}
    \end{subfigure}
    
    \begin{subfigure}[t]{0.48\textwidth}
        \centering
        \includegraphics[width=\linewidth]{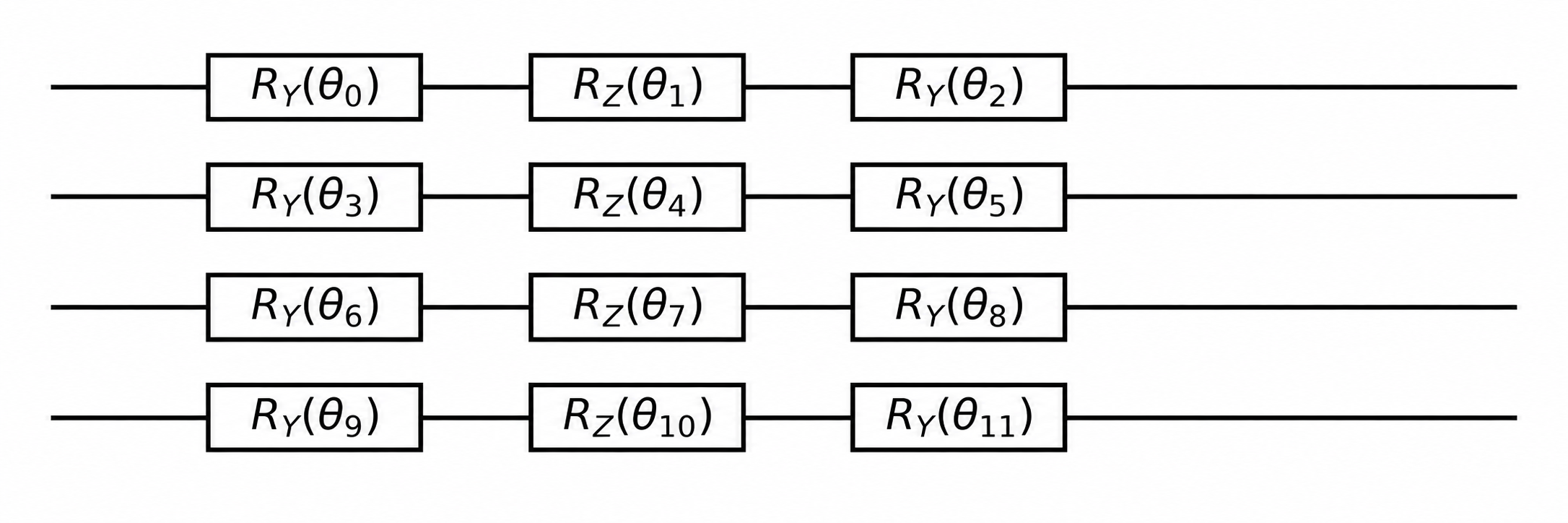}
        \caption{YZY}
        \label{fig:yzy}
    \end{subfigure}
    \hfill
    \begin{subfigure}[t]{0.48\textwidth}
        \centering
        \includegraphics[width=\linewidth]{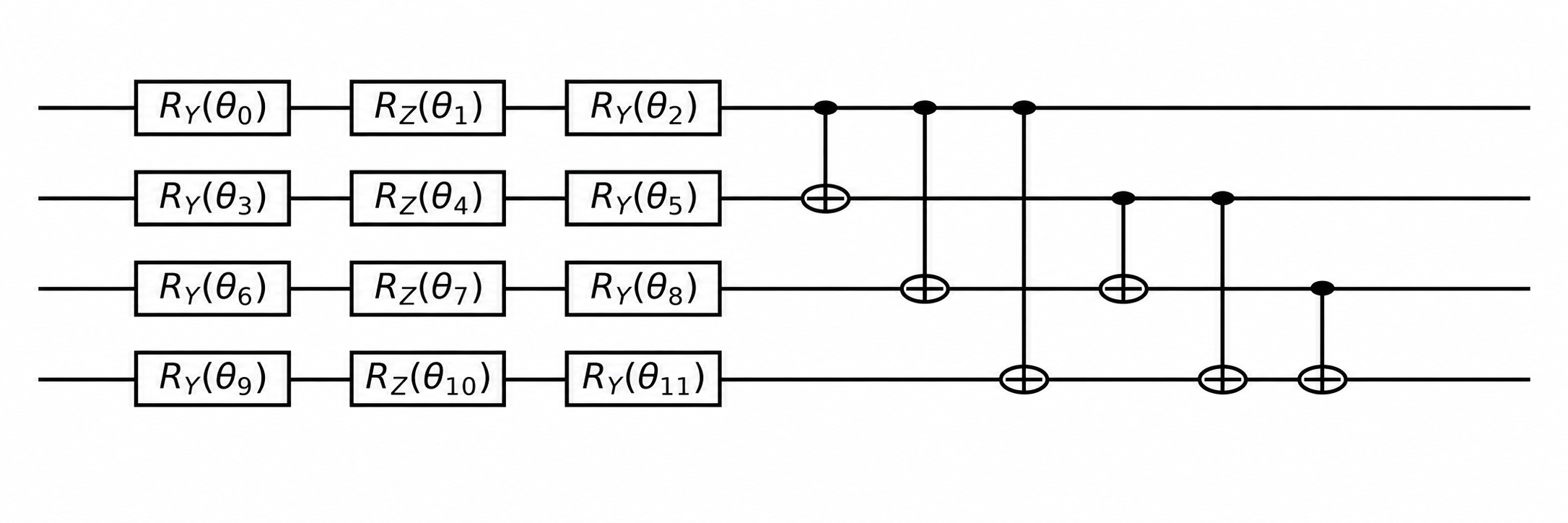}
        \caption{YZY Entangling}
        \label{fig:yzy_entang}
    \end{subfigure}


    \begin{subfigure}[t]{0.48\textwidth}
        \centering
        \includegraphics[width=\linewidth]{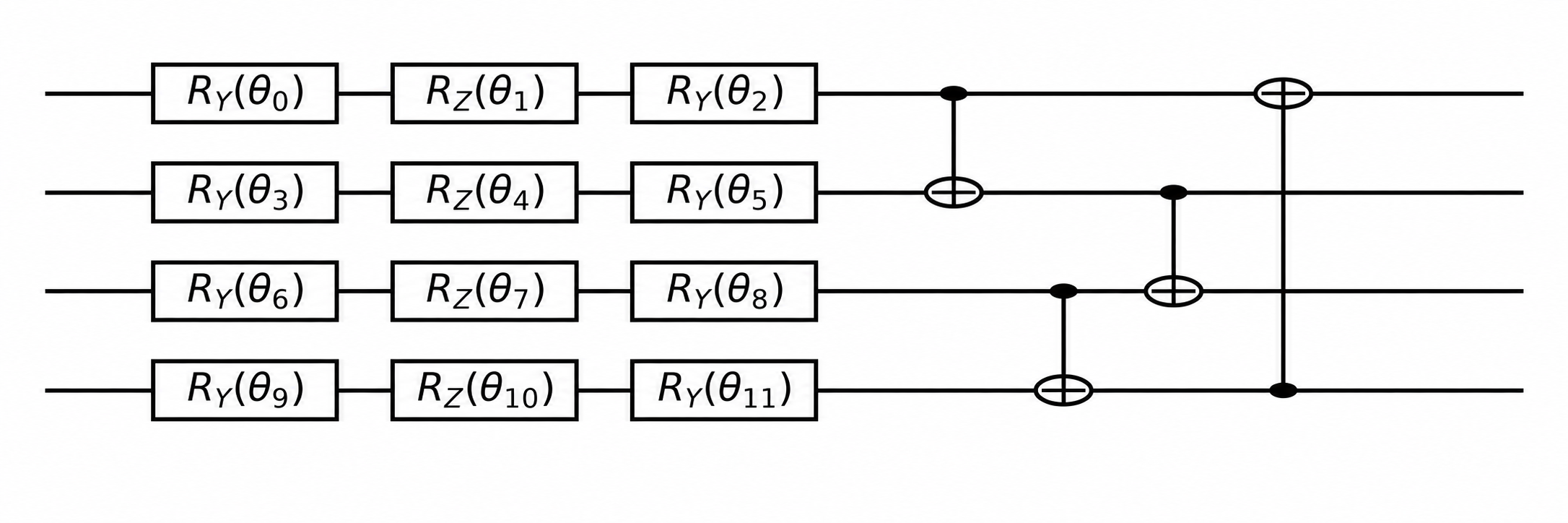}
        \caption{Hardware Efficient Ansatz (HEA)}
        \label{fig:HEA}
    \end{subfigure}
    \hfill
    \begin{subfigure}[t]{0.48\textwidth}
        \centering
        \includegraphics[width=\linewidth]{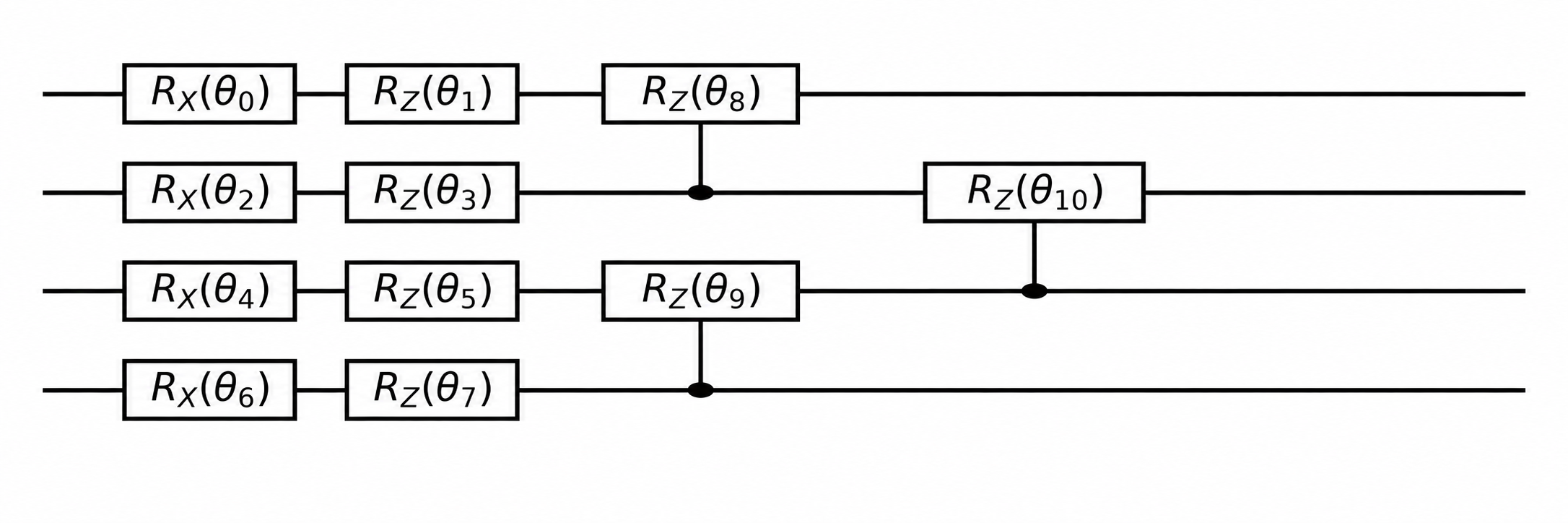}
        \caption{Circuit 16}
        \label{fig:Circuit16}
    \end{subfigure}


    \begin{subfigure}[t]{0.48\textwidth}
        \centering
        \includegraphics[width=\linewidth]{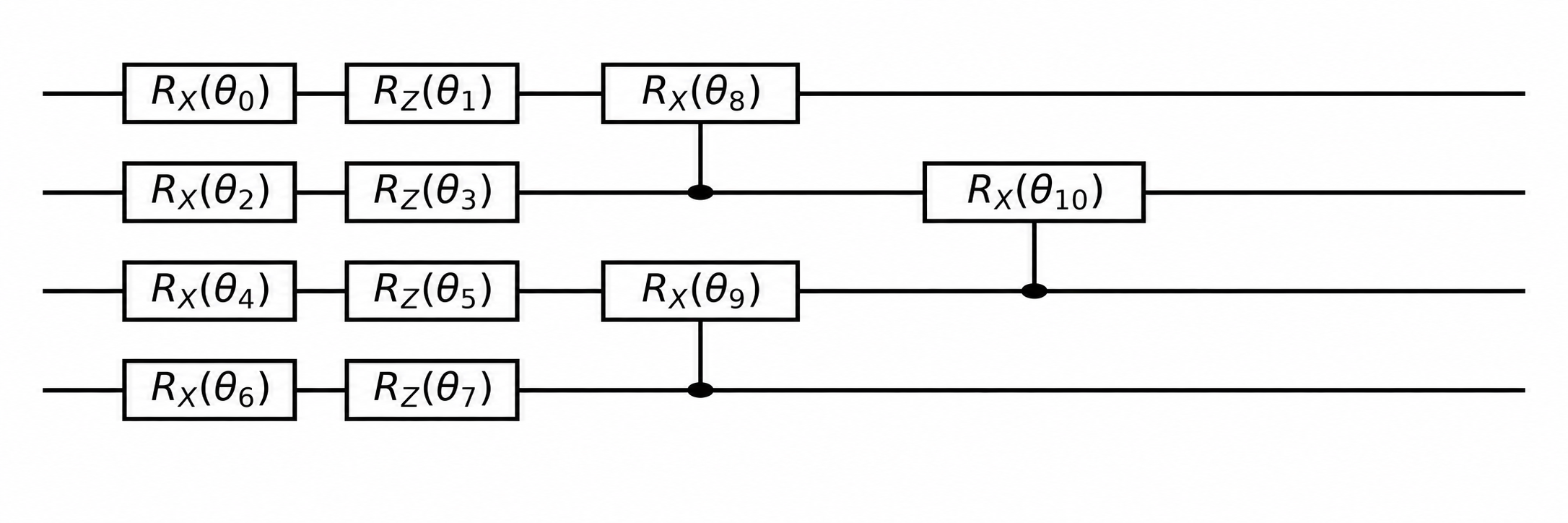}
        \caption{Circuit 17}
        \label{fig:Circuit17}
    \end{subfigure}
    \hfill
    \begin{subfigure}[t]{0.48\textwidth}
        \centering
        \includegraphics[width=\linewidth]{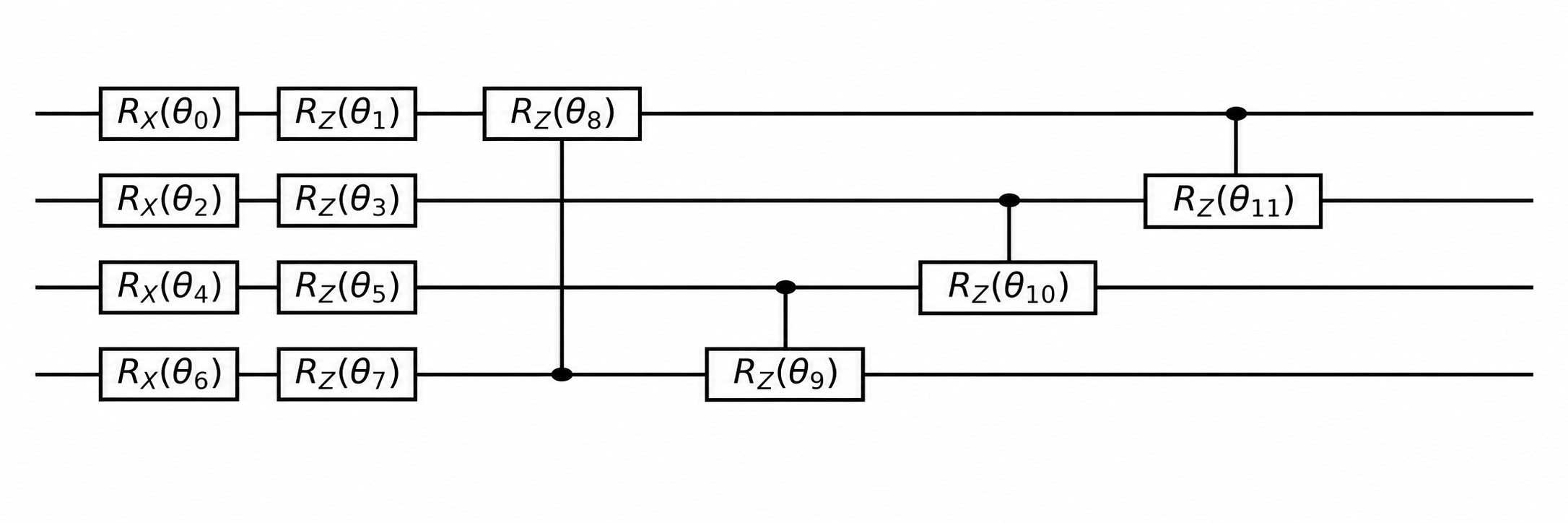}
        \caption{Circuit 18}
        \label{fig:Circuit18}
    \end{subfigure}


    \begin{subfigure}[t]{0.48\textwidth}
        \centering
        \includegraphics[width=\linewidth]{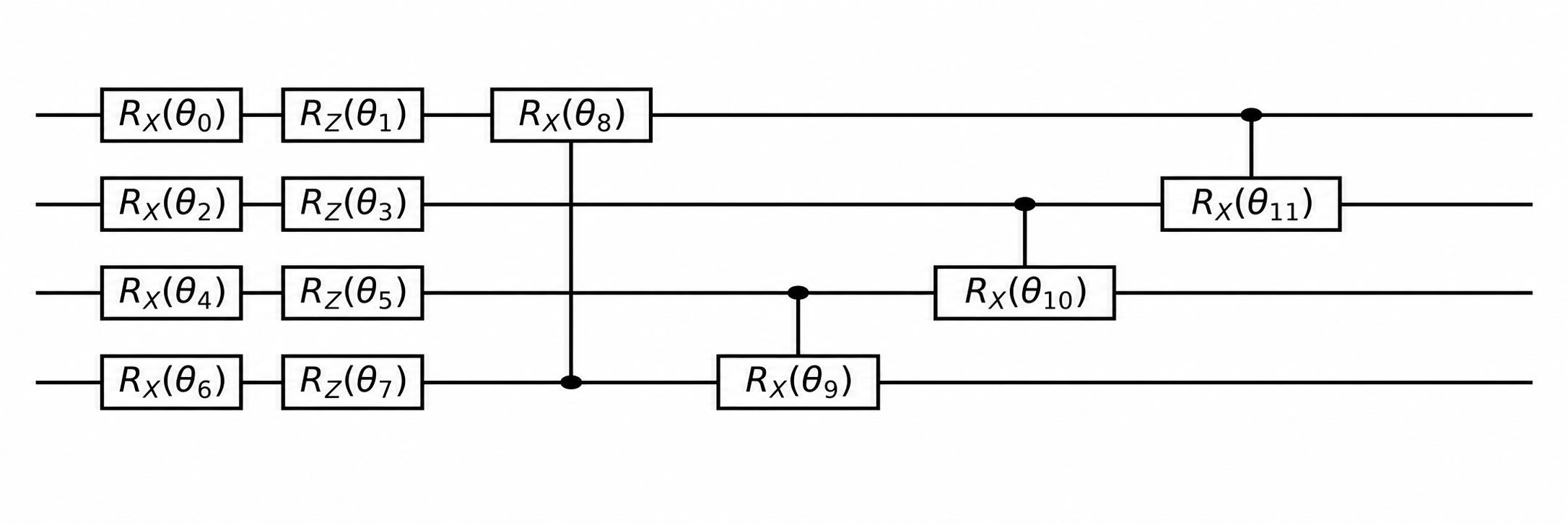}
        \caption{Circuit 19}
        \label{fig:Circuit19}
    \end{subfigure}

    \caption{PQC baseline architectures considered in the experiments.}
    \label{fig:circuit_architectures}
\end{figure*}

\section{Experimental Setup}
\label{sec:appendix_experiment_setup}
In this section, we describe the details of datasets that are utilized in the experiments. 

\subsection{Detailed Experiment Scenarios}
\label{appendix:ExpDetails}

\textbf{Sinusoidal Function Approximation.}
In this experiment, we use a synthetic dataset generated from a sinusoidal target function. Input points are uniformly sampled over a fixed interval, and target values are obtained directly from the analytic function. This dataset provides a controlled benchmark for evaluating the model's ability to learn oscillatory behavior and frequency-dependent patterns. The proposed model was trained using the Adam optimizer with the default PyTorch settings. The training was performed for 1000 epochs, where each epoch iterates over all mini-batches in the training dataloader. The loss function used for optimization was the smooth $L1$ loss.

\textbf{Supervised PDE Regression.} Both settings regress PDEBench solutions \citep{takamoto2022pdebench} and scale the coordinates by $\pi$ before encoding (Table~\ref{tab:pdebench_setup}). The single-upload setting uses the first sample of the advection data with $\beta=0.4$ and of the Burgers data with $\nu=0.001$, which share their initial condition, each on its full $201\times1024$ grid with the raw field $u$. Every circuit uses 14 qubits and encodes the coordinates only in its first layer, so additional layers add parameters but no frequencies, and every template is depth-scaled to at least the $152$ parameters of the larger proposed variant. Each circuit is trained with Adam at learning rate $0.015$, without clipping or decay, for 6 epochs at batch size 64 on a random 80\% of the grid ($15{,}432$ steps), and we report the MSE over the full grid divided by the variance of $u$, $1.70\times10^{-3}$ for advection and $1.01\times10^{-3}$ for Burgers, so a constant predictor scores $1$. The variant with trainable entangling strengths replaces each $CZ$ by a controlled-phase gate initialized at $\pi$, where it equals $CZ$, and all circuits share the same 3 seeds. The re-uploading setting uses advection only. There, every circuit uses 7 qubits and re-encodes the input between blocks, on a $48\times24$ space-time grid with an 80/20 train-test split and targets standardized to zero mean and unit variance, uses the readout $a\langle Z_0\rangle+b$, and is trained with Adam for 30 epochs at batch size 64 over 8 seeds. Each model uses its best learning rate from $\{0.005,0.02,0.05,0.1,0.2\}$. The noise study retrains 5 of the 8 seeds with the same procedure, requires each to reproduce its noiseless test error, and averages finite-shot readout over 10 draws per seed.

\begin{table}[h]
\centering
\scriptsize
\caption{PDEBench regression setup.}
\label{tab:pdebench_setup}
\renewcommand{\arraystretch}{1.15}
\resizebox{0.8\linewidth}{!}{
\begin{tabular}{c|c|l}
\Xhline{3\arrayrulewidth}
\textbf{Component} & \textbf{Notation} & \multicolumn{1}{c}{\textbf{Equation}} \\
\Xhline{1\arrayrulewidth}
Advection equation & $\beta$ & $\displaystyle \partial_t u(x,t) + \beta\,\partial_x u(x,t) = 0$ \\
Burgers equation & $\nu$ & $\displaystyle \partial_t u(x,t) + \partial_x\bigl(u^{2}(x,t)/2\bigr) = (\nu/\pi)\,\partial_{xx} u(x,t)$ \\
Input encoding & $(x,t)$ & $\displaystyle q_1\leftarrow x,\quad q_2\leftarrow t,\quad q_3\leftarrow x,\quad q_4\leftarrow t,\quad \cdots$ \\
Model mapping & $f_\theta$ & $\displaystyle f_\theta:(x,t)\mapsto \hat{u}_\theta(x,t)$ \\
\midrule
\textbf{Training objective} & $\mathcal{L}$ & $\displaystyle \min_{\theta} \frac{1}{N}\sum_{i=1}^{N} \left| \hat{u}_\theta(x_i,t_i) - u(x_i,t_i) \right|^2$ \\
\Xhline{3\arrayrulewidth}
\end{tabular}}
\end{table}

\textbf{Physics-Informed Wave Equation.} We solve $\partial_{tt}E=\partial_{xx}E$ on $[-1,1]\times[0,1]$ with $E(\pm1,t)=0$, $E(x,0)=\sin(\pi x)$ and $\partial_tE(x,0)=0$, whose exact solution is $E=\sin(\pi x)\cos(\pi t)$, with the objective of Table~\ref{tab:wave_setup}. Every circuit encodes $x$ and $t$ once with $R_X(\pi\,\cdot)$, one input per qubit, so the solution has frequency $1$ in each input, and the proposed circuit is the $R\,C\,S\,R\,C\,R$ block of Table~\ref{tab:ablation} on 2 qubits, with a single $CZ$ gate as $C$. Every circuit is trained for 300 steps with Adam at a constant learning rate of $0.015$ after gradient clipping at norm $1$, resampling $1000$ collocation, $200$ boundary and $200$ initial points at every step, over 3 shared seeds. We score the relative $L^2$ error against the exact solution on a $64\times32$ grid. Rather than matching one parameter budget, we report each template with one and two layers, and the Strongly Entangling layer with one and three, and list every parameter count in Table~\ref{tab:wave_app}. Every template with one layer outputs the zero field, which satisfies the equation and the boundary conditions but leaves the initial-condition loss at $0.51$. With more parameters, templates can exceed our accuracy, and in earlier runs with 56 parameters and 500 steps, Circuit 18 reached $1.5\times10^{-6}$.
\begin{table}[h]
\centering
\scriptsize
\caption{Wave-equation training objective.}
\label{tab:wave_setup}
\renewcommand{\arraystretch}{1.15}
\resizebox{0.8\linewidth}{!}{
\begin{tabular}{c|c|l}
\Xhline{3\arrayrulewidth}
\textbf{Component} & \textbf{Notation} & \multicolumn{1}{c}{\textbf{Equation}} \\
\Xhline{1\arrayrulewidth}
Governing equation & $E_\theta(x,t)$ &
$\displaystyle \partial_{tt}E_\theta-\partial_{xx}E_\theta=0,\quad x\in[-1,1],\ t\in[0,1]$ \\
PDE residual & $\mathcal{L}_{\mathrm{PDE}}$ &
$\displaystyle \frac{1}{N_c}\sum_{i=1}^{N_c}
\left|\partial_{tt}E_\theta(x_i,t_i)-\partial_{xx}E_\theta(x_i,t_i)\right|^2$ \\
Boundary condition & $\mathcal{L}_{\mathrm{BC}}$ &
$\displaystyle E_\theta(-1,t)=0,\quad E_\theta(1,t)=0$ \\
Initial condition & $\mathcal{L}_{\mathrm{IC}}$ &
$\displaystyle E_\theta(x,0)=\sin(\pi x),\quad \partial_t E_\theta(x,0)=0$ \\
\midrule
\textbf{Total objective} & $\mathcal{L}$ &
$\displaystyle \mathcal{L}_{\mathrm{PDE}}+10\,\mathcal{L}_{\mathrm{BC}}+10\,\mathcal{L}_{\mathrm{IC}}$ \\
\Xhline{3\arrayrulewidth}
\end{tabular}}
\end{table}
\begin{table}[h]
\centering
\scriptsize
\caption{Physics-informed wave equation with one encoding layer, mean relative $L^2$ error over 3 seeds. Every template uses 6 qubits.}
\label{tab:wave_app}
\setlength{\tabcolsep}{4pt}
\begin{tabular}{lccc|lccc}
\Xhline{3\arrayrulewidth}
\multicolumn{4}{c|}{One layer} & \multicolumn{4}{c}{More layers, 6 qubits} \\
Circuit & Qubits & $|\theta|$ & Rel.\ $L^2$ & Circuit & $L$ & $|\theta|$ & Rel.\ $L^2$ \\
\Xhline{1\arrayrulewidth}
HEA circular & 6 & 20 & $1.0004$ & Circuit 16 & 2 & 36 & $8.1\times10^{-5}$ \\
Strongly Ent. & 6 & 20 & $1.0004$ & Circuit 17 & 2 & 36 & $3.1\times10^{-5}$ \\
Circuit 15 & 6 & 14 & $1.0003$ & Circuit 18 & 2 & 38 & $3.6\times10^{-4}$ \\
Circuit 16 & 6 & 19 & $1.0006$ & Circuit 19 & 2 & 38 & $3.0\times10^{-5}$ \\
Circuit 17 & 6 & 19 & $1.0005$ & YZY & 2 & 38 & $0.336$ \\
Circuit 18 & 6 & 20 & $1.0006$ & YZY (ent.) & 2 & 38 & $0.336$ \\
Circuit 19 & 6 & 20 & $1.0005$ & HEA circular & 2 & 38 & $0.455$ \\
YZY, YZY (ent.) & 6 & 20 & $1.0006$ & Circuit 15 & 2 & 26 & $1.000$ \\
Basic Ent. & 6 & 8 & $1.0004$ & Basic Ent. & 2 & 14 & $1.000$ \\
\cline{1-4}
\textbf{Proposed} & 2 & 20 & $2.7\times10^{-5}$ & Strongly Ent. & 3 & 56 & $9.9\times10^{-3}$ \\
\Xhline{3\arrayrulewidth}
\end{tabular}
\end{table}

\subsection{Computational Resources}
Table~\ref{tab:computational_resources} lists the hardware and software we used.
\begin{table}[h]
    \centering
    \small
    \caption{Computational resources and software libraries used in the experiments.}
    \renewcommand{\arraystretch}{1.0}
    \begin{tabular}{c|c}
    \Xhline{3\arrayrulewidth}
    \textbf{Resource / Library} & \textbf{Specification} \\
    \hline
    \hline
    CPU & AMD Ryzen 7 9700X 8-Core Processor, 3.80 GHz \\
    RAM & 64 GB \\
    GPU & NVIDIA GeForce RTX 5080, 16 GB \\
    Operating System & 64-bit Windows, x64-based processor \\
    Storage & 1.82 TB SSD \\
    Quantum ML Library & PennyLane \\
    ML / Auto-differentiation Backend & JAX \\
    Optimizer Library & Optax \\
    Numerical Libraries & NumPy, Pandas \\
    Visualization Library & Matplotlib \\
    \Xhline{3\arrayrulewidth}
    \end{tabular}
    \label{tab:computational_resources}
\end{table}

\subsection{Placement Ablation}
\label{app:ablation}
\textbf{Circuits.} All circuits act on 14 qubits, encode $x$ on odd-indexed and $t$ on even-indexed qubits with $R_{X}(\pi\,\cdot)$, and measure $\hat{Z}$ on qubit $1$, which encodes $x$, except Leaf readout, which measures qubit $2$, which encodes $t$. Each $R$ applies $R_{Z}R_{Y}R_{Z}$ to every qubit ($42$ angles), $C$ applies $CZ$ between qubit $1$ and every other qubit, and $P$ applies $CZ$ along the path $2$--$1$--$3$--$4$--$\cdots$--$14$, on which qubit $1$ has one $x$-neighbour and one $t$-neighbour. The proposed circuit $R\,C\,S\,R\,C\,R$ is $W_{R}=CR$ and $W_{L}=RCR$ in operator order, so the rotation between the encoder and the star of $W_{L}$ is Align and the rotation after that star is Expose. No Expose, $R\,C\,S\,R\,R\,C$, therefore places the star directly before the readout, and No Align, $R\,C\,S\,C\,R\,R$, places it directly after the encoder.

\textbf{Counting reachable pairs.} With one encoding layer, $f(x,t)=\sum c(\omega_{x},\omega_{t})\,e^{i\pi(\omega_{x}x+\omega_{t}t)}$ with $|\omega_{x}|,|\omega_{t}|\le7$, which gives $225$ coefficients. By Corollary~\ref{cor:activation}, whose proof holds verbatim for two inputs (Appendix~\ref{app:activation}), each coefficient is, as a function of the angles, either identically zero or nonzero at almost every setting. We therefore draw all $126$ angles uniformly at random, evaluate $f$ on a $16\times16$ grid of encoding angles over one period and obtain every coefficient exactly by a two-dimensional discrete Fourier transform, which does not alias because $16>2\cdot7+1$. A pair is reachable if its coefficient is nonzero, and $(k_{x},k_{t})$ is the largest $|\omega_{x}|$ and $|\omega_{t}|$ over the reachable pairs. Over 20 draws, reachable coefficients exceed $10^{-10}$ in magnitude at every draw and the others stay below $10^{-15}$, so a single draw gives the exact count, and no data or training enters it.

\textbf{Predictions.} Without Expose, $U_{\mathrm{ent}}$ commutes with $\hat{Z}_{n}$, so $O$ acts on $n$ alone and reaches only $\omega_{t}=0$ and $|\omega_{x}|\le1$. The leaf readout and the path have $\deg(n)=1$ and $\deg(n)=2$, and the bound $|\omega_{x}|\le|L\cap V_{x}|$, $|\omega_{t}|\le|L\cap V_{t}|$ of Appendix~\ref{app:activation} gives $(1,1)$ and $(2,1)$, with every pair inside reachable at almost every parameter setting. Without Align, the merged rotations tilt the readout into $M_{1}=m_{x}\hat{X}_{n}+m_{y}\hat{Y}_{n}+m_{z}\hat{Z}_{n}$ and the star gives $O=m_{z}\hat{Z}_{n}+(m_{x}\hat{X}_{n}+m_{y}\hat{Y}_{n})\hat{Z}_{\mathcal{N}(n)}$. Each $\hat{Z}_{k}=(E_{k}^{+}+E_{k}^{-})/2$ adds $\pm1$ to the frequency of its input and all seven $t$-qubits are leaves, so the strings through the leaves reach every $\omega_{x}$ but only odd $\omega_{t}$, and $\hat{Z}_{n}$ adds $(\pm1,0)$, which leaves $15\times8+2=122$ pairs.

\textbf{Best possible MSE.} The output of a variant is $w\,f(x,t)+b$, where $f$ contains only its reachable pairs, so its error can never fall below the least-squares fit of $u$ on the constant and on $\cos(\pi(\omega_{x}x+\omega_{t}t))$ and $\sin(\pi(\omega_{x}x+\omega_{t}t))$ over those pairs. We compute this fit on the full grid of Table~\ref{tab:advection14} and divide the mean squared residual by the variance of $u$. It needs no training and bounds from below the error of every parameter setting, whatever the training budget.

\textbf{Training.} We train each variant with the settings of Table~\ref{tab:advection14} on the grid subsampled with stride $8$ ($25{,}728$ points), which gives $2{,}412$ optimizer steps against $15{,}432$, with seeds $0$ to $4$ shared across variants. At this budget the proposed circuit reaches $0.628\pm0.212$ and the four variants stay between $0.957$ and $1.022$, with Holm-corrected Welch $p=0.082$ for every variant, limited by the spread of the proposed circuit across seeds. The variant without Expose stays at the constant predictor on every seed, between $1.000$ and $1.058$, as its output has a single frequency in $x$ and none in $t$. Every trained error lies above the bound of its variant.

\newpage
\section{Additional Experiment Results}
\subsection{Case Study of $O$ and $\rho$}
We provide additional ablation results to further support the projection-based interpretation of Proposition~\ref{proposition:multi-qubit_Projected}. In the main text, we argued that the encoder determines the accessible frequency set, while the actual Fourier amplitudes depend on the projections of the effective observable $O$ and effective state $\rho$ onto the corresponding invariant planes. Tables~\ref{tab:tableeffectofO} and \ref{tab:tableeffectofRho} provide detailed circuit-level evidence for this mechanism.

Table~\ref{tab:tableeffectofO} examines the effect of entanglement placement on the effective observable $O$. Circuit (i) follows the proposed structure, where local rotations are placed before and after the entangling layer. This allows the initially local measurement operator to be first rotated into a noncommuting Pauli direction, then dressed by entanglement into mixed multi-qubit Pauli strings, and finally adjusted by the remaining local rotations. As a result, the effective observable can acquire nonzero overlap with higher-frequency invariant planes. This explains why Circuit (i) accurately reconstructs $f_{\mathrm{target},1}(x)$.

In contrast, Circuits (ii)--(iv) modify the placement of the local rotations and entangling blocks. Although these circuits may still contain entanglement, their ordering does not effectively promote the measured observable into the mixed Pauli-string components required for higher-frequency activation. Consequently, the projection of $O$ onto the relevant invariant subspaces is weakened, and the learned function underfits the higher-frequency oscillations. This confirms that entanglement alone is not sufficient; its placement relative to local rotations is essential for shaping the effective observable.

Table~\ref{tab:tableeffectofRho} further investigates the role of the effective state $\rho=W_R\rho_0W_R^\dagger$. Since each Fourier coefficient depends on the joint projection of both $\rho$ and $O$, a target frequency can be suppressed if either side has insufficient overlap with the corresponding invariant plane. Circuit (i) jointly dresses both the state and observable, enabling both operators to support the mixed Pauli strings associated with $f_{\mathrm{target},2}(x)$. This leads to the best reconstruction performance.

The remaining circuits break this mechanism in different ways. Circuit (ii) mainly shapes the observable side but leaves the state insufficiently prepared. Circuit (iii) introduces entanglement on the state side without first generating the necessary local noncommuting components, limiting its ability to populate the target invariant subspaces. Circuit (iv) partially restores observable-side dressing, but the state remains unentangled, which limits the size of its projection onto the target planes. These results show that higher-frequency learning requires coordinated design of both $O$ and $\rho$, rather than arbitrary entanglement placement.

Overall, the appendix results reinforce the main conclusion: the success of the proposed architecture comes from its ability to align both the effective observable and effective state with the invariant frequency subspaces induced by the data-encoding generator. This provides a circuit-level explanation for why the proposed entanglement placement activates higher Fourier components more reliably than heuristic alternatives.
\begin{table*}[h]
\centering
\small
\caption{Effect of entanglement placement and $O$ on PQC. }
\renewcommand{\arraystretch}{1.0}
\resizebox{0.99\linewidth}{!}{
\begin{tabular}{c|c||c|c}
\toprule[1pt]
\textbf{Circuit (i)} & $f_{\rm target,1}(x)$ & Circuit (ii) & $f_{\rm target,1}(x)$
\\ 
\midrule[1pt]
\includegraphics[width=0.26\linewidth]{img/sandwichcircuit.pdf}     & \includegraphics[width=0.22\linewidth]{img/sandwich_PQC_target.png}  &   
\includegraphics[width=0.26\linewidth]{img/New_left.pdf}     & \includegraphics[width=0.22\linewidth]{img/New_left.png}     \\ 

\toprule[1pt]
Circuit (iii) & $f_{\rm target,1}(x)$ & Circuit (iv) & $f_{\rm target,1}(x)$
\\
\midrule[1pt]

\includegraphics[width=0.27\linewidth]{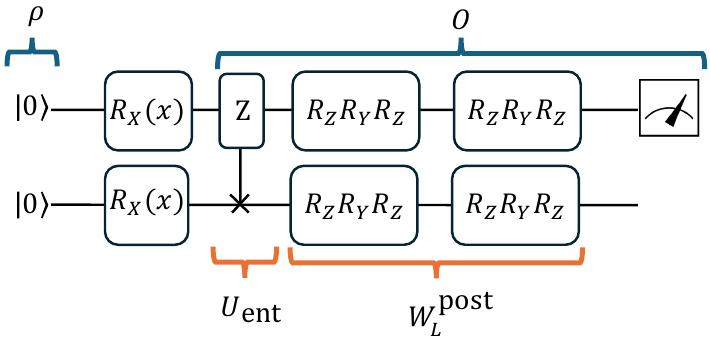}     & \includegraphics[width=0.22\linewidth]{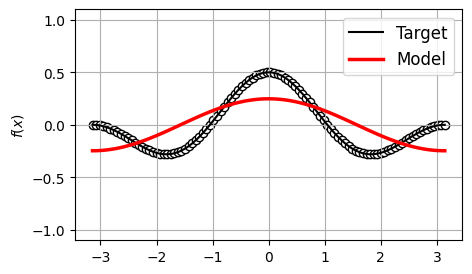}    
&
\includegraphics[width=0.27\linewidth]{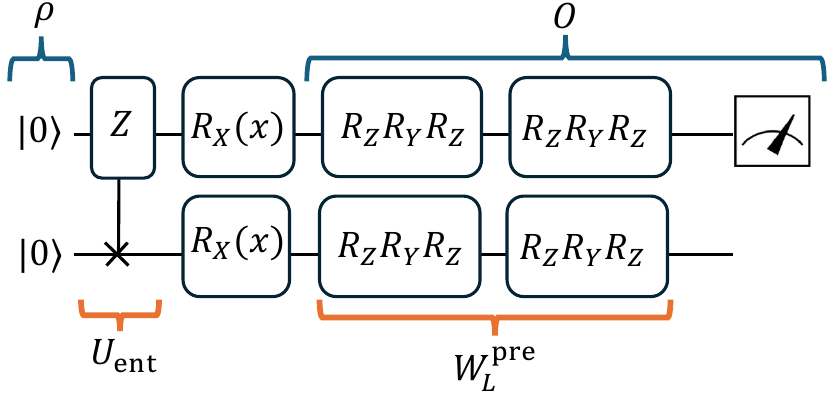}     & \includegraphics[width=0.22\linewidth]{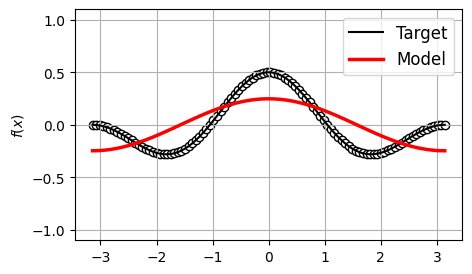}     \\

\bottomrule[1pt]
\end{tabular}
}
\label{tab:tableeffectofO}
\end{table*}

\begin{table*}[h]
\centering
\small
\caption{Effect of changing initial quantum state on PQC. }
\renewcommand{\arraystretch}{1.0}
\resizebox{0.99\linewidth}{!}{
\begin{tabular}{c|c||c|c}
\toprule[1pt]
Circuit (i) (Proposed) & $f_{\rm target,2}(x)$ & Circuit (ii) & $f_{\rm target,2}(x)$
\\ 
\midrule[1pt]
\includegraphics[width=0.27\linewidth]{img/PreRotEntRho.pdf}     & \includegraphics[width=0.22\linewidth]{img/FourQubitEntangledPreRotation.png}    &   
\includegraphics[width=0.22\linewidth]{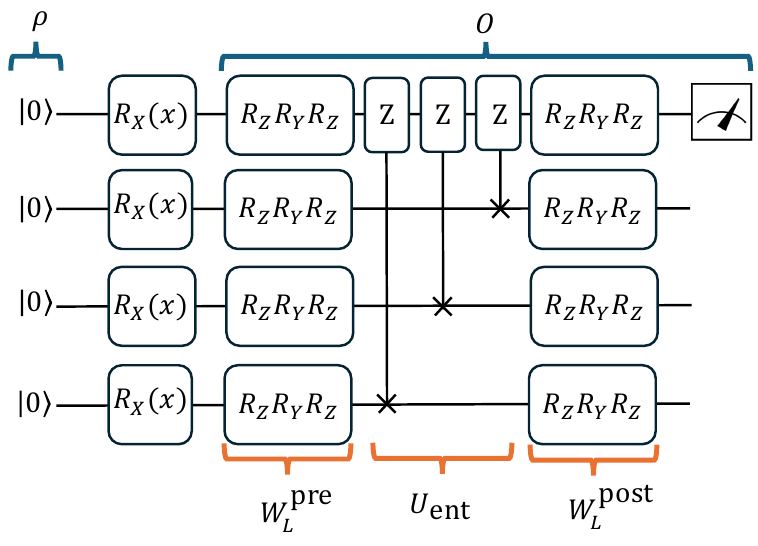}     & \includegraphics[width=0.22\linewidth]{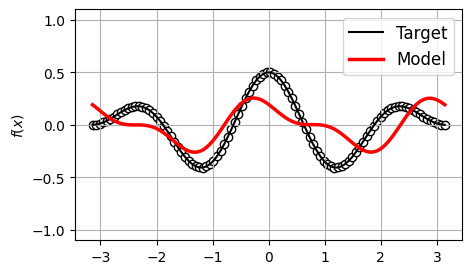}
    \\   

\toprule[1pt]
Circuit (iii) & $f_{\rm target,2}(x)$ & Circuit (iv) & $f_{\rm target,2}(x)$
\\
\midrule[1pt]

\includegraphics[width=0.28\linewidth]{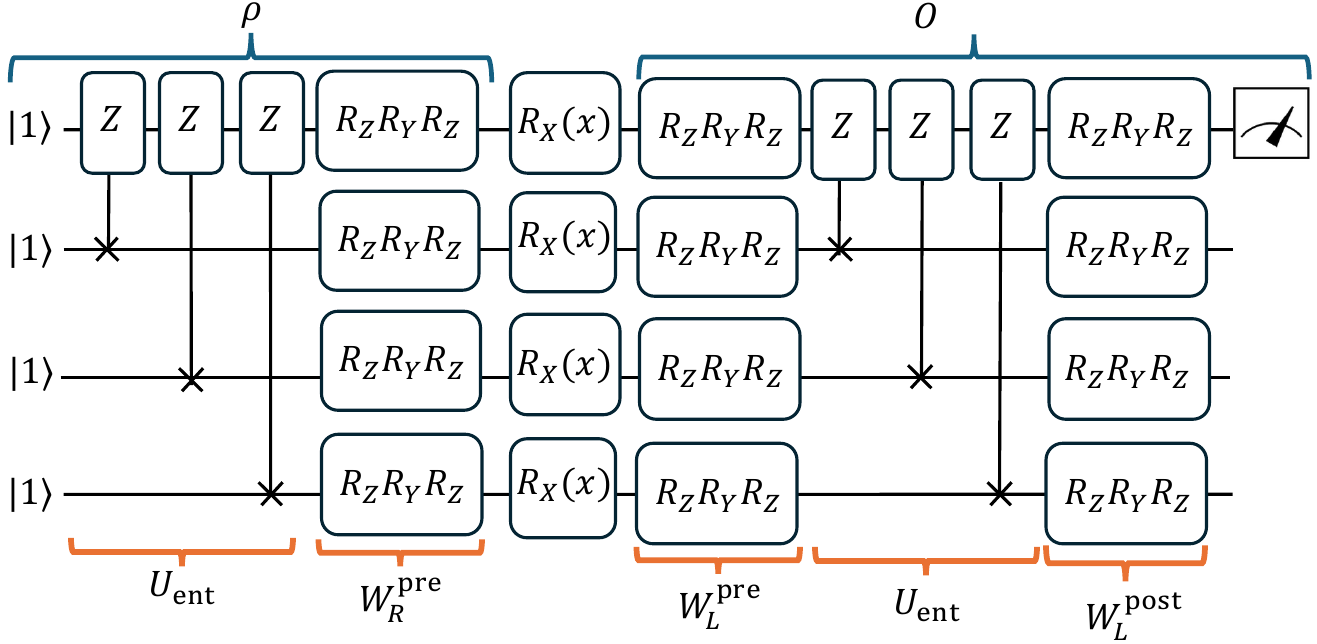}     & \includegraphics[width=0.22\linewidth]{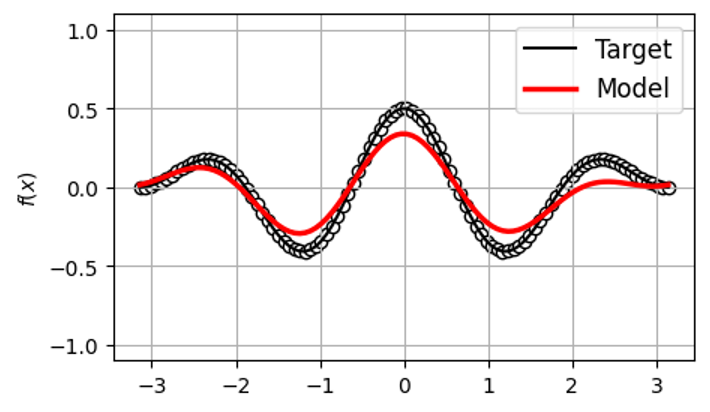} &

\includegraphics[width=0.26\linewidth]{img/PreRotRho.pdf}     & \includegraphics[width=0.22\linewidth]{img/FourQubitSimplePreRotation.png} \\

\bottomrule[1pt]
\end{tabular}
\label{tab:tableeffectofRho}
}
\end{table*}

\subsection{Results with Data Re-uploading}
\label{app:reupload}
Our analysis covers one encoding layer, so we treat data re-uploading as a test beyond it. The setting favors the templates, since every added layer re-encodes the input (Appendix~\ref{appendix:DataReupload}). We run it on advection at 7 qubits, which keeps an exact density-matrix noise study tractable, with every circuit depth-scaled to about 330 parameters (Appendix~\ref{appendix:ExpDetails}). Both proposed variants lead every matched template sharing our $R_X(\pi x)$ encoding (Holm $p<0.001$), the better reaching $0.104$ against $0.169$ for the nearest, Circuit 17, while using fewer parameters and $48$ two-qubit gates against the $96$--$329$ of the entangling templates (Table~\ref{tab:advection}).

\begin{table}[h]
\caption{Supervised advection regression at 7 qubits with re-uploading, test MSE over 8 seeds. \textbf{Bold} and \textcolor{blue}{\textbf{blue bold}} mark the best and second-best result.}
\label{tab:advection}
\centering
\scriptsize
\setlength{\tabcolsep}{4pt}
\begin{tabular}{l|cc|c||l|cc|c}
\Xhline{3\arrayrulewidth}
\textbf{Method} & $|\theta|$ & 2q & MSE $\downarrow$ & \textbf{Method} & $|\theta|$ & 2q & MSE $\downarrow$ \\
\Xhline{1\arrayrulewidth}
Circuit 18 & 338 & 112 & $0.1699 \pm 0.0181$ & Circuit 19 & 338 & 112 & $0.1762 \pm 0.0216$ \\
Circuit 17 & 322 & 96 & $0.1694 \pm 0.0263$ & Circuit 16 & 322 & 96 & $0.1833 \pm 0.0168$ \\  HEA circular & 338 & 112 & $0.4378 \pm 0.0525$ & 
YZY (ent.) & 338 & 96 & $0.3336 \pm 0.0292$   \\
StronglyEnt. & 338 & 112 & $0.1835 \pm 0.0106$ & BasicEnt. & 331 & 329 & $0.5170 \pm 0.0146$ \\
\Xhline{1\arrayrulewidth}
\textbf{Proposed} & 254 & 48 & \textcolor{blue}{$\mathbf{0.1158 \pm 0.0189}$} & \textbf{Proposed ($\gamma$)} & 302 & 48 & $\mathbf{0.1041 \pm 0.0073}$ \\
\Xhline{3\arrayrulewidth}
\end{tabular}
\end{table}

\textbf{Noise at 7 qubits.} We evaluate the noiselessly trained advection models of Table~\ref{tab:advection} under three noise models, after retraining each from its original seed to reproduce its noiseless error exactly. Amplitude and phase damping of strength $\lambda$ acts on every qubit after every circuit moment, so exposure grows with depth. Depolarizing noise of strength $\epsilon$ acts after every two-qubit gate and at $\epsilon/10$ after every single-qubit gate, so exposure grows with gate count. Finite-shot readout estimates each expectation from $S$ shots. In Table~\ref{tab:advnoise}, the proposed circuit, which has the fewest two-qubit gates, has the lowest mean error at $\lambda=10^{-3}$, at both gate-noise levels and at $S=512$, while the chain templates, Circuits 16 and 17, lead under stronger damping and at $S=128$.

\begin{table}[h]
\centering
\caption{Test MSE under noise of the advection models of Table~\ref{tab:advection}, mean over 5 seeds. \textbf{Bold} and \textcolor{blue}{\textbf{blue bold}} mark the best and second-best per column.}
\label{tab:advnoise}
\fontsize{6.5}{7.0}\selectfont
\setlength{\tabcolsep}{2.5pt}
\renewcommand{\arraystretch}{0.90}
\resizebox{0.85\textwidth}{!}{%
\begin{tabular}{l|ccc|c|cc|cc|cc}
\Xhline{3\arrayrulewidth}
& & & & & \multicolumn{2}{c|}{Damping $\lambda$} & \multicolumn{2}{c|}{Gate noise $\epsilon$} & \multicolumn{2}{c}{Shots $S$} \\
\textbf{Method} & $|\theta|$ & 2q & Dep. & Clean & $10^{-3}$ & $3{\times}10^{-3}$ & $3{\times}10^{-3}$ & $10^{-2}$ & 128 & 512 \\
\Xhline{1\arrayrulewidth}
Circuit 18 & 338 & 112 & 174 & $0.1757$ & $0.2551$ & $0.5216$ & $0.2122$ & $0.3829$ & $0.2523$ & $0.1966$ \\
Circuit 19 & 338 & 112 & 174 & $0.1857$ & $0.2809$ & $0.5880$ & $0.2246$ & $0.4187$ & $0.3020$ & $0.2130$ \\
Circuit 17 & 322 & 96 & 98 & $0.1743$ & $0.2028$ & $\mathbf{0.3280}$ & $0.1935$ & $0.3031$ & $\mathbf{0.2422}$ & $0.1919$ \\
Circuit 16 & 322 & 96 & 98 & $0.1837$ & $0.2134$ & \textcolor{blue}{$\mathbf{0.3468}$} & $0.2024$ & $0.3129$ & \textcolor{blue}{$\mathbf{0.2435}$} & $0.1990$ \\
StronglyEnt. & 338 & 112 & 92 & $0.1845$ & $0.3076$ & $0.6529$ & $0.2785$ & $0.6202$ & $0.4749$ & $0.2595$ \\
YZY (ent.) & 338 & 96 & 114 & $0.3349$ & $0.4391$ & $0.7393$ & $0.3995$ & $0.6468$ & $0.7325$ & $0.4291$ \\
HEA circular & 338 & 112 & 190 & $0.4154$ & $0.6470$ & $0.9626$ & $0.5014$ & $0.8092$ & $0.8090$ & $0.5181$ \\
BasicEnt. & 331 & 329 & 423 & $0.5179$ & $0.9198$ & $1.0191$ & $0.7691$ & $1.0022$ & $0.6810$ & $0.5537$ \\
\Xhline{1\arrayrulewidth}
Proposed & 254 & 48 & 64 & \textcolor{blue}{$\mathbf{0.1203}$} & \textcolor{blue}{$\mathbf{0.1811}$} & $0.4182$ & \textcolor{blue}{$\mathbf{0.1477}$} & \textcolor{blue}{$\mathbf{0.2986}$} & $0.3619$ & \textcolor{blue}{$\mathbf{0.1804}$} \\
\textbf{Proposed ($\gamma$)} & 302 & 48 & 64 & $\mathbf{0.1077}$ & $\mathbf{0.1705}$ & $0.4186$ & $\mathbf{0.1358}$ & $\mathbf{0.2918}$ & $0.3449$ & $\mathbf{0.1692}$ \\
\Xhline{3\arrayrulewidth}
\end{tabular}%
}
\end{table}

\subsection{Noise at 14 Qubits}
\label{app:noise14}
An exact density matrix at 14 qubits has $4^{14}$ entries per state, so we simulate gate noise and damping by quantum trajectories, with the noise models of the 7-qubit study in Appendix~\ref{app:reupload}. Gate noise is a Pauli channel, which the trajectories unravel exactly, and damping uses the Monte-Carlo wavefunction method. Against an exact simulator at 5 qubits with 800 trajectories, every tested setting agrees within $1.3$ standard errors. The MSE of a trajectory-averaged prediction carries a positive bias equal to the variance of that average, and we subtract its sample estimate. We compare against HEA, the strongest template in Table~\ref{tab:advection14}, because trajectory simulation of the Strongly Entangling layer exceeded our compute budget, and HEA's damping cells stay unresolved at 512 trajectories, since its readout scale of $18.5$ amplifies the trajectory variance. Circuits 16--19 and YZY (ent.) ignore the readout, with $|w|\le0.09$, so noise cannot change their error. Finite-shot readout adds a variance of about $w^{2}(1-\langle Z_{0}\rangle^{2})/S$ to each prediction, and the shot errors of different circuits scale with $w^{2}$ to within $1\%$.

\subsection{Advection Equation}
\label{appendix:advection_equation}

We further evaluate the proposed circuit on the one-dimensional advection equation, which describes transport-dominated dynamics where an initial profile propagates through space at a constant speed without changing its shape. The governing equation is
\begin{equation}
    \partial_t u(x,t) + \beta \partial_x u(x,t)=0,
\end{equation}
where $\beta$ denotes the advection speed. This task is useful for testing whether the model can preserve coherent spatiotemporal wave patterns, since the solution appears as diagonal stripe structures in the space-time domain. We form supervised samples $\{((x_i,t_i),u(x_i,t_i))\}_{i=1}^{N}$ from the PDEBench space-time grid and train the quantum model $f_\theta:(x,t)\mapsto \hat{u}_\theta(x,t)$ by minimizing the mean-squared error
\begin{equation}
    \min_{\theta}
    \frac{1}{N}\sum_{i=1}^{N}
    \left|
    \hat{u}_\theta(x_i,t_i)-u(x_i,t_i)
    \right|^2 .
\end{equation}

Figures~\ref{fig:Advection_equation}, \ref{fig:Advection_equation_02} and \ref{fig:Advection_equation_04} show the reconstruction results for $\beta=0.1$, $0.2$ and $0.4$. Across all advection speeds, the Basic Entangler and Strongly Entangling baselines fail to accurately reproduce the diagonal transport patterns, producing distorted or over-smoothed fields. In contrast, the proposed circuit closely matches the ground-truth stripe patterns for all tested values of $\beta$, indicating that the proposed entanglement design can effectively represent transport-dominated spatiotemporal fields.

\begin{figure}[h]
    \centering
    \small
    \includegraphics[width=0.20\linewidth]{img/Advection/advection_04_qubit14_layer1_colormap.png}
    \begin{tabular}{cccc}
     \includegraphics[width=0.17\linewidth]{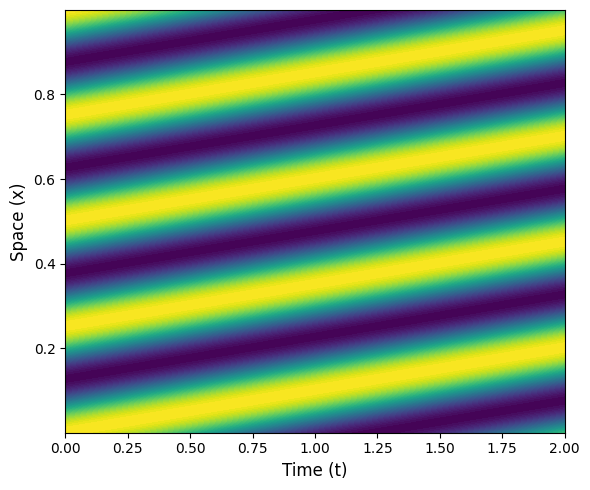} & 
     \includegraphics[width=0.17\linewidth]{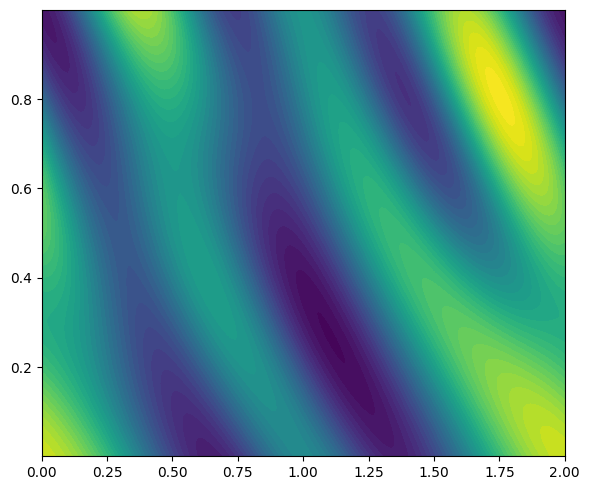} &
     \includegraphics[width=0.17\linewidth]{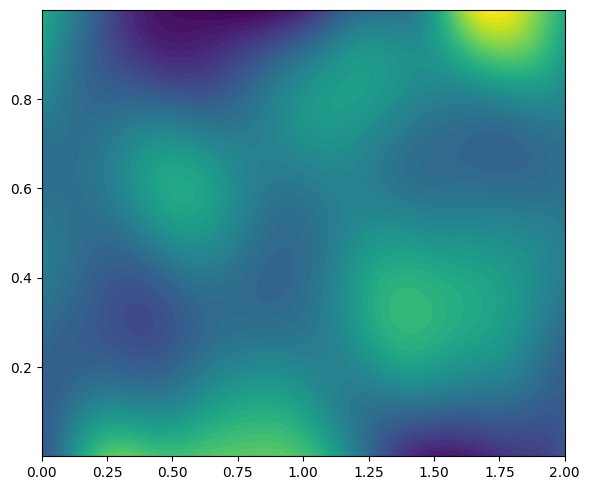} &
     \includegraphics[width=0.17\linewidth]{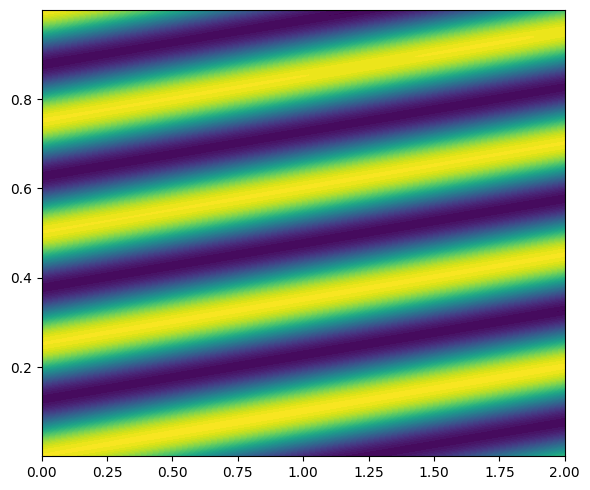}\\
     Ground Truth & Basic Entangler & Strongly Entangling & Proposed 
    \end{tabular}
    \caption{\textbf{Prediction results of the algorithms}. Advection equation ($\beta=0.1$).}
    \label{fig:Advection_equation}
\end{figure}

\begin{figure}[h]
    \centering
    \small
    \includegraphics[width=0.20\linewidth]{img/Advection/advection_04_qubit14_layer1_colormap.png}
    \begin{tabular}{cccc}
     \includegraphics[width=0.17\linewidth]{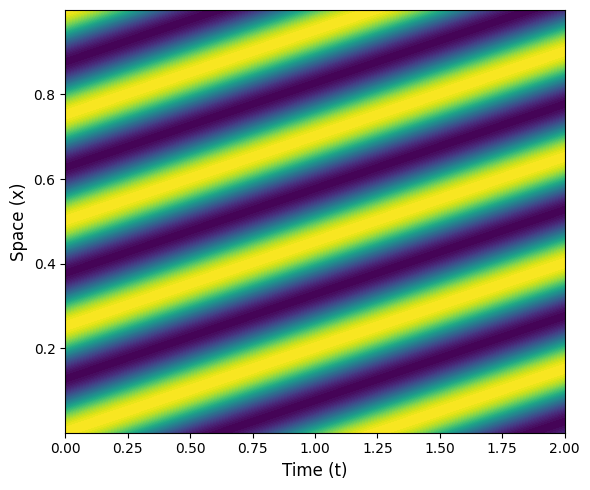} & 
     \includegraphics[width=0.17\linewidth]{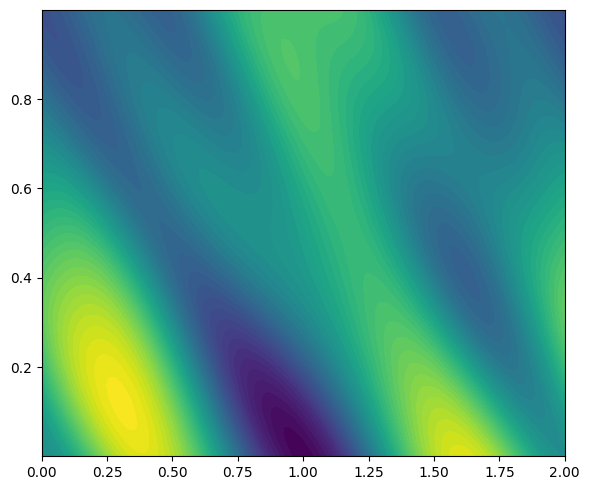} &
     \includegraphics[width=0.17\linewidth]{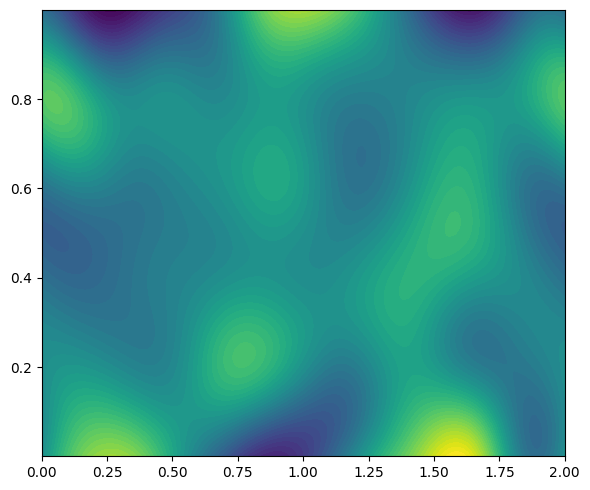} &
     \includegraphics[width=0.17\linewidth]{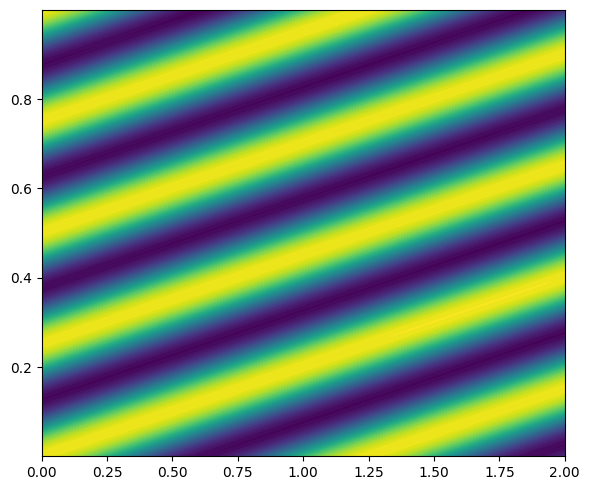}\\
     Ground Truth & Basic Entangler & Strongly Entangling &  Proposed 
    \end{tabular}
    \caption{\textbf{Prediction results of the algorithms}. Advection equation ($\beta=0.2$).}
    \label{fig:Advection_equation_02}
\end{figure}

\subsection{Adaptation to Environment Changes}
\label{app:antenna}
We further test whether a model can learn a family of fields across propagation environments. An antenna phase vector $\boldsymbol{\phi}\in[0,2\pi]^M$ with $M=4$ is concatenated with the space-time coordinates, so that each input is $[x,t,\phi_1,\ldots,\phi_4]$, and the model learns the parametric field $\hat{E}_\theta(x,t;\boldsymbol{\phi})$ over $x,t\in[0,1]$. The target is $E(x,t;\boldsymbol{\phi})=\cos(\pi(t-x))+A(\boldsymbol{\phi})\cos(\pi(t+x)+\psi(\boldsymbol{\phi}))$, where $Ae^{j\psi}=\frac{1}{4}\sum_{m}e^{j\phi_m}$ is the effective array response, and every model is trained with the wave-equation objective of Table~\ref{tab:antenna_adaptation_setup}, with boundary and initial values taken from $E$. Each epoch resamples 2000 collocation points over $(x,t,\boldsymbol{\phi})$, 400 boundary points at $x\in\{0,1\}$ and 400 initial points at $t=0$, and every model is trained for 500 epochs with Adam under an exponential learning-rate decay of $0.99$ every 50 steps after gradient clipping at global norm $1.0$.

\begin{table}[h]
\centering
\scriptsize
\caption{Antenna-configuration field-learning setup. $\mathcal{L}_{\mathrm{PDE}}$ is the residual of $\partial_{tt}\hat{E}_\theta-\partial_{xx}\hat{E}_\theta$, and $\mathcal{L}_{\mathrm{BC}}$, $\mathcal{L}_{\mathrm{IC}}$ match $E$ at $x\in\{0,1\}$ and $E,\partial_tE$ at $t=0$.}
\label{tab:antenna_adaptation_setup}
\renewcommand{\arraystretch}{1.15}
\resizebox{0.75\linewidth}{!}{
\begin{tabular}{c|c|l}
\Xhline{3\arrayrulewidth}
\textbf{Component} & \textbf{Notation} & \multicolumn{1}{c}{\textbf{Equation}} \\
\Xhline{1\arrayrulewidth}
Parametric field & $E_\theta$ &
$\displaystyle E_\theta:(x,t,\boldsymbol{\phi})\mapsto \hat{E}_\theta(x,t;\boldsymbol{\phi}),\quad x,t\in[0,1],\ \boldsymbol{\phi}\in[0,2\pi]^M$ \\
Effective array response & $\Gamma(\boldsymbol{\phi})$ &
$\displaystyle \Gamma(\boldsymbol{\phi})=\frac{1}{M}\sum_{m=1}^{M} e^{j\phi_m}=A(\boldsymbol{\phi})e^{j\psi(\boldsymbol{\phi})}$ \\
Amplitude / phase & $A,\psi$ &
$\displaystyle A(\boldsymbol{\phi})=|\Gamma(\boldsymbol{\phi})|,\quad \psi(\boldsymbol{\phi})=\arg(\Gamma(\boldsymbol{\phi}))$ \\
Analytical field & $E(x,t;\boldsymbol{\phi})$ &
$\displaystyle E(x,t;\boldsymbol{\phi})=\cos\!\big(\pi(t-x)\big)+A(\boldsymbol{\phi})\cos\!\big(\pi(t+x)+\psi(\boldsymbol{\phi})\big)$ \\
\midrule
\textbf{Training objective} & $\mathcal{L}$ &
$\displaystyle \mathcal{L}_{\mathrm{PDE}}+10\,\mathcal{L}_{\mathrm{BC}}+10\,\mathcal{L}_{\mathrm{IC}}$ \\
\Xhline{3\arrayrulewidth}
\end{tabular}}
\end{table}

The Q-PINN uses 6 qubits, encoding one input per qubit, and four re-uploading layers, each applying general single-qubit rotations, a star-graph $CZ$ entangler, the encoding $R_X(fz)$ of the qubit's input $z$ with a trainable frequency $f$, further rotations and a second star-graph entangler, followed by a scaled and shifted $\langle Z_0\rangle$ readout (170 parameters). It is trained at learning rate $0.015$ and retrained over three seeds, one of which reproduces the reported model within seed-level noise. The tanh PINN is a fully connected network with three hidden layers of width 128 (34,049 parameters) that receives $(\pi x,\pi t)$ and the phases in radians, as the Q-PINN does, and is trained at learning rate $10^{-3}$. With inputs normalized to $[-1,1]$ and learning rate $5\times10^{-3}$, the same network reaches $0.11$ and $0.16$ over three seeds. The tanh network of comparable size, with two hidden layers of 10 units (191 parameters), takes normalized inputs and the learning rate with the lowest training loss among $\{10^{-3},5\times10^{-3},1.5\times10^{-2}\}$ over five seeds.

Evaluation uses a $100\times100$ grid in $(x,t)$ for each of six configurations, Config.~1 $[0,\pi/2,\pi/3,\pi/4]$ and Config.~2 $[0,0,0,0]$ of Table~\ref{tab:final_results} together with $[\pi/2,\pi/2,\pi/5,\pi/5]$, $[0,\pi/2,\pi,3\pi/4]$, $[\pi/5,\pi/5,\pi/5,\pi/5]$ and $[0,\pi/7,0,\pi/7]$, and reports the relative $L^2$ error $\|\hat{E}_\theta-E\|_F/\|E\|_F$. On Config.~1 and Config.~2, the Q-PINN reconstructs the field to relative $L^2$ errors of $0.35\pm0.04$ and $0.25\pm0.01$ over three seeds, while the tanh network of comparable size reaches $0.50\pm0.08$ and $0.61\pm0.11$ over five seeds (Holm $p\le0.022$). Over the six configurations the reported Q-PINN averages $0.33$ and is more accurate than the tanh PINN on three of them (average $0.44$).

\begin{table*}[htb]
\centering
\small
\caption{Field reconstruction under two antenna configurations.}
\setlength{\tabcolsep}{1pt}
\resizebox{\linewidth}{!}{
\begin{tabular}{ccc|ccc}
\toprule[1pt]
\multicolumn{3}{c|}{Config.~1, $\boldsymbol{\phi}=[0,\frac{\pi}{2},\frac{\pi}{3},\frac{\pi}{4}]$} & \multicolumn{3}{c}{Config.~2, $\boldsymbol{\phi}=[0,0,0,0]$} \\
Ground Truth & Q-PINN & Classic PINN & Ground Truth & Q-PINN & Classic PINN \\
\midrule
\includegraphics[width=0.17\linewidth]{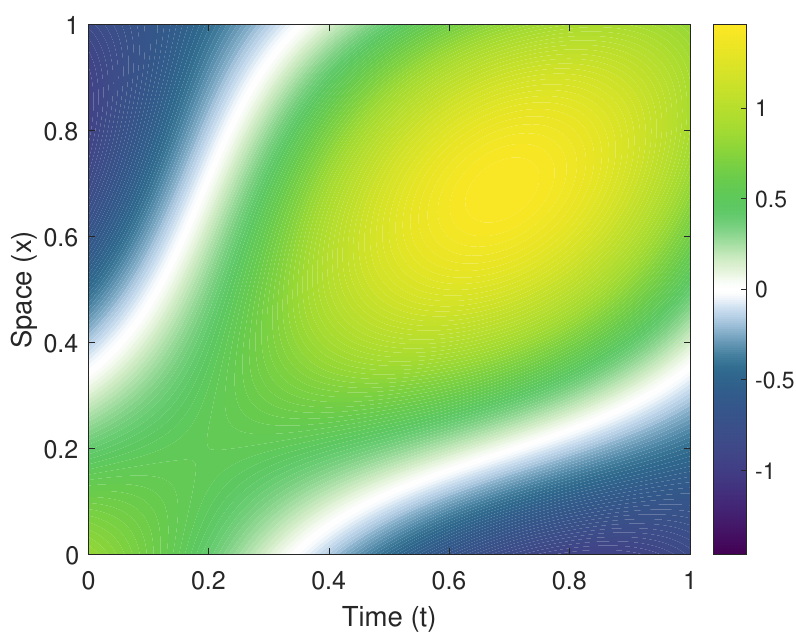} &
\includegraphics[width=0.17\linewidth]{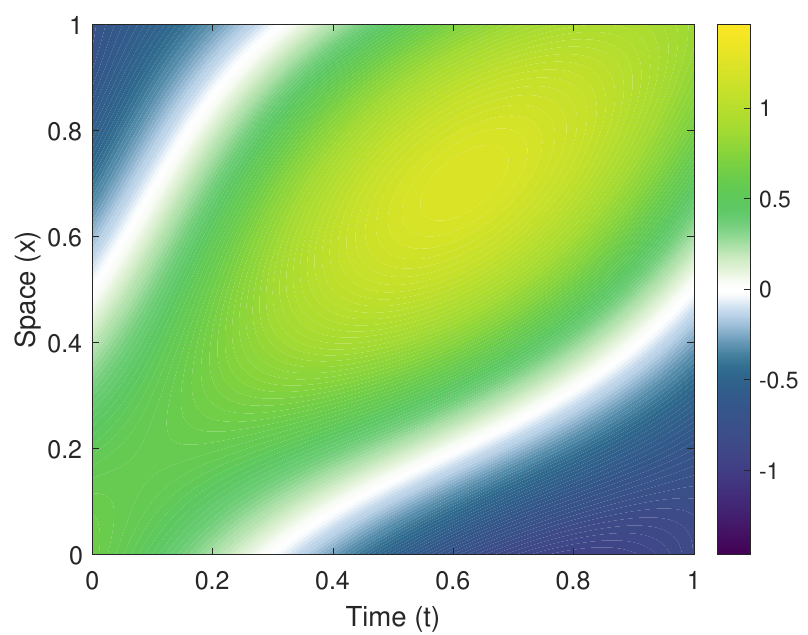} &
\includegraphics[width=0.17\linewidth]{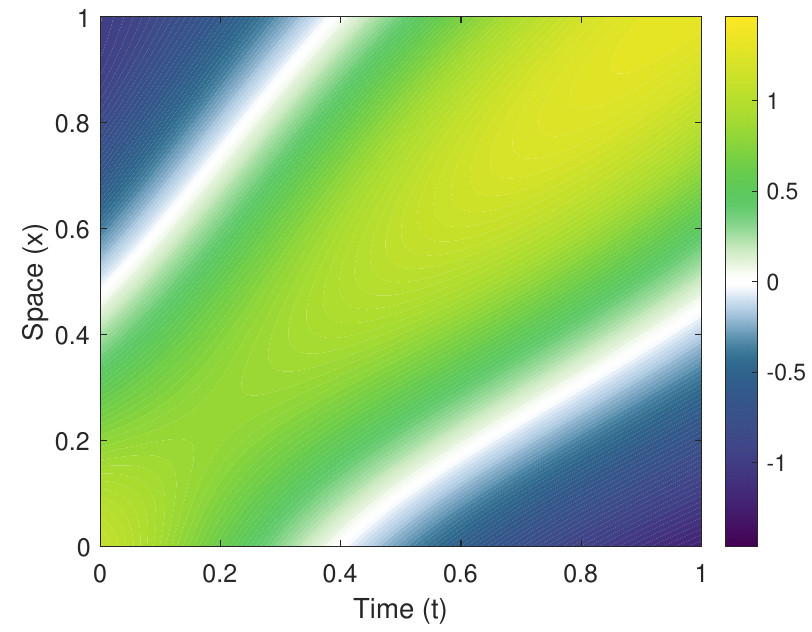} &
\includegraphics[width=0.17\linewidth]{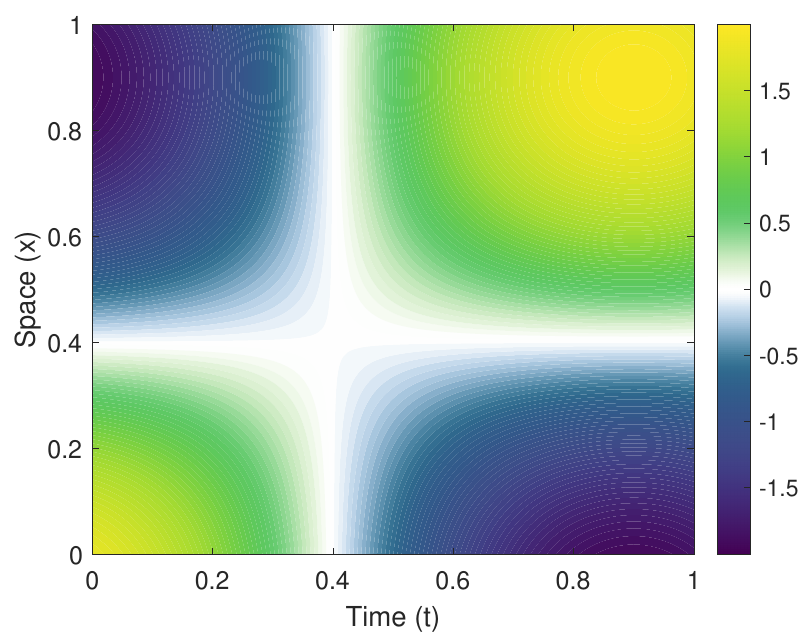} &
\includegraphics[width=0.17\linewidth]{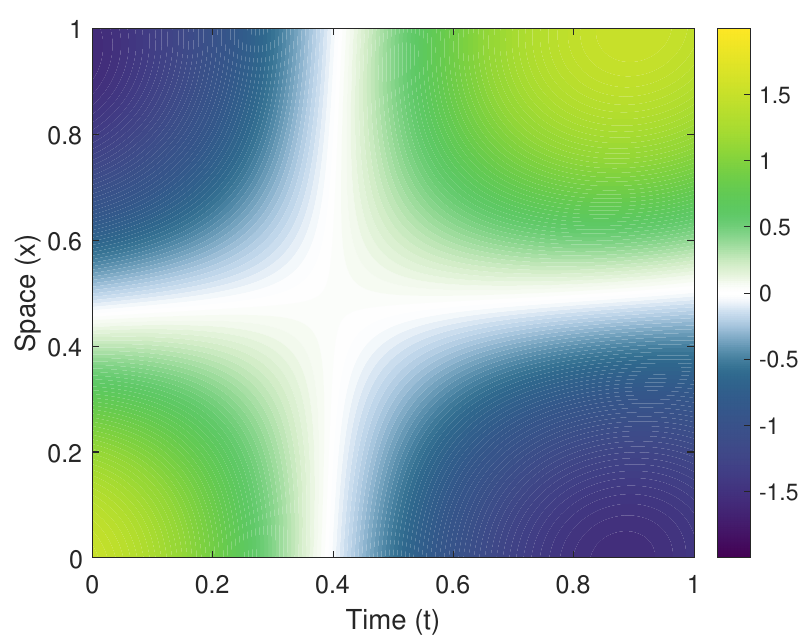} &
\includegraphics[width=0.17\linewidth]{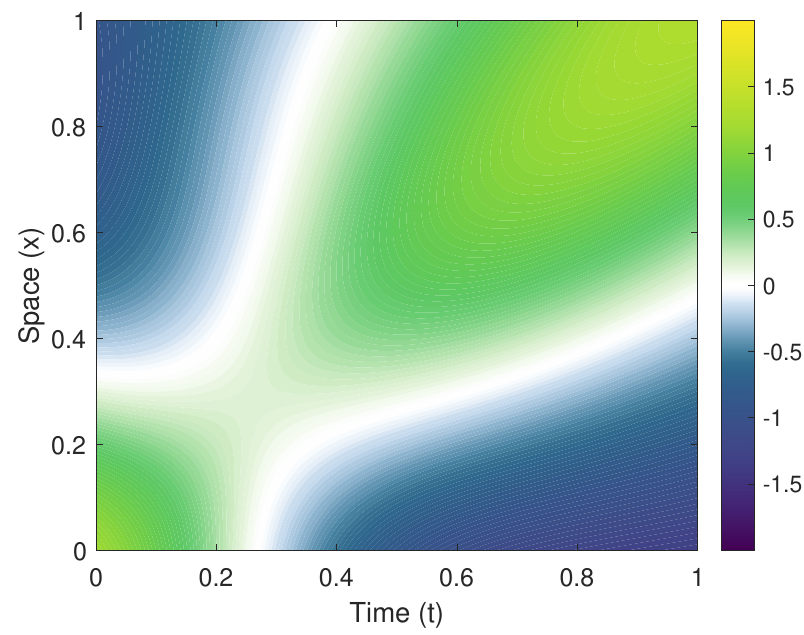} \\
\bottomrule[1pt]
\end{tabular}}
\label{tab:final_results}
\vspace{-4mm}
\end{table*}


\newpage
\section{Proofs of Lemmas and Propositions}
\label{sec:appendix_proofsofLemmasandprop}
\begin{definition}[\textbf{Baker-Campbell-Hausdorff}]
\label{def:BCH}
Let $A,B$ be elements of a Lie algebra for which the BCH series converges.
Then there exists $C$ such that
\begin{equation}
e^{A}e^{B}=e^{C},
\end{equation}
where $C$ is given by the BCH expansion
\begin{align}
C=& A+B+\frac{1}{2}[A,B] \notag \\ &+ \frac{1}{12}[A,[A,B]]+\frac{1}{12}[B,[B,A]] \notag \\ &-\frac{1}{24}[B,[A,[A,B]]]+\cdots .
\end{align}
\end{definition}

\subsection{Proof of Hadamard's Lemma (Lemma \ref{lem:Hadamard's Lemma1})}
\label{appendix:Hadamard_Lemma}
\begin{proof}
Define $F(t)=e^{tA}Be^{-tA}$. Then $F(0)=B$. Differentiate using the product rule:
\begin{equation}
F'(t)=Ae^{tA}Be^{-tA}-e^{tA}BAe^{-tA}=[A,F(t)].
\end{equation}
Evaluating at $t=0$ gives
\begin{equation}
F'(0)=[A,B].
\end{equation}
Differentiate again:
\begin{equation}
F''(t)=[A,F'(t)]=[A,[A,F(t)]]
\quad\Rightarrow\quad
F''(0)=[A,[A,B]].
\end{equation}
By induction,
\begin{equation}
F^{(k)}(t)=\operatorname{ad}_A^k(F(t))
\quad\Rightarrow\quad
F^{(k)}(0)=\operatorname{ad}_A^k(B),
\end{equation}
where $\operatorname{ad}_A^0(B):=B$ and $\operatorname{ad}_A^{k+1}(B):=[A,\operatorname{ad}_A^k(B)]$.

Apply the Taylor expansion of $F(t)$ about $t=0$:
\begin{equation}
F(t)=\sum_{k=0}^{\infty}\frac{t^k}{k!}F^{(k)}(0)
=\sum_{k=0}^{\infty}\frac{t^k}{k!}\,\operatorname{ad}_A^k(B).
\end{equation}
Setting $t=1$ yields $e^{A}Be^{-A}=\sum_{k\ge 0}\frac{1}{k!}\operatorname{ad}_A^k(B)$, and the general-$t$ form follows. This completes the proof. 
\end{proof}

\subsection{Proof of Proposition \ref{prop:2Dclosure} (Closed Invariant Subspace.)}
\label{appendix:closed_invariant}
\begin{proof}
Let $S(t)=e^{-\frac{it}{2}A}$, by the Hadamard's lemma, 
    \begin{equation}
        e^{it/2A}Be^{-it/2A}=\sum_{k\ge 0}\frac{1}{k!}\left(\frac{it}{2}\right)^k\mathrm{ad}_A^k(B),
    \end{equation}
    where 
            $[A,B]=\mathrm{ad}_A(B)=-2i\omega C, \qquad [A,C]=\mathrm{ad}_A(C)=2i\omega B. $
        
        Using this, compute the first few nested commutators:
        \begin{equation}
        \mathrm{ad}_A(B)=-2i\omega C,\qquad
        \mathrm{ad}_A^2(B)=\mathrm{ad}_A(-2i\omega C)=-2i\omega(2i\omega B)=4\omega^2 B, \notag
        \end{equation}
        \begin{equation}
        \mathrm{ad}_A^3(B)=\mathrm{ad}_A(4\omega^2 B)=4\omega^2(-2i\omega C)=-8i\omega^3 C,
        \qquad
        \mathrm{ad}_A^4(B)=16\omega^4 B, \notag
        \end{equation}
        Every $k\ge0$ is either even or odd, so setting $k=2m$ and $k=2m+1$ with $m\ge0$ splits the Hadamard series into two sums:
    \begin{equation}
        e^{\frac{it}{2}A}Be^{-\frac{it}{2}A}
        =\underbrace{\sum_{m\ge 0}\frac{1}{(2m)!}\left(\frac{it}{2}\right)^{2m}\mathrm{ad}_A^{2m}(B)}_{\rm Even}
        +\underbrace{\sum_{m\ge 0}\frac{1}{(2m+1)!}\left(\frac{it}{2}\right)^{2m+1}\mathrm{ad}_A^{2m+1}(B)}_{\rm Odd}.\notag
    \end{equation}

    \textbf{Even terms.} $i^{2m}=(-1)^m$ and $\mathrm{ad}_A^{2m}(B)=2^{2m}\omega^{2m}B$:
    \begin{equation}
        \frac{1}{(2m)!}\cdot\frac{(-1)^m t^{2m}}{2^{2m}}\cdot 2^{2m}\omega^{2m}B
        =\frac{(-1)^m(\omega t)^{2m}}{(2m)!}\,B.\notag
    \end{equation}

    \textbf{Odd terms.} $\left(\frac{it}{2}\right)^{2m+1}=\dfrac{i\,(-1)^m t^{2m+1}}{2^{2m+1}}$
    and $\mathrm{ad}_A^{2m+1}(B)=-i\,2^{2m+1}\omega^{2m+1}C$:
    \begin{equation}
        \frac{1}{(2m+1)!}\cdot\frac{i\,(-1)^m t^{2m+1}}{2^{2m+1}}
        \cdot\left(-i\,2^{2m+1}\omega^{2m+1}\right)C
        =\frac{(-1)^m(\omega t)^{2m+1}}{(2m+1)!}\,C.\notag
    \end{equation}
    Resumming the two series gives the first identity:
    \begin{align}
        e^{\frac{it}{2}A}Be^{-\frac{it}{2}A}
        &=B\sum_{m\ge 0}\frac{(-1)^m(\omega t)^{2m}}{(2m)!}
        +C\sum_{m\ge 0}\frac{(-1)^m(\omega t)^{2m+1}}{(2m+1)!} \\
        &=B\cos(\omega t)+C\sin(\omega t).
    \end{align}
    We can make similar case for 
    \begin{align}
        e^{\frac{it}{2}A}
        Ce^{-\frac{it}{2}A}=C\cos(\omega t)-B\sin(\omega t).
    \end{align}
    This completes the proof.
\end{proof}

\subsection{Proof of Proposition \ref{prop:planes}}
\label{appendix:invariant_planes}
\begin{proof}
Each qubit $j$ is assigned one of $\hat I_j$, $\hat X_j$, $E_j^{+}=\hat Z_j+i\hat Y_j$,
$E_j^{-}=\hat Z_j-i\hat Y_j$, and $\langle A,B\rangle = 2^{-N}\operatorname{Tr}(A^{\dagger}B)$.
 
\textbf{Step 1 (single-qubit eigen-operators).}
Only the $j$-th term of $G=\sum_k \hat X_k$ fails to commute with an operator supported on qubit $j$.
From $[X,Z]=-2iY$ and $[X,Y]=2iZ$,
\begin{equation}
  \operatorname{ad}_G\!\left(E_j^{\pm}\right)
  = \big[\hat X_j,\ \hat Z_j \pm i\hat Y_j\big]
  = -2i\hat Y_j \pm i\big(2i\hat Z_j\big)
  = \mp 2\big(\hat Z_j \pm i\hat Y_j\big)
  = \mp 2\,E_j^{\pm},
\end{equation}
while $\operatorname{ad}_G(\hat I_j)=\operatorname{ad}_G(\hat X_j)=0$.
 
\textbf{Step 2 (eigenvalue of a product).}
By the Jacobi identity $\operatorname{ad}_G$ is a derivation,
$\operatorname{ad}_G(AB)=\operatorname{ad}_G(A)B+A\operatorname{ad}_G(B)$.
The factors of $E_c$ act on pairwise disjoint qubits, so applying the derivation rule across them
and using Step 1,
\begin{equation}
  \operatorname{ad}_G(E_c)
  = \Big(\sum_{j\in S_+}(-2)\;+\;\sum_{j\in S_-}(+2)\;+\;\sum_{j\in T}0\Big)E_c
  = -2\big(|S_+|-|S_-|\big)E_c
  = -2\,\omega(c)\,E_c .
\end{equation}
 
\textbf{Step 3 (orthogonality, norms, and completeness).}
On a single qubit, with $\langle A,B\rangle_1=\tfrac12\operatorname{Tr}(A^{\dagger}B)$,
\begin{equation}
  \langle \hat I,\hat I\rangle_1=\langle \hat X,\hat X\rangle_1=1,
  \qquad
  \langle E^{+},E^{+}\rangle_1=\langle E^{-},E^{-}\rangle_1=2,
\end{equation}
where the last equality uses $(E^{+})^{\dagger}E^{+}=(Z-iY)(Z+iY)=2I+2X$.
All cross terms vanish: $\operatorname{Tr}(X)=\operatorname{Tr}(Y)=\operatorname{Tr}(Z)=0$ kills
$\langle \hat I,\cdot\rangle_1$; $XZ=-iY$ and $XY=iZ$ are traceless, so
$\langle \hat X,E^{\pm}\rangle_1=0$; and $\{Z,Y\}=0$ gives
$(E^{+})^{\dagger}E^{-}=(Z-iY)^{2}=-i\{Z,Y\}=0$, so $\langle E^{+},E^{-}\rangle_1=0$.
Hence $\{\hat I,\hat X,E^{+},E^{-}\}$ is an orthogonal basis of $\mathcal B(\mathcal H_1)$.
 
The inner product factorises over tensor products,
$\big\langle \bigotimes_j A_j,\ \bigotimes_j B_j\big\rangle=\prod_j\langle A_j,B_j\rangle_1$,
so the $4^{N}$ operators $\{E_c\}$ are mutually orthogonal with
\begin{equation}
  \|E_c\|^{2}=2^{|S_+|+|S_-|}=2^{m(c)}\neq 0 .
\end{equation}
Being $4^{N}=\dim\mathcal B(\mathcal H)$ in number, they form an orthogonal basis.
 
\textbf{Step 4 (Hermitian pairs and closure).}
The factors commute and $(E_j^{\pm})^{\dagger}=E_j^{\mp}$, so
$E_c^{\dagger}=E_{\bar c}$ with $\bar c=(S_-,S_+,T)$; thus $\omega(\bar c)=-\omega(c)$,
$m(\bar c)=m(c)$, $E_c^{\dagger}$ is again a basis element, and
$\operatorname{ad}_G(E_c^{\dagger})=+2\omega(c)E_c^{\dagger}$.
Therefore $\mathcal X_c^{\dagger}=\mathcal X_c$, $\mathcal Y_c^{\dagger}=\mathcal Y_c$, and
\begin{equation}
  [G,\mathcal X_c]=-2\omega\big(E_c-E_c^{\dagger}\big)=-2i\omega\,\mathcal Y_c,
  \qquad
  [G,\mathcal Y_c]=\tfrac{1}{i}(-2\omega)\big(E_c+E_c^{\dagger}\big)=2i\omega\,\mathcal X_c ,
\end{equation}
which is the hypothesis of Proposition 1.
For $\omega(c)>0$ we have $S_+\neq S_-$, so $E_c$ and $E_c^{\dagger}$ are \emph{distinct} basis
elements and hence orthogonal, giving
\begin{equation}
  \|\mathcal X_c\|^{2}=\|E_c\|^{2}+\|E_c^{\dagger}\|^{2}=2^{m(c)+1}=\|\mathcal Y_c\|^{2},
  \qquad \langle \mathcal X_c,\mathcal Y_c\rangle=0 .
\end{equation}
The same argument gives $\mathcal X_c,\mathcal Y_c\perp\mathcal X_{c'},\mathcal Y_{c'}$ whenever
$c\neq c'$ with $\omega(c)=\omega(c')>0$.
 
\textbf{Step 5 (decomposition).}
Since $\operatorname{span}_{\mathbb C}\{E_c,E_c^{\dagger}\}
=\operatorname{span}_{\mathbb C}\{\mathcal X_c,\mathcal Y_c\}$, grouping the basis $\{E_c\}$ by
$\omega(c)$ yields the orthogonal decomposition
\begin{equation}
  \mathcal B(\mathcal H)=\mathcal B_0\oplus\bigoplus_{\omega=1}^{N}\mathcal B_\omega,
  \qquad
  \mathcal B_0=\operatorname{span}\{E_c:\omega(c)=0\}=\ker\operatorname{ad}_G,
\end{equation}
with $\mathcal B_\omega=\operatorname{span}\{\mathcal X_c,\mathcal Y_c:\omega(c)=\omega\}$.
Finally $(E_j^{\pm})^{2}=\hat Z_j^{2}-\hat Y_j^{2}\pm i\{\hat Z_j,\hat Y_j\}=0$, so each qubit
carries at most one $E^{\pm}$ factor and $|\omega(c)|\le m(c)\le N$. This completes the proof. 
\end{proof}

\subsection{Proof of Proposition \ref{proposition:multi-qubit_Projected}}
\label{appendix:proofofMulti_qubitProjected}
\begin{proof}
Let $U(x)=W_L S(x) W_R$ with $S(x)=e^{-ixG/2}$ and $G=\sum_j \hat X_j$.
Using $\langle\psi|B|\psi\rangle=\operatorname{Tr}(|\psi\rangle\langle\psi|B)$ and cyclicity,
\begin{equation}
  f(x)=\langle\psi|U(x)^{\dagger}MU(x)|\psi\rangle
  =\operatorname{Tr}\Big(\underbrace{W_R\rho_0W_R^{\dagger}}_{\rho}\,S(x)^{\dagger}
   \underbrace{W_L^{\dagger}MW_L}_{O}\,S(x)\Big)
  =\operatorname{Tr}\big(\rho\,S(x)^{\dagger}OS(x)\big).
\end{equation}
 
\textbf{Step 1 (trace to inner product).}
With $\langle A,B\rangle=2^{-N}\operatorname{Tr}(A^{\dagger}B)$ and $\rho=\rho^{\dagger}$,
\begin{equation}
  \operatorname{Tr}(\rho B)=\operatorname{Tr}(\rho^{\dagger}B)=2^{N}\langle\rho,B\rangle, \quad
  f(x)=2^{N}\big\langle \rho,\ S(x)^{\dagger}OS(x)\big\rangle .
\end{equation}
Because $\rho$, $O$, $\mathcal X_{\omega,r}$ and $\mathcal Y_{\omega,r}$ are all Hermitian, every
inner product below is real.
 
\textbf{Step 2 (orthogonal projection onto one plane).}
By Proposition 2 the pairs $\{(\mathcal X_{\omega,r},\mathcal Y_{\omega,r})\}$ together with
$\mathcal B_0$ are mutually orthogonal and span $\mathcal B(\mathcal H)$, with
$\|\mathcal X_{\omega,r}\|^{2}=\|\mathcal Y_{\omega,r}\|^{2}=2^{m_{\omega,r}+1}$.
The basis is orthogonal but \emph{not} normalised, so the projection carries the norm factor:
\begin{equation}
  \Pi_{\omega,r}(O)=
  \frac{\langle O,\mathcal X_{\omega,r}\rangle\,\mathcal X_{\omega,r}
       +\langle O,\mathcal Y_{\omega,r}\rangle\,\mathcal Y_{\omega,r}}{2^{\,m_{\omega,r}+1}},
  \qquad
  O=\Pi_0O+\sum_{\omega=1}^{N}\sum_{r=1}^{d_\omega}\Pi_{\omega,r}(O),
\end{equation}
where $\Pi_0$ is the orthogonal projection onto $\mathcal B_0$.
 
\textbf{Step 3 (the $\omega=0$ term).}
$\Pi_0O\in\mathcal B_0=\ker\operatorname{ad}_G$ commutes with $G$ and therefore with $S(x)$, so
$S(x)^{\dagger}(\Pi_0O)S(x)=\Pi_0O$ and this term contributes
$2^{N}\langle\rho,\Pi_0O\rangle=\operatorname{Tr}(\rho\,\Pi_0O)$, independent of $x$.
 
\textbf{Step 4 (rotation within a plane).}
Proposition 1 gives
$S^{\dagger}\mathcal XS=\cos(\omega x)\mathcal X+\sin(\omega x)\mathcal Y$ and
$S^{\dagger}\mathcal YS=\cos(\omega x)\mathcal Y-\sin(\omega x)\mathcal X$.
Substituting into $\Pi_{\omega,r}(O)$ and collecting the $\mathcal X$ and $\mathcal Y$ components,
\begin{equation}
\begin{split}
 S^{\dagger}\Pi_{\omega,r}(O)S = 2^{-(m_{\omega,r}+1)}\Big[
 &\big(\langle O,\mathcal X\rangle\cos\omega x-\langle O,\mathcal Y\rangle\sin\omega x\big)\mathcal X\\
 +\;&\big(\langle O,\mathcal X\rangle\sin\omega x+\langle O,\mathcal Y\rangle\cos\omega x\big)\mathcal Y\Big].
\end{split}
\end{equation}
 
\textbf{Step 5 (contraction with $\rho$).}
Applying $2^{N}\langle\rho,\cdot\rangle$ to the previous line,
\begin{equation}
\begin{split}
 f_{\omega,r}(x)=2^{\,N-m_{\omega,r}-1}\Big[
 &\langle\rho,\mathcal X\rangle\big(\langle O,\mathcal X\rangle\cos\omega x-\langle O,\mathcal Y\rangle\sin\omega x\big)\\
 +\;&\langle\rho,\mathcal Y\rangle\big(\langle O,\mathcal X\rangle\sin\omega x+\langle O,\mathcal Y\rangle\cos\omega x\big)\Big],
\end{split}
\end{equation}
which is precisely the bilinear form
\begin{equation}
 f_{\omega,r}(x)=2^{\,N-m_{\omega,r}-1}
 \begin{bmatrix}\langle\rho,\mathcal X_{\omega,r}\rangle\\[2pt]\langle\rho,\mathcal Y_{\omega,r}\rangle\end{bmatrix}^{\!\top}
 \begin{bmatrix}\cos(\omega x) & -\sin(\omega x)\\[2pt] \sin(\omega x) & \cos(\omega x)\end{bmatrix}
 \begin{bmatrix}\langle O,\mathcal X_{\omega,r}\rangle\\[2pt]\langle O,\mathcal Y_{\omega,r}\rangle\end{bmatrix}.
\end{equation}
Summing over all $(\omega,r)$ and adding the constant term of Step 3 gives \eqref{eq:bilinear_form}.
\end{proof}

\subsection{Proof of Lemma \ref{lem:ising}}
\label{appendix:ising_entanglement}
\begin{proof}
Define the local Ising generator on qubit $n$ as
\begin{equation}
H_n:=\sum_{k\in\mathcal{N}(n)}\gamma_{nk} Z_n Z_k .
\end{equation}
Since $Z_n$ acts only on qubit $n$ and each $Z_k$ acts on $k\neq n$, all factors commute and
\begin{equation}
H_n
=\sum_{k\in\mathcal{N}(n)}\gamma_{nk}\, Z_n Z_k
= Z_n \sum_{k\in\mathcal{N}(n)}\gamma_{nk} Z_k
=: Z_n \Phi_n .
\end{equation}
Therefore,
\begin{equation}
U^{\mathrm{ent}}=\exp\!\left(-\frac{i}{2}H_n\right)=\exp\!\left(-\frac{i}{2}Z_n\Phi_n\right).
\end{equation}
Throughout, $\Phi_n$ commutes with $X_n,Y_n,Z_n$ because it is supported only on $\mathcal{N}(n)$.

\paragraph{Conjugation of $Z_n$.}
Because $[Z_n\Phi_n,Z_n]=0$, conjugation leaves $Z_n$ invariant:
\begin{equation}
(U^{\mathrm{ent}})^\dagger Z_n U^{\mathrm{ent}}=Z_n.
\end{equation}

\paragraph{Conjugation of $X_n$.}
Let $A:=\frac{i}{2}Z_n\Phi_n$, so that $(U^{\mathrm{ent}})^\dagger X_n U^{\mathrm{ent}}=e^{A}X_n e^{-A}$.
Using the Baker--Campbell--Hausdorff (BCH) expansion in adjoint form,
\begin{equation}
e^{A}X_n e^{-A}=\sum_{m=0}^\infty \frac{1}{m!}\,\mathrm{ad}_A^{m}(X_n),
\qquad \mathrm{ad}_A(B):=[A,B],
\end{equation}
and the Pauli commutators $[Z_n,X_n]=2iY_n$ and $[Z_n,Y_n]=-2iX_n$, together with $[\Phi_n,X_n]=[\Phi_n,Y_n]=0$, yield
\begin{align}
\mathrm{ad}_{Z_n\Phi_n}(X_n)
&=[Z_n\Phi_n,X_n]=[Z_n,X_n]\Phi_n=2iY_n\Phi_n,\\
\mathrm{ad}_{Z_n\Phi_n}^2(X_n)
&=[Z_n\Phi_n,2iY_n\Phi_n]=2i[Z_n,Y_n]\Phi_n^2=4X_n\Phi_n^2.
\end{align}
By induction, for all $m\ge 0$,
\begin{equation}
\mathrm{ad}_{Z_n\Phi_n}^{2m}(X_n)=2^{2m}X_n\Phi_n^{2m},
\qquad
\mathrm{ad}_{Z_n\Phi_n}^{2m+1}(X_n)=2^{2m+1}iY_n\Phi_n^{2m+1}.
\end{equation}
Since $A=\frac{i}{2}Z_n\Phi_n$, it follows that
\begin{align}
(U^{\mathrm{ent}})^\dagger X_n U^{\mathrm{ent}}
&=\sum_{m=0}^\infty \frac{1}{(2m)!}\left(\frac{i}{2}\right)^{2m}2^{2m}X_n\Phi_n^{2m}
+\sum_{m=0}^\infty \frac{1}{(2m+1)!}\left(\frac{i}{2}\right)^{2m+1}2^{2m+1}iY_n\Phi_n^{2m+1}\notag\\
&=X_n\sum_{m=0}^\infty \frac{(-1)^m\Phi_n^{2m}}{(2m)!}
-Y_n\sum_{m=0}^\infty \frac{(-1)^m\Phi_n^{2m+1}}{(2m+1)!}\notag\\
&=X_n\cos\Phi_n-Y_n\sin\Phi_n,
\end{align}
where $\cos\Phi_n$ and $\sin\Phi_n$ are defined by their convergent power series.

\paragraph{Conjugation of $Y_n$.}
The same argument gives
\begin{equation}
(U^{\mathrm{ent}})^\dagger Y_n U^{\mathrm{ent}}=Y_n\cos\Phi_n+X_n\sin\Phi_n.
\end{equation}
This completes the proof. 
\end{proof}

\subsection{Proof of Corollary~\ref{cor:activation}}
\label{app:activation}
We prove a sharper form, where \emph{almost every} means outside a closed set of measure zero. Let the local blocks be general single-qubit rotations, $\rho_{0}$ a pure product state and $\omega\ge1$. \textbf{(i)} If $\mathbf{b}_{\omega,r}\neq0$ for some $r$ at some $\boldsymbol{\theta}^{(L)}$, then $c_{\omega}(\boldsymbol{\Theta})\neq0$ for almost every $\boldsymbol{\Theta}$, and otherwise $c_{\omega}\equiv0$. \textbf{(ii)} For $W_{L}=W_{L}^{\mathrm{post}}U_{\mathrm{ent}}W_{L}^{\mathrm{pre}}$ with readout $\hat{Z}_{n}$ and every $\sin\gamma_{nk}\neq0$, (i) applies to every $\omega\le\deg(n)+1$, and $c_{\omega}\equiv0$ for $\omega>\deg(n)+1$. \textbf{(iii)} At $\omega=\deg(n)+1$ a single plane carries $c_{\omega}$, and any product state $\rho$ limits the amplitude there to $2^{-\deg(n)}$. With $|\gamma_{nk}|=\pi/2$ and $W_{R}=W_{R}^{\mathrm{post}}U_{\mathrm{ent}}W_{R}^{\mathrm{pre}}$, the circuit can output $f(x)=\cos(\omega x)$ for any single $\omega\le\deg(n)+1$. Throughout, $L=\{n\}\cup\mathcal{N}(n)$ and $Q_{\omega}$ is the orthogonal projection onto $\mathrm{span}\{E_{c}:\omega(c)=\omega\}$.

\textbf{Proof of (i).} By Proposition~\ref{prop:planes}, $S(x)^{\dagger}E_{c}S(x)=e^{-i\omega(c)x}E_{c}$, so expanding $O$ in $\{E_{c}\}$ and using $\mathcal{X}_{c}=E_{c}+E_{c}^{\dagger}$, $\mathcal{Y}_{c}=-i(E_{c}-E_{c}^{\dagger})$ gives
\begin{equation}
\label{eq:coefTrace}
f(x)=\sum_{|\omega|\le N}\nolimits e^{-i\omega x}\,\mathrm{Tr}\bigl(\rho\,Q_{\omega}(O)\bigr),
\qquad
\|\mathbf{b}_{\omega,r}\|=2\,|\langle E_{c},O\rangle|,
\end{equation}
so $|c_{\omega}|=|\mathrm{Tr}(\rho\,Q_{\omega}(O))|$ because $f$ is real, likewise $\|\mathbf{a}_{\omega,r}\|=2|\langle E_{c},\rho\rangle|$, and $\mathbf{b}_{\omega,r}\neq0$ for some $r$ exactly when $Q_{\omega}(O)\neq0$. If this never happens, $c_{\omega}\equiv0$. Otherwise let $Q_{\star}=Q_{\omega}(O)\neq0$ at some $\boldsymbol{\theta}^{(L)}_{\star}$ and write $W_{R}=VW_{R}^{\mathrm{pre}}$. The reachable states include $V\sigma V^{\dagger}$ for every pure product state $\sigma$, and these span $\mathcal{B}(\mathcal{H})$, while $\rho\mapsto\mathrm{Tr}(\rho\,Q_{\star})$ is a nonzero functional, equal to $2^{N}\|Q_{\star}\|^{2}$ at $\rho=Q_{\star}^{\dagger}$, so some $\boldsymbol{\theta}^{(R)}$ gives $c_{\omega}\neq0$. Each gate $e^{-i\theta P/2}=\cos(\theta/2)I-i\sin(\theta/2)P$, including every factor of $U_{\mathrm{ent}}$, is a trigonometric polynomial in its half-angle, so $|c_{\omega}(\boldsymbol{\Theta})|^{2}$ is a real trigonometric polynomial that is not identically zero. Such a polynomial has finitely many zeros per period in one variable, and by induction on the number of variables with Fubini's theorem it vanishes only on a closed set of measure zero.

\textbf{Proof of (ii).} By Appendix~\ref{app:Nqubit}, $O$ is supported on $L$, which gives \eqref{eq:reachbound} and $c_{\omega}\equiv0$ for $\omega>\deg(n)+1$. For $1\le\omega\le\deg(n)+1$ and $\deg(n)\ge1$, the case $\deg(n)=0$ being immediate, choose $K\subseteq\mathcal{N}(n)$ with $|K|=\deg(n)+1-\omega$, let $W_{L}^{\mathrm{post}}$ give $\mu\neq0$, and let $W_{L}^{\mathrm{pre}}$ map $\hat{Z}_{k}$ to $\hat{X}_{k}$ on $K$ and act trivially elsewhere. Since every string of \eqref{eq:M2general} carries $\hat{I}$ or $\hat{Z}$ on each neighbour, the only string of $O$ that overlaps $E_{c}$ for $c=(L\setminus K,\emptyset,K)$, of frequency $\omega$, is $\hat{Y}_{n}\hat{Z}_{\mathcal{N}(n)\setminus K}\hat{X}_{K}$, the image of the top string. Hence $\|\mathbf{b}_{\omega,r}\|=2|\mu|\prod_{k}|\sin\gamma_{nk}|\neq0$ and (i) applies.

\textbf{Proof of (iii).} Every configuration with $\langle E_{c},O\rangle\neq0$ lies inside $L$, so $\omega(c)=|L|$ forces $S_{+}=L$, only this plane contributes, and $|c_{\deg(n)+1}|=2^{-|L|}|\langle E_{c},O\rangle|\,|\mathrm{Tr}(\rho E_{c})|$. Since $\hat{Z}_{j}+i\hat{Y}_{j}=2\ket{-}\bra{+}$, $E_{c}$ has trace norm $2^{N}$, so $|\langle E_{c},O\rangle|\le\|O\|=1$, and a product state gives $|\mathrm{Tr}(\rho E_{c})|=\prod_{j\in L}|\mathrm{Tr}(\rho_{j}(\hat{Z}_{j}+i\hat{Y}_{j}))|\le1$, which bounds the amplitude $2|c_{\deg(n)+1}|$ by $2^{1-|L|}=2^{-\deg(n)}$. For attainment, let $|\gamma_{nk}|=\pi/2$, so $\cos\Phi_{n}$ and $\sin\Phi_{n}$ reduce to $0$ or $\pm\hat{Z}_{\mathcal{N}(n)}$, and fix $L'\subseteq L$ with $n\in L'$ and $|L'|=\omega$. Choosing $M_{1}=\hat{X}_{n}$ and a $W_{L}^{\mathrm{pre}}$ that rotates the factor on $n$ to $\hat{Z}_{n}$ and $\hat{Z}_{k}$ to $\hat{X}_{k}$ for $k\notin L'$ gives $O=\pm\hat{Z}_{L'}\hat{X}_{\mathcal{N}(n)\setminus L'}$. On the state side, $W_{R}^{\mathrm{pre}}$ prepares $\ket{+}$ on $L'$ and $\ket{0}$ elsewhere, $U_{\mathrm{ent}}$ turns $L'$ into a star graph state up to $\hat{Z}$ rotations, and $W_{R}^{\mathrm{post}}$ undoes these rotations, applies a Hadamard gate to qubit $n$ and rotates the qubits outside $L'$ to $\ket{+}$, which yields $(\ket{+}^{\otimes\omega}+\ket{-}^{\otimes\omega})/\sqrt{2}$ on $L'$. The encoder gives $\ket{\pm}^{\otimes\omega}$ the phases $e^{\mp i\omega x/2}$ and $\hat{Z}_{L'}$ exchanges them, so $f(x)=\pm\cos(\omega x)$, with the sign set by a $\pi$ rotation in $W_{L}^{\mathrm{post}}$. At $\omega=\deg(n)+1$ no output of this block does better, since averaging $f$ over the shifts $x\mapsto x+2\pi j/|L|$ keeps only the frequencies $0$ and $|L|$ and leaves an amplitude of at most $1$.

\textbf{Remarks.} $CZ$ differs from the Ising form at $|\gamma_{nk}|=\pi/2$ only by single-qubit $\hat{Z}$ rotations, which the adjacent local blocks absorb. When $x$ and $t$ are encoded on disjoint qubit sets $V_{x}$ and $V_{t}$, each $E_{c}$ carries a frequency vector and the proof of (i) holds verbatim. The construction of (ii), which also uses the string $\hat{X}_{n}\hat{Z}_{\mathcal{N}(n)}$ when $n\in T$, then reaches every vector with $|\omega_{x}|\le|L\cap V_{x}|$ and $|\omega_{t}|\le|L\cap V_{t}|$, which for the 14-qubit star of our advection experiment is $|\omega_{x}|,|\omega_{t}|\le7$.

\newpage
\section{Detailed Explanation}
\subsection{Defining Closed-Invariant Subspaces}
\label{appendix:defineclosedinvariant}
Equipped with Proposition \ref{prop:2Dclosure}, our objective is to represent the quantum model's output as Fourier series representation in terms of the operators. The question we can ask is how we can define such set of operator pairs. For this, we employ the method of \emph{eigen-operators of the adjoint map} \citep{gardiner2004quantum}. This technique builds operator pairs from eigen-operators that close under Lie commutation with the generator. To construct the target operators $\mathcal{X}$ and $\mathcal{Y}$ from such eigen-operators, we begin by establishing the fundamental building blocks on the single-qubit. Additionally, the key observation is that the encoding unitary $S(x)$ is often constructed from a single Pauli generator type; in our case, it is generated by Pauli generator $X$.
In this light, we first define on each qubit $n$ the local eigen-operators as
$E_n^+ = Z_n + iY_n$, 
$E_n^- = Z_n - iY_n$,
which satisfy
$[X_n, E_n^+] = -2 E_n^+$, 
$[X_n, E_n^-] = 2 E_n^-$.

For the global multi-qubit generator $G = \sum_{n=1}^N \hat{X}_n$, we apply tensor product to these local eigen-operators over pairwise disjoint index sets $S_+, S_-, T \subseteq \{1, \dots, N\}$. Because $\hat{X}_k$ commutes with $G$ just as $I_k$ does, a site outside $S_+\cup S_-$ may carry either factor, and both must be retained for the resulting family to span $\mathcal{B}(\mathcal{H})$. Specifically, define
\begin{equation}
E_{S_+,S_-,T}\triangleq\bigotimes_{n\in S_+} E_n^+ \otimes \bigotimes_{m\in S_-} E_m^- \otimes \bigotimes_{t\in T} X_t \otimes \bigotimes_{k\notin S_+\cup S_-\cup T} I_k.
\end{equation}

Then, if $n\in S_+$, the commutation is described as
\begin{align}
\hspace{-3mm}
    [\hat{X}_n,E_{S_+,S_-,T}]&=I^{\otimes (n-1)}\otimes [X_n,E_n^+]\otimes I^{\otimes (N-n)}=(-2)E_{S_+,S_-,T},
\end{align}
and likewise $[\hat{X}_m,E_{S_+,S_-,T}]=(+2)E_{S_+,S_-,T}$ for $m\in S_-$, while a site $j\in T$ or a site carrying the identity contributes nothing, since $[\hat{X}_j,\hat{X}_j]=0$ and $[\hat{X}_j,\hat{I}_j]=0$.
Thus, Lie commutation of $G$ with $E_{S_+,S_-,T}$ can be as
\begin{align}
    [G,E_{S_+,S_-,T}]&=\sum_{j=1}^N\nolimits [\hat{X}_j,E_{S_+,S_-,T}] =\sum_{j\in S_+}\limits-2E_{S_+,S_-,T}+\sum_{j\in S_-}\limits2E_{S_+,S_-,T} \\
    &=-2(|S_+| - |S_-|)E_{S_+,S_-,T} =-2\omega E_{S_+,S_-,T}.
\end{align}
The result shows that the difference of the cardinality of the disjoint sets determines $\omega=|S_+| - |S_-|$.

Furthermore, note that there could be many combinations of disjoint sets $(S_+,S_-,T)$ that yield the same $\omega$. We express such set as
\begin{multline}
\mathcal{S}_\omega \triangleq \Big\{(S_+,S_-,T): S_+,S_-,T \text{ pairwise disjoint},\; S_+,S_-,T\subseteq\{1,\dots,N\},\; |S_+|-|S_-|=\omega\Big\} \label{eq:degenerateSet}.
\end{multline}
Denote $d_\omega = |\mathcal{S}_\omega|$, the cardinality of \eqref{eq:degenerateSet}. This cardinality has a closed form. Each qubit independently carries one of $I_k$, $X_k$, $E_k^+$, $E_k^-$, contributing $z^{0}$, $z^{0}$, $z^{+1}$, $z^{-1}$ respectively when the frequency is tracked by a formal variable $z$, so
\begin{equation}
\sum_{(S_+,S_-,T)} z^{\,\omega(S_+,S_-,T)} = \big(2+z+z^{-1}\big)^{N}=z^{-N}(1+z)^{2N},
\end{equation}
and reading off the coefficient of $z^{\omega}$ gives
\begin{equation}
\label{eq:planeCount}
d_\omega=\binom{2N}{\,N+\omega\,},\qquad \omega=-N,\dots,N .
\end{equation}
In particular $d_0=\binom{2N}{N}=\dim\mathcal{B}_0$, and $d_0+2\sum_{\omega=1}^{N}d_\omega=4^{N}$, confirming that the family exhausts $\mathcal{B}(\mathcal{H})$. The extreme planes $|\omega|=N$ are attained only by $S_-=T=\varnothing$, so $d_{\pm N}=1$; for $N=2$ one has $d_1=\binom{4}{3}=4$ and $d_2=\binom{4}{4}=1$.
We choose an enumeration $\{(S_+^{(r)},S_-^{(r)},T^{(r)})\}_{r=1}^{d_\omega}$ of $\mathcal{S}_\omega$ and define
$E_{\omega,r} \triangleq E_{S_+^{(r)},S_-^{(r)},T^{(r)}}$.
Accordingly, the exact operators can be formulated as
\begin{align}
\label{eq:ClosedInvariantOperators}
\mathcal{X}_{\omega, r} = E_{\omega,r}+ E_{\omega,r}^\dagger, \quad
\mathcal{Y}_{\omega, r} = \frac{1}{i}(E_{\omega,r}- E_{\omega,r}^\dagger).
\end{align}


Appendix \ref{appendix:invariant_planes} proves that the operators in \eqref{eq:ClosedInvariantOperators} satisfy
\eqref{eq:Conjugations}, that pairs belonging to distinct $(\omega,r)$ are mutually orthogonal, and that $\|\mathcal{X}_{\omega,r}\|^{2}=\|\mathcal{Y}_{\omega,r}\|^{2}
=2^{\,m_{\omega,r}+1}$, as used in Proposition \ref{proposition:multi-qubit_Projected}.

Furthermore, we note that high-frequency invariant-plane operators are not made purely of single-qubit Pauli string. For instance, for a configuration with $S_+=S\subseteq\{1,\dots,N\}$ and $S_-=T=\varnothing$,
the operators $\mathcal{X}_{\omega,r}$ and $\mathcal{Y}_{\omega,r}$ contain mixed multi-qubit Pauli strings supported on $S$ for $j$-th qubit, including terms of the form
$\hat{Y}_j \otimes \hat{Z}_{S\setminus\{j\}},~ j\in S,$
as well as higher-order mixtures with multiple $Y$ factors.
This reflects that higher $\omega$ components are associated with multi-qubit operator support, rather than purely single-qubit Pauli terms.

\subsection{Complete 2-Qubit System}
\label{appendix:proofofentanglementplacement}
Consider an encoder with generator $G=\hat X_1+\hat X_2$ and a $\hat Z_1$ measurement on qubit $1$.
Throughout, for any single-qubit operator $A\in\mathbb{C}^{2\times 2}$ we denote its $N$-qubit embedding by
\begin{equation}
\hat A_n \;\triangleq\; I^{\otimes (n-1)}\otimes A \otimes I^{\otimes (N-n)}.
\end{equation}
In particular, for $N=2$ we have $\hat A_1=A\otimes I$ and $\hat A_2=I\otimes A$.

We adopt the proposed decomposition
\begin{equation}
W_L = W_L^{\rm post}\,U_{\rm ent}\,W_L^{\rm pre},
\end{equation}
where 
$W_L^{\rm post}=W_{L,1}^{\rm post}\otimes W_{L,2}^{\rm post}$ and 
$W_L^{\rm pre}=W_{L,1}^{\rm pre}\otimes W_{L,2}^{\rm pre}$.    
The Ising entangler is
\begin{equation}
U_{\rm ent}=\exp\!\left(-\frac{i}{2}\gamma\,\hat Z_1\hat Z_2\right).
\end{equation}

For $N=2$, the frequency index is determined by $\omega=|S_+|-|S_-|$ for pairwise disjoint sets $S_+,S_-,T\subseteq\{1,2\}$. The highest frequency is $\omega=2$, which occurs only for $(S_+,S_-,T)=(\{1,2\},\varnothing,\varnothing)$,
so $d_2=1$ and $r=1$ in this case.

Using $\mathrm{ad}_G$, an explicit eigen-operator is
\begin{equation}
E_{2,1} \;=\; ( \hat Z_1+i\hat Y_1)( \hat Z_2+i\hat Y_2),
\end{equation}
which yields the invariant operators
\begin{align}
\mathcal{X}_{2,1} &= E_{2,1}+E_{2,1}^\dagger = 2\,(\hat Z_1\hat Z_2-\hat Y_1\hat Y_2),\\
\mathcal{Y}_{2,1} &= \frac{1}{i}(E_{2,1}-E_{2,1}^\dagger) = 2\,(\hat Z_1\hat Y_2+\hat Y_1\hat Z_2).
\end{align}
In particular, the $\omega=2$ plane contains the mixed Pauli strings $\hat Y_1\hat Z_2$ and $\hat Z_1\hat Y_2$.

Let the full circuit be
\begin{equation}
U(x)=(W_{L,1}^{\rm post}\otimes W_{L,2}^{\rm post})\,
U_{\rm ent}\,
(W_{L,1}^{\rm pre}\otimes W_{L,2}^{\rm pre})\,
(R_{X_1}(x)\otimes R_{X_2}(x)),
\end{equation}
and the model output is
\begin{equation}
f(x)=\langle00|U(x)^\dagger \hat Z_1\, U(x)|00\rangle.
\end{equation}

Write the data-encoding as
\begin{equation}
S(x)=e^{-ixG/2},
\end{equation}
and write the $x$-independent trainable part as
\begin{equation}
W_L=(W_{L,1}^{\rm post}\otimes W_{L,2}^{\rm post})\,
U_{\rm ent}\,
(W_{L,1}^{\rm pre}\otimes W_{L,2}^{\rm pre}),
\end{equation}
so that $U(x)=W_LS(x)$.

Consider the conjugated observable
\begin{equation}
M_1:=(W_L^{\rm post})^\dagger \hat Z_1 W_L^{\rm post}.
\end{equation}
Since $W_L^{\rm post}=W_{L,1}^{\rm post}\otimes W_{L,2}^{\rm post}$ and $\hat Z_1=Z\otimes I$,
only $W_{L,1}^{\rm post}$ affects the conjugation.
Parameterize
\begin{equation}
W_{L,1}^{\rm post}=\exp\!\left(-\frac{i}{2}\beta\,\mathbf{n}\cdot\boldsymbol{\sigma}\right),
\qquad
\boldsymbol{\sigma}=(X,Y,Z),\quad \|\mathbf{n}\|=1,
\end{equation}
and define the single-qubit generator $A:= i\frac{\beta}{2}\mathbf{n}\cdot\boldsymbol{\sigma}$.
Embed it to the two-qubit space as $\hat A_1:=A\otimes I$ so that
\begin{equation}
M_1 = e^{\hat A_1}\,\hat Z_1\,e^{-\hat A_1}.
\end{equation}
Applying Hadamard's lemma and evaluating commutators on qubit $1$ yields
\begin{equation}
M_1 = m_x\hat X_1 + m_y\hat Y_1 + m_z\hat Z_1,
\qquad m_x^2+m_y^2+m_z^2=1.
\end{equation}

Now conjugate through the entangler:
\begin{equation}
M_2:=U_{\rm ent}^\dagger M_1 U_{\rm ent}.
\end{equation}
\begin{align}
M_1
&=(W_L^{\rm post})^\dagger \hat Z_1 W_L^{\rm post} \notag\\
&=\big((W_{L,1}^{\rm post})^\dagger\otimes (W_{L,2}^{\rm post})^\dagger\big)(Z\otimes I)\big(W_{L,1}^{\rm post}\otimes W_{L,2}^{\rm post}\big)\notag\\
&=\big((W_{L,1}^{\rm post})^\dagger Z\, W_{L,1}^{\rm post}\big)\otimes \big((W_{L,2}^{\rm post})^\dagger I\, W_{L,2}^{\rm post}\big)\notag\\
&=\big((W_{L,1}^{\rm post})^\dagger Z\, W_{L,1}^{\rm post}\big)\otimes I.
\label{eq:M1_reduce}
\end{align}

Applying Hadamard's lemma gives
\begin{equation}
e^{A}Ze^{-A}=\sum_{k=0}^{\infty}\frac{1}{k!}\,\mathrm{ad}_A^k(Z),
\qquad \mathrm{ad}_A(B)=[A,B].
\end{equation}
Using the Pauli commutator identity
\begin{equation}
[\mathbf u\cdot\boldsymbol{\sigma},\,\mathbf v\cdot\boldsymbol{\sigma}]
=2i(\mathbf u\times \mathbf v)\cdot\boldsymbol{\sigma},
\end{equation}
with $\mathbf e_z=(0,0,1)$ (so $Z=\mathbf e_z\cdot\boldsymbol{\sigma}$), we obtain
\begin{equation}
[A,Z]
= i\frac{\beta}{2}[\mathbf n\cdot\boldsymbol{\sigma},\,\mathbf e_z\cdot\boldsymbol{\sigma}]
= -\beta(\mathbf n\times \mathbf e_z)\cdot\boldsymbol{\sigma}
=\beta(n_x Y-n_y X).
\end{equation}
Here, for any real vectors $\mathbf u=(u_x,u_y,u_z)^\top$ and $\mathbf v=(v_x,v_y,v_z)^\top$ in $\mathbb{R}^3$, we use the shorthand
\[
\mathbf u\cdot\boldsymbol{\sigma} := u_x X + u_y Y + u_z Z,
\qquad
\mathbf v\cdot\boldsymbol{\sigma} := v_x X + v_y Y + v_z Z.
\]
We also write the rotation axis as a unit vector $\mathbf n=(n_x,n_y,n_z)^\top\in\mathbb{R}^3$ with $\|\mathbf n\|=1$, and define $\mathbf e_z:=(0,0,1)^\top$ so that $Z=\mathbf e_z\cdot\boldsymbol{\sigma}$.
The series sums to the closed-form Rodrigues rotation, yielding
\begin{equation}
(W_{L,1}^{\rm post})^\dagger Z\, W_{L,1}^{\rm post}
= m_x X + m_y Y + m_z Z,
\qquad m_x^2+m_y^2+m_z^2=1,
\label{eq:singlequbit_Z_rotation}
\end{equation}
for real coefficients $(m_x,m_y,m_z)$ determined by $(\beta,\mathbf n)$.
Combining \eqref{eq:M1_reduce} and \eqref{eq:singlequbit_Z_rotation} gives
\begin{equation}
M_1
=(m_x X+m_y Y+m_z Z)\otimes I
= m_x\hat X_1+m_y\hat Y_1+m_z\hat Z_1.
\end{equation}
Conjugating through $U_{\mathrm{ent}}=\exp(-\tfrac{i}{2}\gamma\hat Z_1\hat Z_2)$ gives
\begin{align}
U_{\rm ent}^\dagger \hat Z_1 U_{\rm ent} &= \hat Z_1,\\
U_{\rm ent}^\dagger \hat X_1 U_{\rm ent} &= \hat X_1\cos\gamma - \hat Y_1\hat Z_2\sin\gamma,\\
U_{\rm ent}^\dagger \hat Y_1 U_{\rm ent} &= \hat Y_1\cos\gamma + \hat X_1\hat Z_2\sin\gamma,
\end{align}
and therefore
\begin{equation}
M_2
= m_z\hat Z_1
+ m_x(\hat X_1\cos\gamma-\hat Y_1\hat Z_2\sin\gamma)
+ m_y(\hat Y_1\cos\gamma+\hat X_1\hat Z_2\sin\gamma).
\end{equation}

Finally,
\begin{equation}
O:=(W_L^{\rm pre})^\dagger M_2 W_L^{\rm pre},
\qquad
W_L^{\rm pre}=W_{L,1}^{\rm pre}\otimes W_{L,2}^{\rm pre}.
\end{equation}
For example, on the mixed term $\hat Y_1\hat Z_2$,
\begin{equation}
(W_L^{\rm pre})^\dagger(\hat Y_1\hat Z_2)W_L^{\rm pre}
=
\big((W_{L,1}^{\rm pre})^\dagger Y\,W_{L,1}^{\rm pre}\big)\otimes
\big((W_{L,2}^{\rm pre})^\dagger Z\,W_{L,2}^{\rm pre}\big),
\end{equation}
so $W_L^{\rm pre}$ preserves the support $\{1,2\}$ while reweighting and mixing the local Pauli factors.

\subsection{Generalization to $N$-Qubit Systems}
\label{app:Nqubit}

We now carry the same three stages through for arbitrary $N$ and arbitrary
interaction graph $\mathcal{G}=(V,E)$, for a single encoding layer and a single entangling round, and show that the graph then determines a hard ceiling on the frequencies a given readout can reach.

Fix the measured qubit $n$, write $\mathcal{N}(n)=\{k:\{n,k\}\in E\}$ for its
neighbour set and $d_{n}=|\mathcal{N}(n)|$ for its degree. The measurement is
$M=\hat{Z}_{n}$ and the trainable block is
$W_{L}=W_{L}^{\mathrm{post}}U_{\mathrm{ent}}W_{L}^{\mathrm{pre}}$ with
$W_{L}^{\mathrm{post}}=\bigotimes_{j}W_{L,j}^{\mathrm{post}}$ and
$W_{L}^{\mathrm{pre}}=\bigotimes_{j}W_{L,j}^{\mathrm{pre}}$.

\paragraph{Stage 1 (Expose).} Since $W_{L}^{\mathrm{post}}$ is a tensor product
and $\hat{Z}_{n}$ acts only on qubit $n$, only $W_{L,n}^{\mathrm{post}}$ affects
the conjugation. The Rodrigues argument of
Appendix~\ref{appendix:proofofentanglementplacement} applies verbatim on that
factor, giving
\begin{equation}
M_{1}=(W_{L}^{\mathrm{post}})^{\dagger}\hat{Z}_{n}W_{L}^{\mathrm{post}}
=m_{x}\hat{X}_{n}+m_{y}\hat{Y}_{n}+m_{z}\hat{Z}_{n},
\qquad m_{x}^{2}+m_{y}^{2}+m_{z}^{2}=1,
\end{equation}
independently of $N$ and of the graph.

\paragraph{Stage 2 (Spread).} The operators $\hat{Z}_{k}$, $k\in\mathcal{N}(n)$,
commute, so $e^{-i\Phi_{n}}=\prod_{k\in\mathcal{N}(n)}
(\cos\gamma_{nk}\,\hat{I}-i\sin\gamma_{nk}\hat{Z}_{k})$. Expanding the product
over subsets $T\subseteq\mathcal{N}(n)$ and writing
$\hat{Z}_{T}\triangleq\bigotimes_{t\in T}\hat{Z}_{t}$,
\begin{equation}
\label{eq:kappaT}
\kappa_{T}\triangleq\prod_{k\in\mathcal{N}(n)\setminus T}\cos\gamma_{nk}
\prod_{k\in T}\sin\gamma_{nk},
\end{equation}
\begin{equation}
\label{eq:cosSinExpansion}
\cos\Phi_{n}=\!\!\sum_{|T|\,\mathrm{even}}\!\!(-1)^{|T|/2}\kappa_{T}\hat{Z}_{T},
\qquad
\sin\Phi_{n}=\!\!\sum_{|T|\,\mathrm{odd}}\!\!(-1)^{\frac{|T|-1}{2}}\kappa_{T}
\hat{Z}_{T}.
\end{equation}
Substituting into Lemma~\ref{lem:ising},
\begin{align}
\label{eq:M2general}
M_{2}=U_{\mathrm{ent}}^{\dagger}M_{1}U_{\mathrm{ent}}
=m_{z}\hat{Z}_{n}
+\!\!\sum_{|T|\,\mathrm{even}}\!\!(-1)^{\frac{|T|}{2}}\kappa_{T}
\bigl(m_{x}\hat{X}_{n}+m_{y}\hat{Y}_{n}\bigr)\hat{Z}_{T} \notag\\ 
+\!\!\sum_{|T|\,\mathrm{odd}}\!\!(-1)^{\frac{|T|-1}{2}}\kappa_{T}
\bigl(m_{y}\hat{X}_{n}-m_{x}\hat{Y}_{n}\bigr)\hat{Z}_{T}.
\end{align}
Every string in $M_{2}$ is supported on $\{n\}\cup T$ for some
$T\subseteq\mathcal{N}(n)$; the graph fixes which qubits can appear, and the
largest support is $\{n\}\cup\mathcal{N}(n)$.

\paragraph{Stage 3 (Align).} $W_{L}^{\mathrm{pre}}$ is a tensor product, so
conjugation acts qubit-wise and maps each single-qubit factor into a combination
of $\hat{X}$, $\hat{Y}$, $\hat{Z}$ on the same qubit. It reweights and mixes the
Pauli factors but cannot enlarge the support set. The effective observable
$O=W_{L}^{\dagger}\hat{Z}_{n}W_{L}$ is therefore supported on
$\{n\}\cup\mathcal{N}(n)$, of size $d_{n}+1$.

\paragraph{Graph-limited reachability.} This support constraint caps the
reachable spectrum. By Proposition~\ref{prop:planes}, a Pauli string belongs to a
plane with configuration $(S_{+},S_{-},T)$ only if it carries $\hat{Y}$ or
$\hat{Z}$ on $S_{+}\cup S_{-}$ and $\hat{X}$ on $T$, so its frequency obeys
$|\omega|\le|S_{+}|+|S_{-}|$, which is at most the size of its support. Hence
\begin{equation}
\label{eq:reachbound}
\langle\mathcal{X}_{\omega,r},O\rangle=\langle\mathcal{Y}_{\omega,r},O\rangle=0
\qquad\text{for all }\ \omega>d_{n}+1 .
\end{equation}
The bound is attained. The string $\hat{Y}_{n}\hat{Z}_{\mathcal{N}(n)}$ appears in
\eqref{eq:M2general} with coefficient $\pm\mu\kappa_{\mathcal{N}(n)}$, where
$\mu=m_{y}$ for $d_{n}$ even and $\mu=m_{x}$ for $d_{n}$ odd, and carries exactly
one $\hat{Y}$ factor, so it lies in $\mathcal{Y}_{d_{n}+1,r}$ for the plane with
$S_{+}=\{n\}\cup\mathcal{N}(n)$ and $S_{-}=T=\emptyset$. Since all other strings
of $M_{2}$ have smaller support and are orthogonal to it,
\begin{equation}
\label{eq:topoverlap}
\bigl|\langle \mathcal{Y}_{d_n+1,r}, M_2\rangle\bigr| = 2\,|\mu|\prod_{k\in\mathcal{N}(n)}|\sin\gamma_{nk}|.
\end{equation}

Two consequences follow. First, the reachable spectrum is set by the local degree rather than by $N$: measuring a qubit with $d_{n}$ neighbours caps the accessible frequency at $d_{n}+1$ no matter how many qubits the circuit has or how the parameters are trained. Activating frequency $\omega$ therefore requires $\deg(n)\ge\omega-1$ in the interaction graph, a condition a designer can check before training, and Corollary~\ref{cor:activation} shows that it also suffices.  Second, by \eqref{eq:topoverlap} the top-plane overlap scales as $\prod_{k}|\sin\gamma_{nk}|$, so the top frequency becomes unreachable if any
coupling angle is $0$ or $\pi$. Since
$\sum_{T\subseteq\mathcal{N}(n)}\kappa_T^2=\prod_k(\cos^2\gamma_{nk}+\sin^2\gamma_{nk})=1$, the coupling angles distribute a fixed budget of the effective observable across its Pauli supports $\{n\}\cup T$, so no single choice of $\gamma$ maximizes its overlap with every invariant plane at once, which is why we treat the entangling
strengths as trainable parameters.

Setting $d_{n}=1$ recovers Appendix~\ref{appendix:proofofentanglementplacement}:
$\kappa_{\emptyset}=\cos\gamma$, $\kappa_{\{2\}}=\sin\gamma$, the reachable
ceiling is $\omega=2$, and the top-plane overlap is $2|m_x\sin\gamma|$, matching. The star graph of Fig.~\ref{fig:Design_flow} has
$d_{1}=2$ and ceiling $\omega=3$, which is the frequency shown there. This
analysis assumes a single entangling round; with $L$ rounds the support spreads to
the radius-$L$ neighbourhood of $n$ and the ceiling grows accordingly, which we do
not analyze here.

\subsection{Design Procedure}
\label{app:design_alg}
\begin{algorithm}[h]
\caption{Placement rule for a target frequency}
\label{alg:design}
\begin{algorithmic}[1]
\Require interaction graph $\mathcal{G}=(V,E)$ with $|V|=N$, target frequency $\omega^{\star}\le N$, with $\omega^{\star}=N$ if the spectrum is unknown
\State choose the readout qubit $n\in\arg\max_{v\in V}\deg(v)$
\If{$\deg(n)<\omega^{\star}-1$}
\State \Return add entangling rounds or data re-uploading \Comment{one round reaches at most $\deg(n)+1$}
\EndIf
\State $U_{\mathrm{ent}}\gets\prod_{k\in\mathcal{N}(n)}e^{-i\gamma_{nk}\hat{Z}_{n}\hat{Z}_{k}/2}$ \Comment{Spread, $CZ$ at $|\gamma_{nk}|=\pi/2$}
\State $W_{L}\gets W_{L}^{\mathrm{post}}\,U_{\mathrm{ent}}\,W_{L}^{\mathrm{pre}}$ \Comment{Expose, Spread, Align}
\State $W_{R}\gets W_{R}^{\mathrm{post}}\,U_{\mathrm{ent}}\,W_{R}^{\mathrm{pre}}$ \Comment{state side, sets the amplitude}
\State \Return $U(x)=W_{L}\,S(x)\,W_{R}$ with readout $\hat{Z}_{n}$
\end{algorithmic}
\end{algorithm}
The procedure reads each vertex degree once and places $2\deg(n)$ two-qubit gates and four layers of single-qubit rotations, so it runs in $O(|V|+|E|)$ time. The invariant planes enter only its correctness proof, never its execution, so the $4^{N}$-dimensional operator space is never formed. By Corollary~\ref{cor:activation} and Appendix~\ref{app:activation}, every $\omega\le\deg(n)+1$ of the returned circuit has a nonzero coefficient at almost every parameter setting, and each can reach amplitude $1$ on its own. For $\omega^{\star}=N$, a vertex of degree $N-1$ needs $N-1$ edges, so the star centered on the measured qubit is the smallest graph that reaches every frequency of one encoding layer in one round, and it is the circuit of our experiments. The star is native to platforms with all-to-all connectivity, such as trapped ions with M{\o}lmer--S{\o}rensen gates, while on bounded-degree devices, such as heavy-hex processors with degree at most three, one round reaches $\omega\le4$ and larger targets need further rounds or re-uploading.

\subsection{Scope of the Results}
\label{app:general_ent}
\textbf{Any encoding generator.} Let $G=\sum_{k}\lambda_{k}\Pi_{k}$ be the spectral decomposition of a Hermitian generator. Every operator splits as $A=\sum_{j,k}\Pi_{j}A\Pi_{k}$ with $\mathrm{ad}_{G}(\Pi_{j}A\Pi_{k})=(\lambda_{j}-\lambda_{k})\Pi_{j}A\Pi_{k}$, and $\mathrm{ad}_{G}$ is self-adjoint for the Hilbert--Schmidt inner product, so $\mathcal{B}(\mathcal{H})$ is the orthogonal sum of the eigenspaces $\mathcal{B}_{\omega}=\{A:\mathrm{ad}_{G}(A)=-2\omega A\}$, on which $S(x)^{\dagger}AS(x)=e^{-i\omega x}A$. For an orthogonal basis $\{E_{\omega,r}\}_{r}$ of $\mathcal{B}_{\omega}$ with $\omega>0$, the operators $\mathcal{X}_{\omega,r}=E_{\omega,r}+E_{\omega,r}^{\dagger}$ and $\mathcal{Y}_{\omega,r}=-i(E_{\omega,r}-E_{\omega,r}^{\dagger})$ are closed pairs in the sense of Proposition~\ref{prop:2Dclosure}, mutually orthogonal and of equal norm, and the proof of Proposition~\ref{proposition:multi-qubit_Projected} applies verbatim with $2^{N-m_{\omega,r}-1}$ replaced by $2^{N}/\|\mathcal{X}_{\omega,r}\|^{2}$. The frequencies are the half-differences of the eigenvalues of $G$ \citep{Schuld}, and Proposition~\ref{prop:planes} is the explicit basis for $G=\sum_{n}\hat{X}_{n}$.

\textbf{Pauli encodings and readouts.} If qubit $j$ is encoded along a Pauli axis $\sigma_{j}$ and the readout is a Pauli operator on qubit $n$, single-qubit Clifford gates map $\sum_{j}\sigma_{j}$ to $\sum_{j}\hat{X}_{j}$ and the readout to $\hat{Z}_{n}$. These gates merge into the adjacent rotation layers, which are general, so every result holds with the Pauli labels exchanged.

\textbf{Entanglers.} Let $P_{n}$ be a Pauli operator on the measured qubit and $Q_{k}$ a Pauli operator on $k\in\mathcal{N}(n)$. Single-qubit Clifford gates that map $P_{n}$ to $\hat{Z}_{n}$ and each $Q_{k}$ to $\hat{Z}_{k}$ turn $\exp\bigl(-\tfrac{i}{2}\sum_{k}\gamma_{nk}P_{n}Q_{k}\bigr)$ into $U_{\mathrm{ent}}$ of Lemma~\ref{lem:ising}, with signs absorbed into $\gamma_{nk}$, and they merge into the adjacent rotation layers. Lemma~\ref{lem:ising}, Corollary~\ref{cor:activation}, the amplitude results of Appendix~\ref{app:activation} and Algorithm~\ref{alg:design} therefore hold unchanged for every entangler of commuting two-qubit couplings that share their factor $P_{n}$. This class contains $CZ$ and controlled-phase gates, $CNOT$ gates with the measured qubit as common control or common target, whose single-qubit factors commute with the couplings, echoed cross-resonance gates with a common control, and M{\o}lmer--S{\o}rensen $\hat{X}\hat{X}$ gates. For any other layer of two-qubit gates on the edges at $n$, $O$ remains supported on $\{n\}\cup\mathcal{N}(n)$, so the ceiling $\deg(n)+1$ still holds, and part (i) of the corollary applies whenever $O$ reaches the top plane at one parameter setting. This can fail. In our numerical checks, a star of iSWAP gates, whose $\hat{X}\hat{X}$ and $\hat{Y}\hat{Y}$ terms do not share a factor on $n$, leaves $c_{\deg(n)+1}\equiv0$, whereas a star of $\sqrt{\mathrm{iSWAP}}$ gates activates it, so the class above is sufficient but not necessary.

\newpage
\section{Data Reuploading}
\label{appendix:DataReupload}
\textbf{Data Reuploading.} The previous analysis considered a single data-encoding block $S(x)$. A direct
extension is to encode the input multiple times throughout the circuit:
\begin{equation}
    U_{\rm DR}(x,\Theta)
    =
    W_LS(x)W_{L-1}S(x)\cdots W_1S(x)W_0 .
\end{equation}
This is known as data reuploading \citep{perez2020data}. From the Fourier viewpoint, each encoding block introduces another opportunity to generate input-dependent rotations in the operator space. Therefore, repeating the encoding increases the range of accessible Fourier components. In particular, if one encoding layer with $N$ qubits can generate frequencies up to $\omega \leq N$, then $L$ reuploading layers can, in principle, generate higher frequencies up to $\omega \leq LN$. Thus, data reuploading provides a practical way to increase spectral expressivity without increasing the number of qubits. Access to higher frequencies does not mean that the model uses them. Our results are proved for a single encoding layer, and we do not extend Proposition~\ref{proposition:multi-qubit_Projected}, Corollary~\ref{cor:activation} or Algorithm~\ref{alg:design} to $L>1$. Our re-uploading experiments are therefore empirical and test whether the placement remains useful beyond the proved setting.

\section{Limitations and Potentials}
\label{appendix:limitation}
Our placement analysis covers one encoding layer and one entangling round on each side of the encoder, with single-qubit Pauli encodings and readouts and with entanglers whose two-qubit couplings commute and share a factor on the measured qubit (Appendix~\ref{app:general_ent}), and we evaluate data re-uploading only empirically. One round reaches frequencies up to $\deg(n)+1$, so on devices of bounded connectivity, such as heavy-hex processors of degree three, larger targets need further entangling rounds, which our analysis does not cover, or re-uploading. High-frequency planes are spanned by Pauli strings that act on many qubits, and local noise such as depolarization damps a string exponentially in the number of qubits it acts on, so the frequencies that our rule activates are also the most exposed to noise. At 14 qubits, finite-shot readout needs more than $10^{4}$ shots per point before any circuit we tested beats the constant predictor (Table~\ref{tab:noise14}). We evaluate all circuits in classical simulation and leave experiments on quantum hardware to future work. In broader impact, an interpretable design rule can help practitioners build PQCs that are both expressive and controllable, and it can reduce unnecessary circuit depth and two-qubit gate counts.

\end{document}

%% file: math_commands.tex
\usepackage{amsmath,amsfonts,bm}

\def\eqref#1{equation~\ref{#1}}

\def\1{\bm{1}}

\DeclareMathAlphabet{\mathsfit}{\encodingdefault}{\sfdefault}{m}{sl}
\SetMathAlphabet{\mathsfit}{bold}{\encodingdefault}{\sfdefault}{bx}{n}

